\documentclass[aps,prx,onecolumn,nofootinbib,notitlepage, 10pt, longbibliography]{revtex4-1}

\usepackage{amsmath, amssymb, amsthm, amsfonts, graphicx}
\usepackage{hyperref}
\hypersetup{
	bookmarksnumbered=true, 
	unicode=false, 
	pdfstartview={FitH}, 
	pdftitle={}, 
	pdfauthor={}, 
	pdfsubject={}, 
	pdfcreator={}, 
	pdfproducer={}, 
	pdfkeywords={}, 
	pdfnewwindow=true, 
	colorlinks=true, 
	linkcolor=blue, 
	citecolor=blue, 
	filecolor=blue, 
	urlcolor=blue 
}
\usepackage{complexity}
\usepackage{thmtools}
\usepackage[capitalize,nameinlink]{cleveref} 
\usepackage{verbatim}
\usepackage{qcircuit}
\newcommand{\ket}[1]{| #1\rangle}        
\newcommand{\bra}[1]{\langle #1|}        
\newcommand{\braket}[2]{\langle #1 | #2 \rangle} 
\newcommand{\ketbra}[2]{| #1 \rangle\!\langle #2 |} 

\newcommand{\ii}{\mathbb{I}}		

\newtheorem{theorem}{Theorem}

\newtheorem{corollary}[theorem]{Corollary}
\newtheorem{lemma}[theorem]{Lemma}

\newcommand{\eq}[1]{Eq.~\hyperref[eq:#1]{(\ref*{eq:#1})}}
\renewcommand{\sec}[1]{\hyperref[sec:#1]{Section~\ref*{sec:#1}}}
\newcommand{\app}[1]{\hyperref[app:#1]{Appendix~\ref*{app:#1}}}
\newcommand{\tab}[1]{\hyperref[tab:#1]{Table~\ref*{tab:#1}}}
\newcommand{\fig}[1]{\hyperref[fig:#1]{Figure~\ref*{fig:#1}}}
\newcommand{\figa}[2]{\hyperref[fig:#1]{Figure~\ref*{fig:#1}#2}}
\newcommand{\figx}[2]{\hyperref[fig:#1]{Figure~\ref*{fig:#1}(#2)}}
\newcommand{\thm}[1]{\hyperref[thm:#1]{Theorem~\ref*{thm:#1}}}
\newcommand{\lem}[1]{\hyperref[lem:#1]{Lemma~\ref*{lem:#1}}}
\newcommand{\cor}[1]{\hyperref[cor:#1]{Corollary~\ref*{cor:#1}}}
\newcommand{\defn}[1]{\hyperref[def:#1]{Definition~\ref*{def:#1}}}
\newcommand{\alg}[1]{\hyperref[alg:#1]{Algorithm~\ref*{alg:#1}}}
\newcommand{\prob}[1]{\hyperref[prob:#1]{Problem~\ref*{prob:#1}}}
\newcommand{\threatM}[1]{\hyperref[threat:#1]{Threat Model \ref*{threat:#1}}}

\usepackage{tikz}
\usepackage{mathtools}
\usepackage{amsmath, amssymb, amsthm, amsfonts, graphicx}
\usepackage{verbatim}
\usepackage{pgfmath}
\usetikzlibrary{backgrounds,fit,decorations.pathreplacing,calc}
\newcommand{\ketb}[2]{\ensuremath{\left|#1\right\rangle_{#2}}}

\makeatletter
\tikzset{
	adjust fit placement/.style={
		every fit/.style={
			text height/.expanded=\the\pgf@y+.5\ht\pgfnodeparttextbox -.5\dp\pgfnodeparttextbox,
		}
	}
}
\makeatother
\tikzset{adjust fit placement}

\tikzstyle{operator} = [draw,fill=white,minimum size=1.5em,thin] 
\tikzstyle{joint} = [draw,fill=white,rounded corners=1.0mm, minimum size=1.25em,thin] 
\tikzstyle{ctrl} = [draw,fill,shape=circle,minimum size=5pt,inner sep=0pt]
\tikzstyle{oplus} = [inner sep=-1pt]
\tikzstyle{vdots} = [above=-5mm]
\tikzstyle{bigBoxSpace} = [text width = 35pt]
\tikzstyle{doubleLine} = [fill=transparent, double distance=14pt] 
\tikzstyle{doubleLineNarrow}  = [fill=transparent, double distance=5pt] 
\tikzstyle{operatorDoubleLine} = [operator, minimum size=20pt] 
\tikzstyle{ctrlDoubleLine} = [ctrl, above = 4.5pt]
\tikzstyle{ctrlDoubleLineDown} = [ctrl, below = 4.5pt]
\tikzstyle{doubleLineVdots} = [above=-8.5pt]
\tikzstyle{operatorDoubleRegister} = [operator, text width = 30pt] 
\tikzstyle{operatorDoubleRegisterWidth} = [text width = 30pt]

\begin{document}
\title{Hamiltonian Simulation in the Interaction Picture}
\author{Guang Hao Low and Nathan Wiebe}
\affiliation{Quantum Architectures and Computation, Microsoft Research, Redmond, Washington, USA}
\begin{abstract}
	We present a low-space overhead simulation algorithm based on the truncated Dyson series for time-dependent quantum dynamics. This  algorithm is applied to simulating time-independent Hamiltonians by transitioning to the interaction picture, where some portions are made time-dependent. This can provide a favorable complexity trade-off as the algorithm scales exponentially better with derivatives of the time-dependent component than the original Hamiltonian. We show that this leads to an exponential improvement in gate complexity for simulating some classes of diagonally dominant Hamiltonian.  Additionally we show that this can reduce the gate-complexity scaling for simulating $N$-site Hubbard models for time $t$ with arbitrary long-range interactions as well as reduce the cost of quantum chemistry simulations within a similar-sized plane-wave basis to $\widetilde{\mathcal{O}}(N^2t)$ from $\widetilde{\mathcal{O}}(N^{11/3}t)$. We also show a quadratic improvement in query complexity for simulating sparse time-dependent Hamiltonians, which may be of independent interest.
\end{abstract}
\date{\today}
\maketitle

\section{Introduction}
Simulating quantum dynamics has become in recent years an an increasingly sophisticated field whose growth has been buoyed up by a host of recent successes in both general purpose simulation methods~\cite{Aharonov2003Adiabatic,Berry2007Efficient,Wiebe2010,Childs2012,Berry2012,Berry2014Exponential,Berry2015Hamiltonian,berry2015simulating,Novo2016corrected,Low2016qubitization,Low2016HamSim,Low2017USA} as well as in chemistry and material simulation~\cite{whitfield2011simulation,babbush2017low,bauer2016hybrid}. The majority of the advances that we have seen in the field have not come from identifying ways to exploit the structure of the Hamiltonian; rather, most have arisen from either better analysis of the simulation methods or from designing more efficient ways to implement the propagators.  In this work, we show a method that can explicitly take advantage of structures within the Hamiltonian to further reduce the complexity of simulations.

The central intuition behind our work stems from the interaction picture.  Quantum computation is often discussed in the Schr\"odinger picture wherein the time dynamics of the quantum state is given by

\begin{equation}
\partial_t \ket{\psi(t)} =-iH(t) \ket{\psi(t)},
\end{equation}
for the general case where the Hamiltonian $H(t)$ is time-dependent.  The quantum state carries the entirety of the dynamics here.  Alternatively, one can work in the Heisenberg picture where the time evolution is absorbed into the operators that are being measured.

  The interaction picture can be viewed as a compromise between the Heisenberg and the Schr\"odinger pictures.  In the interaction picture, some of the dynamics are carried by the operators and others by the state.  This view allows one to focus on effects of the interaction, and is particularly fruitful in manual calculations when interactions are perturbative corrections to the free-theory. For instance, if the Hamiltonian is $H=A+B$ then an analytic evaluation of the time-ordered propagator of the interaction picture Hamiltonian $H_I(t)=e^{i A t}Be^{-i A t}$ is possible by perturbative expansions based on Green's functions and Feynman diagrams.  Without assuming any kind of perturbative limit, this division in our case is precisely what allows us to gain an advantage for certain quantum simulation problems.

Simulation of a Hamiltonian that is time-independent in the Schr\"{o}dinger picture can be much more challenging to simulate in the interaction picture. There, evolution by a time-independent Hamiltonian $H$ is transformed in the rotating frame $\ket{\psi_I(t)}=e^{i A t}\ket{\psi(t)}$ to evolution by a time-dependent Hamiltonian $H_I(t)=e^{i A t}Be^{-i A t}$. This follows from an elementary manipulation of the Schr{\"o}dinger equation:
\begin{align}
\label{eq:schrodinger_equation}
i\partial_{t}\ket{\psi(t)}=(A+B)\ket{\psi(t)}
\longrightarrow
i\partial_{t}\ket{\psi_I(t)}=e^{i A t}Be^{-i A t}\ket{\psi_I(t)}.
\end{align}
Implementing the time-ordered propagator $\mathcal{T}[\exp{(-i\int_0^t H(s)\mathrm{d}s)}]$ that solves~\cref{eq:schrodinger_equation} on a quantum computer requires time-dependent simulation algorithms. These are generally more complicated than time-independent algorithms, and exhibit different cost trade-offs that do not appear favorable. For instance, an order-$k$ time-dependent Trotter-Suzuki product formula~\cite{Wiebe2010} has cost that scales with the rate of change of $H(t)$ like $\mathcal{O}(e^{\mathcal{O}(k)}(t\Lambda)^{1+1/(2k)})$\footnote{The standard big-$\mathcal{O}$ notation defines $f(n)\in{\mathcal{O}}(g(n))$ for positive functions $f(n), g(n)>0$ as the existence of absolute constants $a>0,b>0$ such that for any $n>a$, $f(n)\le b g(n)$. We also use $f(n)=\tilde{\mathcal{O}}(g(n))$ when $f(n)\le b g(n)\operatorname{polylog}(g(n))$, and $f(n)=\Omega(g(n))$ when $f(n)\ge b g(n)$, and $f(n)=\Theta(g(n))$ when both $f(n)\in{\mathcal{O}}(g(n))$ and $f(n)=\Omega(g(n))$ are true.}, where 
$\Lambda=\max_{s}\|\dot{H}(s)\|^{1/2}\in{\mathcal{O}}(\|[A,B]\|^{1/2})$. More advanced techniques based on compressed fractional queries~\cite{Berry2014Exponential} appear to scale better like $\sim t\|B\|\frac{\log{(\Lambda t/\epsilon)}}{\log\log{(\Lambda t/\epsilon)}}$ but in terms of queries to a unitary oracle that obscures the gate complexity as it expresses Hamiltonian matrix elements at different times in binary, and may be difficult to implement in practice. One proposed technique~\cite{berry2015simulating} directly implements a truncated Dyson series of a the time-ordered propagator and argues, though without proof, a similar scaling in terms of queries to a different type of oracle.

We show that simulation in the interaction picture can substantially improve the efficiency of time-independent simulation. In~\cref{sec:truncated_dyson_series}, we complete the general time-dependent simulation algorithm by a truncated Dyson series proposed by~\cite{berry2015simulating} by providing a rigorous analysis of the approximation and explicit circuit constructions, with improvements in gate and space complexity over previously expected costs. In~\cref{sec:interaction_picture}, we identify situations where the gate complexity of implementing these queries scale with the interaction strength $\mathcal{O}(\|B\|)$, and not the larger uninteresting component $\mathcal{O}(\|A\|)$. Such are the cases where simulation in the interaction picture is advantageous. In~\cref{sec:Hubbard}, we showcase the potential of interaction-picture simulation by an electronic structure application in the plane-wave basis. We rigorously bound the cost of simulating the time-evolution of $N$ spin-orbitals subject to long-range electron-electron interactions to $\tilde{\mathcal{O}}(N^2 t)$ gates, which is close to a quadratic improvement over prior art of $\tilde{\mathcal{O}}(N^{11/3} t)$~\cite{babbush2017low}. In~\cref{sec:sparse_Hamiltonian_simulation}, we present a complexity theoretic perspective of our work by considering the abstract problem of simulating time-dependent sparse Hamiltonians in the standard query model. We obtain a quadratic improvement in sparsity scaling~\cite{Berry2014Exponential}, and find optimized algorithms for simulating diagonally dominant Hamiltonians.

\section{Outline of paper} A detailed summary of main results in each section follows.  
\\

\cref{sec:truncated_dyson_series} -- {\bf Time-dependent Hamiltonian simulation by a truncated Dyson series}
\\
We present our main algorithmic contribution: a general time-dependent simulation algorithm, described in~\cref{Thm:Compressed_TDS} with a rigorous analysis of its performance and explicit circuit constructions, that is based on synthesizing an approximate Dyson series for general time-dependent Hamiltonians $H(t)$ characterized by spectral-norm $\alpha \ge \max_t \|H(t)\|$ and average rate-of-change $\langle\|\dot{H}\|\rangle$. Bounds on the approximation error of truncating and discretizing the Dyson series are proven in~\cref{sec:proof_truncation_discretization_dyson}, which is used to obtain the cost of simulating the time-ordered evolution operator. In~\cref{sec:compresson_gadget}, this cost is determined to be $\mathcal{O}\big(\alpha t\frac{\log{(\alpha t/\epsilon)}}{\log\log{(\alpha t /\epsilon)}}\big)$ queries. Compared to the original proposal by~\cite{berry2015simulating}, worked out in~\cref{sec:dyson_series_alg_duplicated_registers}, our approach has a gate complexity that scales with $\mathcal{O}(\log{(\langle\|\dot H\|\rangle)})$,  instead of the worst case $\mathcal{O}\big(\log{(\max_{t}\|\dot H(t)\|)}\big)$. The qubit overhead is also reduced by a multiplicative factor of $\mathcal{O}(\log{\big(\frac{t}{\epsilon}\langle\|\dot H\|\rangle\big)})$. The trick we use is of independent interest as it also reduces the space overhead of the time-independent truncated Taylor series algorithm~\cite{berry2015simulating}, discussed in~\cref{sec:truncated_taylor_algorithm} for completeness.  
\\

\cref{sec:interaction_picture} -- {\bf Interaction picture simulation}
\\
 We apply this truncated Dyson series algorithm to simulate time-evolution by a time-independent Hamiltonian $H=A+B$ in the interaction picture. In~\cref{sec:implementation_HAM_T}, we evaluate the gate complexity of constructing the query $\operatorname{HAM-T}$ for an interaction picture Hamiltonian $H_I(t)=e^{i A t}Be^{-i A t}$. This leads to a simulation, described in~\cref{thm:int_pic_sim}, of $e^{-i(A+B)t}$ using 
\begin{align}
\label{eq:int_pic_sim_cost_intro}
\mathcal{O}\left(\alpha_B t \operatorname{polylog}((\alpha_A+\alpha_B)t/\epsilon)\right)
\end{align}
queries to a unitary oracle $O_B$ such that $(\bra{0}_a\otimes \openone_s) O_B (\ket{0}_a\otimes \openone_s) = \frac{B}{\alpha_B}$, queries to unitary time-evolution $e^{iA\tau}$ by $A$ alone for time $\tau=\mathcal{O}(\alpha_B^{-1})$, and additional primitive quantum gates. The parameter $\alpha_A\ge \|A\|$ is also any upper bound on the spectral norm of $A$. This may be compared to state-of-art Schr\"odinger picture simulation algorithms for time-independent Hamiltonians, which require
\begin{align}
\mathcal{O}\left((\alpha_B+\alpha_A) t \operatorname{polylog}((\alpha_A+\alpha_B)t/\epsilon)\right)
\end{align}
queries to $O_B$, queries an analogous oracle for $O_A$, and additional primitive quantum gates. Our result~\cref{eq:int_pic_sim_cost_intro} is then advantageous in cases where roughly $\|A\| \gg \|B\|$ and the gate complexity of $e^{iA\tau}$ is of the same order as $O_B$. In other words, the dominant scaling in gate complexity is the interaction strength $\|B\|$, and not the larger uninteresting component $\|A\|$.
\\

\cref{sec:Hubbard} -- {\bf Application to the Hubbard model with long-ranged interactions} 
\\
 We demonstrate the advantage of Hamiltonian simulation in the interaction picture over time-independent simulation algorithms with the example of a general second-quantized Hubbard model on $N$ lattice sites in an arbitrary number of dimensions.

The model we consider allows for arbitrary single-site disorder, in addition to arbitrary periodic translationally invariant kinetic hopping terms and long-ranged density-density interactions. Provided that the energy of the kinetic term is extensive, our interaction-picture algorithm has gate complexity $\tilde{\mathcal{O}}\left(N^2 t\right)$. Most remarkably, the potential energy only needs to be polynomial in $N$, which is an extremely lax constraint. In particular, this model generalizes electronic structure simulations in the plane-wave basis~\cite{babbush2017low} (which has potential energy $\mathcal{O}(N^2)$), considered in~\cref{sec:Chemistry}. In this case, our result achieves almost a quadratic improvement over the prior art of $\tilde{\mathcal{O}}\left(N^{11/3} t\right)$ gates. This is complementary to recent work by~\cite{Haah2018quantum} which achieves $\tilde{\mathcal{O}}(N t)$ scaling, but under the much stronger assumption of short-range exponentially decaying interactions.
\\

\cref{sec:sparse_Hamiltonian_simulation} -- {\bf Application to sparse Hamiltonian simulation}
\\
We consider a complexity-theoretic generalization of the technology we develop for time-dependent simulation and simulation in the interaction picture. This is through the standard query model for black-box $d$-sparse Hamiltonian simulation, which assumes access to a unitary oracle that provides the positions and values of non-zero entries of the Hamiltonian. Each each row has at most $d$ non-zero entries, and the maximum absolute value of any entry is $\|H\|_{\rm max}$. This information is also provided as a function of a time index for time-dependent Hamiltonians. In~\cref{sec:spase_ham_sim_time_dependent}, we consider this time-dependent case and describe in~\cref{Cor:sparse_time_dependent_simulation} how the time-ordered evolution operator may be simulated using 
$
\mathcal{O}\big(td\|H\|_{\rm max}\frac{\log{(td\|H\|_{\rm max}/\epsilon)}}{\log\log{(td\|H\|_{\rm max}/\epsilon)}}\big)
$
queries. Though linear scaling with respect to $d$ is well-known in the time-independent case~\cite{Berry2012}, this is a quadratic improvement in sparsity scaling over prior art for the time-dependent case~\cite{Berry2014Exponential}. An analogous treatment for simulating sparse time-independent Hamiltonian in the interaction picture in~\cref{sec:sparse_diagonally_dominant},~\cref{thm:sparse_diagonal_dominant} has an identical query complexity, except that $\|H\|_{\rm max}$ is replaced by the maximum absolute value of any off-diagonal entry. This improvement is particularly advantageous for the simulation of diagonally dominant Hamiltonians which arise in many physical systems expressed within an appropriate basis.

\section{Time-dependent Hamiltonian simulation by a truncated Dyson series}
\label{sec:truncated_dyson_series}
In the Schr\"{o}dinger picture, the dynamics of a quantum state is given by $i\partial_{t}\ket{\psi(t)}=H(t)\ket{\psi(t)}$ --  given the initial state $\ket{\psi(0)}$ at time $t=0$, the time-evolved state is $\ket{\psi(t)}=U(t)\ket{\psi(0)}$.
If $H(t)$ is time-independent $U(t)$ can always be written as $e^{-iHt}$.  In the time-dependent case the time evolution operator no longer is $e^{-iHt}$ and indeed it does not in general have a closed-form expression.  The following notation is customarily used to represent the time-evolution operator, $U(t): \ket{\psi(0)} \mapsto \ket{\psi(t)}$, in the case where $H:\mathbb{R} \mapsto \mathbb{C}^{N\times N}$ is a piecewise continuous function:
\begin{equation}
U(t) = \lim_{r\rightarrow \infty} \prod_{j=1}^{r} e^{-iH(t(j-1)/r)t/r}:= \mathcal{T} e^{-i \int_{0}^t H(s) \mathrm{d}s},
\end{equation}
where $\mathcal{T}$ is known as the time-ordering operator.

The fact that time-dependent dynamics lacks a closed form makes simulating its dynamics slightly more challenging than the time-independent case.  This arises because approximations, such as Taylor series,
fail to give a simple series expansion for $U(t)$ unless $[H(t),H(t')]= 0$.  Fortunately, there exists a more general expansion known as the Dyson series that fills the exact same role that the Taylor series fills for the time-independent case.
 For any $t>0$ and bounded $\|H(t)\|$, the Dyson series gives the following absolutely convergent expansion for $U(t)$
\begin{align}
\label{eq:dyson_series}
U(t)
=
\openone
-i\int^t_0 H(t_1)\mathrm{d} t_1
-\int^{t}_{t_2}\int^{t_2}_{0}H(t_2) H(t_1)\mathrm{d} t_1\mathrm{d} t_2
+i\int^{t}_{t_3}\int^{t_3}_{t_2}\int^{t_2}_{0}H(t_3)H(t_2) H(t_1)\mathrm{d} t_1\mathrm{d} t_2\mathrm{d} t_3
+\cdots.
\end{align}
This may be compactly represented using the time-ordering operator $\mathcal{T}$ which sorts any sequence of $k$ operators according to the times $t_j$ of their evaluation, that is, $\mathcal{T}\left[H(t_k)\cdots H(t_2)H(t_1)\right]=H(t_{\sigma(k)})\cdots H(t_{\sigma(2)})H(t_{\sigma(1)})$, where $\sigma$ is a permutation such that $t_{\sigma(1)}\le t_{\sigma(2)}\le \cdots \le t_{\sigma(k)}$. For instance, $\mathcal{T}\left[H(t_2)H(t_1)\right]=\theta(t_2-t_1)H(t_2)H(t_1)+\theta(t_1-t_2)H(t_1)H(t_2)$ using the Heaviside step function $\theta$. With this notation, the propagator is formally expressed as a time-ordered evolution operator $U(t)=\mathcal{T}\left[e^{-i\int_0^t H(s) \mathrm{d}s}\right]$ defined as
\begin{align}
\label{eq:dyson_series_time_ordered}
\mathcal{T}\left[e^{-i\int_0^t H(s) \mathrm{d}s}\right]
=
\sum^\infty_{k=0} (-i)^k D_k,
\quad
D_k=\frac{1}{k!}\int_0^t\cdots\int_0^t \mathcal{T}\left[H(t_k)\cdots H(t_1)\right]\mathrm{d}^k t.
\end{align}

The aim of our work is to approximate $U(t)$ within error $\epsilon$ (as measured by the spectral-norm $\|\cdot\|$ of the difference between the approximation and the true dynamics) for any $t\ge 0$.  We achieve this by constructing the Dyson series expansion $U(t) \approx \sum_{k=0}^K(-i)^k D_k$ and truncate it at finite order $K$ to control the error.  This idea was suggested earlier by~\cite{berry2015simulating}, though without proof. 

\subsection{Input model}
The cost of our algorithm is expressed in terms of unitary oracles $\operatorname{HAM}$ and $\operatorname{HAM-T}$ that encode Hamiltonians in a so-called standard-form~\cite{Low2016qubitization,Low2017USA}. When $H$ is time-independent, we assume access to the following oracle.
\begin{restatable}[Time-independent matrix encoding]{definition}{Blockencoding}
	\label{eq:standard-form-TI}
	Given a matrix ${H} \in \mathbb {C}^{2^{n_s}\times 2^{n_s}}$, and a promise $\|H\|\le\alpha$ assume there exists a unitary oracle $\operatorname{HAM}\in \mathbb {C}^{2^{n_a+n_s}\times 2^{n_a+n_s}}$ such that
	\begin{align}
	\operatorname{HAM} = 
\left(\begin{matrix}
	H/\alpha & \cdot \\
	\cdot & \cdot
\end{matrix}\right)
\quad
\Rightarrow \quad(\bra{0}_a\otimes \ii_s)\operatorname{HAM}(\ket{0}_a\otimes \ii_s) = \frac{H}{\alpha}.
\end{align}
\end{restatable}
Use of this is justified as it generalizes a variety of different input models~\cite{Low2016qubitization}. As an example, $H=\sum^{2^{n_a}}_{j=1}a_j U_j$~\cite{Childs2012} could be a linear combination of $l=2^{n_a}$ unitaries. Then the circuit depicted in~\cref{fig:standard-form} implements $\operatorname{HAM}$ with normalization constant $\alpha=\sum^{l}_{j=1}a_j$ using the unitary oracles
\begin{align}
\operatorname{HAM}=(\operatorname{PREP}^\dag \otimes \openone_s)\cdot\operatorname{SEL}\cdot(\operatorname{PREP}\otimes \openone_s),
\quad
\operatorname{PREP}\ket{0}_a=\sum^{l}_{j=1}\sqrt{\frac{a_j}{\alpha}}\ket{j}_a,
\quad
\operatorname{SEL}=\sum_{j=1}^{l}\ket{j}\bra{j}_a\otimes U_j.
\label{eq:LCU}
\end{align}
These unitaries each cost $\mathcal{O}(l)$ gates -- $\operatorname{PREP}$ is implemented by an a algorithm for preparing arbitrary  $l$-dimensional quantum states~\cite{shende2006synthesis}, and $\operatorname{SEL}$ is implemented by binary-tree control logic~\cite{Childs2017Speedup}. 
\begin{figure*}[t]
		\begin{tabular}{ccc}
	    \begin{tikzpicture}[thick]
	    %
	    \matrix[row sep=0.2cm, column sep=0.2cm] (circuit) {
	    	\node (q1) {\ketb{0}{a}}; &  \node {\textbackslash}; &
	    	\node[joint] (H11) {}; & 
	    	\coordinate (end1); \\
	    	\node (q2) {\ketb{\psi}{s}}; & \node {\textbackslash}; &
	    	\node[operator] (H21) {$\operatorname{HAM}$}; &
	    	\coordinate (end2);\\
	    };
	    \begin{pgfonlayer}{background}
	    \draw[thick] (H11) -- (H21);
	    \draw[thick] (q1) -- (end1) (q2) -- (end2); 
	    \end{pgfonlayer}
	    \end{tikzpicture}
	    &
	    \begin{tikzpicture}[thick]
	    \matrix[row sep=0.2cm, column sep=0.2cm] (circuit) {
	    	\node (q0) {\ketb{m}{d}}; &\node {\textbackslash}; &
	    	\node[joint] (H01) {T}; & 
	    	\coordinate (end0); \\
	    	\node (q1) {\ketb{0}{a}}; &\node {\textbackslash}; &
	    	\node[joint] (H11) {}; &
	    	\coordinate (end1); \\
	    	\node (q2) {\ketb{\psi}{s}}; &\node {\textbackslash}; &
	    	\node[operator] (H21) {$\operatorname{HAM-T}$}; &
	    	\coordinate (end2);\\
	    };
	    \begin{pgfonlayer}{background}
	    \draw[thick] (H01) -- (H21);
	    \draw[thick] (q0) -- (end0) (q1) -- (end1) (q2) -- (end2); 
	    \end{pgfonlayer}
	    \end{tikzpicture}
	    
	    &
	    \begin{tikzpicture}[thick]
	    \matrix[row sep=0.2cm, column sep=0.2cm] (circuit) {
	    	\node (q1) {\ketb{0}{a}};&\node {\textbackslash}; & \node[operator] (123) {PREP}; & \node[joint] (H11) {}; & \node[operator] (123) {$\operatorname{PREP}^\dag$}; &
	    	\coordinate (end1); \\
	    	\node (q2) {\ketb{\psi}{s}}; &\node {\textbackslash}; & &
	    	\node[operator] (H21) {$\operatorname{SEL}$}; & &
	    	\coordinate (end2);\\
	    };
	    \begin{pgfonlayer}{background}
	    \draw[thick] (H11) -- (H21);
	    \draw[thick]  (q1) -- (end1) (q2) -- (end2); 
	    \end{pgfonlayer}
	    \end{tikzpicture}
 \end{tabular}
	\caption{\label{fig:standard-form}Quantum circuit representation of (left) an oracle $\operatorname{HAM}$ from~\cref{eq:standard-form-TI} encoding a time-independent Hamiltonian, (center) an oracle $\operatorname{HAM-T}$ from~\cref{def:HAM-T} encoding a time-dependent Hamiltonian, and (right) an example implementation of $\operatorname{HAM}$ from with a linear-combination of unitaries from~\cref{eq:LCU}. Bold horizontal lines with a backslash `$\backslash$' depict registers that in general comprise of multiple qubits. Vertical lines connecting boxes depict unitaries that act jointly on all registers covered by the boxes. A small square box marked by `$\operatorname{T}$' indicates control by a time index.}
\end{figure*}
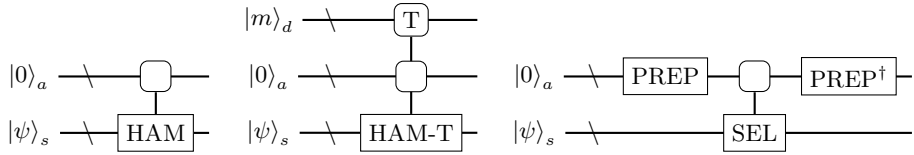

A direct time-dependent generalization of~\cref{eq:standard-form-TI} is the unitary oracle $\operatorname{HAM-T}$ that encodes the Hamiltonian $H(s)$ defined over time $s\in[0,t]$, with $t>0$. The continuous parameter $s$ is discretized into an integer number of $M>0$ time bins which we index by $m\in\{0,1,\cdots, M-1\}$. 
\begin{restatable}[Time-dependent matrix encoding]{definition}{BlockencodingTD}
	\label{def:HAM-T}
	Given a matrix ${H}(s) : [0,t] \rightarrow \mathbb {C}^{2^{n_s}\times 2^{n_s}}$, integer $M>0$, and a promise $\|H\|\le\alpha$, assume there exists a unitary oracle $\operatorname{HAM-T}\in \mathbb {C}^{M2^{n_a+n_s}\times M2^{n_a+n_s}}$ such that
	\begin{align}
	\operatorname{HAM-T} &= 
	\left(\begin{matrix}
	H/\alpha & \cdot \\
	\cdot & \cdot
	\end{matrix}\right),
	\quad H=\operatorname{Diagonal}[H(0),H(t/M),\cdots,H((M-1)t/M)],
	\\\nonumber
	\quad
	&\Rightarrow (\bra{0}_a\otimes \openone_s) \operatorname{HAM-T} (\ket{0}_a\otimes \openone_s) = \sum^{M-1}_{m=0}\ket{m}\bra{m} \otimes \frac{H(mt/M)}{\alpha}.
	\end{align}
\end{restatable}
In later sections where the time-dependent simulation algorithm is applied to the interaction picture, we `open up' this oracle and discuss the gate complexity of its implementation. We assume that the query complexity to a controlled-unitary black-box is the same as that to the original black-box. In general, this will not affect the gate complexity. Though there are often cleverer ways to implement an arbitrary controlled-unitary, in the worst-case, all quantum gates may be replaced by their controlled versions with a constant multiplicative overhead.

\subsection{Truncated Dyson series algorithm}
We now state our main algorithm for general time-dependent simulation.
\begin{restatable}[Hamiltonian simulation by a truncated Dyson series]{theorem}{mainThm}
	\label{Thm:Compressed_TDS}
	Let $H(s) : [0,t] \rightarrow \mathbb{C}^{2^{n_s}\times 2^{n_s}}$, let it be promised that $\max_{s}\|H(s)\|\le\alpha$ and $\langle\|\dot H\|\rangle=\frac{1}{t}\int^t_{0} \left\|\frac{\mathrm{d} H(s)}{ \mathrm{d} s}\right\| \mathrm{d}s$ and assume $M\in{\mathcal{O}}\left( \frac{t^2}{\epsilon}\left({\langle \|\dot{H}\| \rangle}  +{\max_s \|H(s)\|^2}\right)\right)$ in \cref{def:HAM-T}. 
	 For all $t\in[0,\frac{1}{2\alpha}]$ and $\epsilon> 0$, an operation $W$  can be implemented  such that $\left\|W-\mathcal{T}\left[ e^{-i\int_0^t H(s) \mathrm{d} s}\right]\right\| \le \epsilon$ with failure probability $\mathcal{O}(\epsilon)$ with the following costs.
	\begin{enumerate}
		\item Queries to $\operatorname{HAM-T}$: $\mathcal{O}\left(\frac{\log{(1/\epsilon)}}{\log\log{(1 /\epsilon)}}\right)$,
		\item Qubits: $n_s + \mathcal{O}\left(n_a+\log{\left( \frac{t^2}{\epsilon}\left({\langle \|\dot{H}\| \rangle}  +{\max_s \|H(s)\|^2}\right)\right)}\right)$,
		\item Primitive gates: $\mathcal{O}\left(\left(n_a + \log{\left( \frac{t^2}{\epsilon}\left({\langle \|\dot{H}\| \rangle}  +{\max_s \|H(s)\|^2}\right)\right)}\right)\frac{\log{(1/\epsilon)}}{\log\log{(1/\epsilon)}}\right)$.
		\end{enumerate}
\end{restatable}
The proof and circuit construction of~\cref{Thm:Compressed_TDS} is given in~\cref{sec:compresson_gadget}. This algorithm simulates time-evolution for short durations $|\tau|\le \frac{1}{2\alpha}$, which we call a segment. Thus simulation for longer durations $|t|>\frac{1}{2\alpha}$ require multiple segments that each query a different oracle $\operatorname{HAM-T}_j$ encoding $H(s)$ over a different time domain. The complexity of this multi-segment approach is as follows.
\begin{restatable}[Multi-segment Hamiltonian simulation by a truncated Dyson series]{corollary}{corMain}
	\label{Thm:Compressed_TDS_Uniform}
	Let $H(s) : [0,t] \rightarrow \mathbb{C}^{2^{n_s}\times 2^{n_s}}$, and let it be promised that $\max_{s}\|H(s)\|\le\alpha$, and $\langle\|\dot H\|\rangle=\frac{1}{t}\int^{t}_{0} \left\|\frac{\mathrm{d} H(s)}{ \mathrm{d} s}\right\| \mathrm{d}s$. Let $\tau=t/\left\lceil 2\alpha t\right\rceil$ and assume $H_j(s)=H((j-1)\tau + s): s\in[0,\tau]$ is accessed by an oracle $\operatorname{HAM-T}_j$ of the form specified in~\cref{def:HAM-T} with $M\in{\mathcal{O}}\left( \frac{t}{\alpha\epsilon}\left({\langle \|\dot{H}\| \rangle}  +{\max_s \|H(s)\|^2}\right)\right)$.  For all $|t|\ge 0$ and $\epsilon> 0$, an operation $W$ can be implemented with failure probability at most $\mathcal{O}(\epsilon)$ such that $\left\|W-\mathcal{T}\left[ e^{-i\int_0^t H(s) \mathrm{d} s}\right]\right\| \le \epsilon$ with the following cost.
	\begin{enumerate}
		\item Queries to all $\operatorname{HAM-T}_j$: $\mathcal{O}\left(\alpha t\frac{\log{(\alpha t/\epsilon)}}{\log\log{(\alpha t /\epsilon)}}\right)$,
		\item Qubits: $n_s + \mathcal{O}\left(n_a+\log{\left( \frac{t}{\alpha\epsilon}\left({\langle \|\dot{H}\| \rangle} +{\max_s \|H(s)\|^2}\right)\right)}\right)$,
		\item Primitive gates: $\mathcal{O}\left(\alpha t\left(n_a +\log{\left( \frac{t}{\alpha\epsilon}\left({\langle \|\dot{H}\| \rangle}  +{\max_s \|H(s)\|^2}\right)\right)}\right)\frac{\log{(\alpha t/\epsilon)}}{\log\log{(\alpha t/\epsilon)}}\right)$.
	\end{enumerate}
\end{restatable}
\begin{proof}
	The time-ordered evolution operator $\mathcal{T}\left[ e^{-i\int_0^t H(s) \mathrm{d} s}\right]=\prod^L_{j=1}\mathcal{T}\left[ e^{-i\int_{t_{j-1}}^{t_j} H(s) \mathrm{d} s}\right]$ may be broken into $L=\mathcal{O}(\alpha t)$ segments, where $[0,t]=\cup^L_{j=1}[t_{j-1},t_j]$ and $0=t_0<t_1<\cdots<t_L=t$ and $t_{j}-t_{j-1}\in \mathcal{O}(1/\alpha)$. Each segment is then simulated using~\cref{Thm:Compressed_TDS} to error $\delta$. From the proof of~\cref{Thm:Compressed_TDS} in~\cref{sec:compresson_gadget}, each segment is a unitary quantum circuit $V_j$ such that the (near-unitary) operations $W_j=(\bra{0}\otimes \openone_{s}) V_j (\ket{0}\otimes \openone_{s})$ where
	\begin{align}
	\left\|(\bra{0}\otimes \openone_{s}) V_j (\ket{0}\otimes \openone_{s})-\mathcal{T}\left[ e^{-i\int_{t_{j-1}}^{t_j} H(s) \mathrm{d} s}\right]\right\| \le \delta.
	\end{align}
	To obtain the error of applying these $W_j$ in sequence, note that in general if $A_j$ and $B_j$ are a sequence of arbitrary bounded operators and $\|\cdot \|$ is a sub-multiplicative norm then it is straightforward to show using an inductive argument and a triangle inequality that for all positive integer $L$,
	\begin{align}
	\label{eq:product_norm_error}
	\left\|\prod_{j=1}^L A_j - \prod_{j=1}^L B_j\right\|\le \sum_{k=1}^L \left( \prod_{j=1}^{k-1} \left\|A_j \right\|\right)\left\| A_k-B_k\right\|\left( \prod_{j=k+1}^L \left\|B_j\right\|\right).
	\end{align}
	Let us simply notation by setting $A_j=W_j$ and $B_j=\mathcal{T}\big[ e^{-i\int_{t_{j-1}}^{t_j} H(s) \mathrm{d} s}\big]$. Using the fact that $\|W_j\|\le 1$ and $\|U_j\|=1$,~\cref{eq:product_norm_error} yields
	\begin{align}
	\left\|\prod^L_jW_j-\mathcal{T}\left[ e^{-i\int_{0}^{t} H(s) \mathrm{d} s}\right]\right\|\le L\delta.
	\end{align}  
	As $W_j$ is obtained by applying $V_j$ on the $\ket{0}$ state in the ancilla register followed by a projection back onto $\bra{0}$ state, let $\Pi=(\openone_\cdot-\ket{0}\bra{0}_\cdot)\otimes \openone_s$ be this projector. Then
	\begin{align}
	\left\|(\openone-\Pi) V_j(\ket{0}\otimes \openone_s)\right\|\le\sqrt{1-(1-\delta)^2}=\sqrt{2\delta-\delta^2}\le\sqrt{2\delta}.
	\end{align}  
	Thus the failure probability of projecting onto $\bra{0}$ is $\le 2\delta$. With $L$ repetitions, this failure probability is $\le 1-(1-\delta)^L\in \mathcal{O}(L\delta)$. Thus we identify $W=W_L\cdots W_1$ and choose choose $\delta=\epsilon/L$ to ensure that error of $W$ and the failure probability of its application is at most $\mathcal{O}(\epsilon)$.  The result then follows from the fact that $L \in \Theta(\alpha t)$.
\end{proof}

\subsection{Discretization and truncation error}
Note that~\cref{Thm:Compressed_TDS} and~\cref{Thm:Compressed_TDS_Uniform} all make a particular choice of $M$, which controls the number of points at which $H(s)$ is evaluated in the oracle of~\cref{def:HAM-T}. This is determined precisely by the error incurred in truncating the Dyson series at some finite order $k=K\ge0$, and evaluating $H(t_j)$ at some finite number of $M$ time-steps of size $\Delta=t/M$. Thus we have the approximation
\begin{align}
\mathcal{T}\left[e^{-i\int_0^t H(s) \mathrm{d}s}\right]
\approx 
\sum^K_{k=0} (-i)^k D_k
\approx
\sum^\infty_{k=0} \frac{(-it)^k }{k! M^k}\tilde{B}_k,
\quad
\tilde{B}_k=\sum_{m_1, \cdots, m_k=0}^{M-1}\mathcal{T}\left[H\left({m_k \Delta}\right)\cdots H\left({m_1 \Delta}\right)\right],
\end{align}
which converges to $U(t)$ in the limit $K,M\rightarrow \infty$ if $H(t)$ is Riemann integrable.
The time-ordering operator in $\tilde{B}_k$ may be removed with a slightly different approximation
\begin{align}
\tilde{B}_k= k! B_k+C_k,
\quad 
B_k=\sum_{0\le m_1 <\cdots < m_k < M}H\left({m_k \Delta}\right)\cdots H\left({m_1 \Delta}\right),
\end{align}
where $C_k$ captures terms where at least one pair of indices $m_j=m_{k}$ collide for $j\neq k$. Given a target error $\epsilon$ and failure probability $\mathcal{O}(\epsilon)$, the required $K$ and $M$ are given by the following.
\begin{restatable}[Error from truncating and discretizing the Dyson series]{lemma}{truncatedDysonError}
\label{Thm:Truncated_Dyson_Algorithm}
Let $H(s): [0,t] \mapsto \mathbb{C}^{N\times N}$ be differentiable and  $\langle\|\dot{H}\|\rangle:=\frac{1}{t}\int^t_{0} \left\|\frac{\mathrm{d} H(s)}{ \mathrm{d} s}\right\| \mathrm{d}s$.  For any $\epsilon\in (0,2^{1-e}]$, an approximation to the time ordered operator exponential of $-iH(s)$ can be constructed such that
$$
\left\|\mathcal{T}\left[e^{-i\int_0^t H(s) \mathrm{d}s}\right]- \sum^K_{k=0} {(-i t/M)^k} B_k\right\|\le \epsilon, \quad B_k=\sum_{0\le m_1 <\cdots < m_k < M}H\left({m_k t/M}\right)\cdots H\left({m_1 t/M}\right),
$$
if we take all the following are true.
\begin{enumerate}
	\item $\max_s\|H(s)\|t\le \ln 2$.
	\item $K= \left\lceil-1+\frac{2\ln(2/\epsilon)}{\ln\ln(2/\epsilon) +1}\right\rceil$.
	\item $M \ge  \max\left\{\frac{16 t^2}{\epsilon} \left(\langle \|\dot{H}\| \rangle  +\max_s \|H(s)\|^2\right),K^2\right\}$.
\end{enumerate}
\end{restatable}
A detailed proof is presented in~\cref{sec:proof_truncation_discretization_dyson}. On a quantum computer, it is possible to compute the $B_k$ exactly and efficiently even if they sum over exponentially many points $M$. In contrast, computing these Riemann sums on classical computer would be prohibitive, even by approximate Monte-Carlo sampling, which is exacerbated by the sign problem. However, this efficient quantum computation crucially assumes that information describing the Hamiltonian at different times are made accessible in a certain coherent manner -- in our case, this is information accessed through black-box unitary oracles described by~\cref{def:HAM-T}.

\section{Interaction picture simulation}
\label{sec:interaction_picture}
Time-independent Hamiltonians $H$ become time-dependent $H_I(t)$ in the interaction picture. Simulating this requires the use of time-dependent Hamiltonian simulation algorithms, which scale with parameters of $H_I(t)$ that differ from those for the time-independent case. For certain broad classes of Hamiltonian identified in~\cref{sec:implementation_HAM_T}, these different dependencies allow us to improve the gate complexity of approximating the time-evolution operator $e^{-iHt}$ by instead performing simulation in the interaction picture using the truncated Dyson series algorithm~\cref{Thm:Compressed_TDS}.

The interaction picture can be viewed as an intermediate between the Schr\"odinger and Heisenberg pictures wherein some of the dynamics is absorbed into the state and the remainder is absorbed into the dynamics of the operators. If the Hamiltonian in the Schr\"odinger picture $H=A+B$ generates time-evolution like $\ket{\psi(t)} =e^{-iHt} \ket{\psi(0)}$, then the Hamiltonian in the interaction picture is $H_I(t) = e^{iAt}Be^{-iAt}$ and $i\partial_{t} \ket{\psi_I(t)} =H_I(t) \ket{\psi_I(t)}$ with $\ket{\psi_I(t)}=e^{i A t}\ket{\psi(t)}$ for all $t$.  These relations can easily be seen by substituting into the Schr\"odinger equation:
\begin{align}
i\partial_t \ket{\psi_I(t)} &=i\partial_t\left( e^{iA t} \ket{\psi(t)}\right)
=e^{iA t}(-A + H) \ket{\psi(t)}=e^{iA t}Be^{-iA t} e^{iA t} \ket{\psi(t)}\nonumber\\
&= H_{I}(t) \ket{\psi_I(t)}.
\end{align}
Note that if we started with time-dependent $B(t)$, that is $H(t)= A+B(t)$, the interaction picture Hamiltonian is $H_I(t) = e^{iAt}B(t)e^{-iAt}$. Our following results generalize easily to this situation, and so we consider time-independent $B$ for simplicity.

The advantage of this representation is a Hamiltonian with a smaller norm $\|H_I(t)\|=\|B\| \le \|A\|+\|B\|$, but at the price of introducing time-dependence. The following notation is commonly used to express this time-evolution operator $\mathcal{T}\left[e^{-i\int_0^t H(s) \mathrm{d} s}\right]= \lim_{r\rightarrow \infty} \prod_{j=1}^r e^{-i H(jt/r) t/r}$ where this product is implicitly defined to be time ordered. Given an initial state $\ket{\psi(0)}$, the state after evolution for $t>0$ in the Schr\"{o}dinger picture may thus be written as
\begin{align}
\label{eq:Interaction_picture_state_evolutionA}
\ket{\psi(t)} &= e^{-iAt}\ket{\psi_I(t)}=e^{-iAt}\mathcal{T}\left[e^{-i\int_0^t H_I(s) \mathrm{d} s}\right]\ket{\psi_I(0)}
=e^{-iAt}\mathcal{T}\left[ e^{-i\int_0^t H_I(s) \mathrm{d} s}\right]\ket{\psi(0)}
=e^{-i(A+B)t}\ket{\psi(0)}.
\end{align}
As this is true for any $t$, evolution by the full duration is identical to evolution by $L$ shorter segments of duration $\tau= t/L$, such as
\begin{align}
\label{eq:Interaction_picture_state_evolution}
\ket{\psi(t)} &= (e^{-i(A+B)\tau})^L\ket{\psi(0)}
=\left(e^{-iA\tau}\mathcal{T}\left[ e^{-i\int_0^\tau H_I(s) \mathrm{d} s}\right]\right)^L\ket{\psi(0)}.
\end{align}

Using the simulation algorithm~\cref{Thm:Compressed_TDS} to simulate each segment in~\cref{eq:Interaction_picture_state_evolution} leads to the following result.
\begin{lemma}[Query complexity of Hamiltonian simulation in the interaction picture]
	\label{thm:int_pic_sim_query}
Let $A\in\mathbb{C}^{2^{n_s}\times 2^{n_s}}$, $B\in\mathbb{C}^{2^{n_s}\times 2^{n_s}}$, let $\alpha_A$ and $\alpha_B$ be known constants such that $\|A\|\le \alpha_A$ and $\|B\|\le \alpha_B$.  Assume the existence of a unitary oracle that implements the Hamiltonian within the interaction picture, denoted $\operatorname{HAM-T}\in \mathbb{C}^{2^{n_s+n_a}\times 2^{n_s+n_a}}$ which implicitly depends on the time-step size $\tau\in\mathcal{O}(\alpha_B^{-1})$ and number of time-steps $M\in{\mathcal{O}}\left( \frac{t}{\epsilon}\left(\alpha_A+\alpha_B\right)\right)$,  such that
\begin{align}
\label{eq:int_pic_sim_Ham-T}
(\bra{0}_a\otimes \openone_s) \operatorname{HAM-T} (\ket{0}_a\otimes \openone_s) &= \sum^{M-1}_{m=0}\ket{m}\bra{m} \otimes \frac{e^{iA\tau m/M}Be^{-iA\tau m/M}}{\alpha_B},
\end{align}
For all $t\ge 2\alpha_B\tau$, the time-evolution operator $e^{-i(A+B)t}$ may be approximated to error $\epsilon$ with the following cost.
\begin{enumerate}
	\item Simulations of $e^{-iA\tau}$: $\mathcal{O}(\alpha_B t)$, 
	\item Queries to $\operatorname{HAM-T}$: $\mathcal{O}\left(\alpha_B t\frac{\log{(\alpha_B t/\epsilon)}}{\log\log{(\alpha_B t /\epsilon)}}\right)$,
	\item Qubits: $n_s + \mathcal{O}\left(n_a+\log{\left( \frac{t}{\epsilon}\left(\alpha_A+\alpha_B\right)\right)}\right)$,
	\item Primitive gates: $\mathcal{O}\left(\alpha_B t\left(n_a + \log{\left( \frac{t}{\epsilon}\left(\alpha_A+\alpha_B\right)\right)}\right)\frac{\log{(\alpha_B t/\epsilon)}}{\log\log{(\alpha_B t/\epsilon)}}\right)$.
\end{enumerate}
\end{lemma}
\begin{proof}
	According to~\cref{Thm:Compressed_TDS_Uniform}, the maximum time duration of simulation $\tau$ in each segment of~\cref{eq:Interaction_picture_state_evolution} is limited to $\tau\in{\mathcal{O}}(\alpha_B^{-1})$. Thus there are $L\in{\mathcal{O}}(\alpha_B t)$ segments. Each segment also requires one application of $e^{-iA\tau}$, thus the query complexity to $e^{-iA\tau}$ is $L$. Note that for all $t\le \frac{1}{2\alpha_B}$, only one application of $e^{-iAt}$ is required.	Each segment also requires one application of $\mathcal{T}\left[ e^{-i\int_0^\tau H_I(s) \mathrm{d} s}\right]$, which we approximate with $\operatorname{TDS}$ from~\cref{Thm:Compressed_TDS}. 
	By a triangle inequality, it suffices to simulate each segment with error $\delta=\mathcal{O}(\epsilon /L)$ in order to obtain a total error $\epsilon$. Simulating each segment makes $\mathcal{O}\left(\frac{\log{(\alpha_B t/\epsilon)}}{\log\log{(\alpha_B t /\epsilon)}}\right)$ queries to $\operatorname{HAM-T}$. Thus simulation for the full duration makes $\mathcal{O}\left(\alpha_B t\frac{\log{(\alpha_B t/\epsilon)}}{\log\log{(\alpha_B t /\epsilon)}}\right)$ queries to $\operatorname{HAM-T}$. 
	
	The number of qubits, primitive gates, and discretization points $M$ required are obtained directly from the conditions of~\cref{Thm:Compressed_TDS_Uniform}. We obtain the stated $M$ by substituting the facts $\max_s \|H_I(s)\|\le \|B\|$, $\langle\|\dot H\|\rangle=\|[A,B]\|\le 2\|A\|\|B\|$, and $\tau\in{\mathcal{O}}(\alpha_B^{-1})$. Thus it suffices to choose
\begin{align}
\label{eq:int_pic_sim_M_points}
M\in{\mathcal{O}}\left( \frac{\alpha_B t\tau^2}{\epsilon}\left({\langle \|\dot{H}\| \rangle}  +{\max_s \|H(s)\|^2}\right)\right)
\subseteq{\mathcal{O}}\left( \frac{ t}{\epsilon}\left(\frac{\|A\|\|B\|}{\alpha_B} + \alpha_B\right)\right)\subseteq {\mathcal{O}}\left( \frac{ t}{\epsilon}\left(\alpha_A + \alpha_B\right)\right).
\end{align}

\end{proof}

\subsection{Comparison with simulation of time-independent Hamiltonians in the Schr\"odinger picture}
\label{sec:implementation_HAM_T}
We now compare the cost of simulation in the interaction picture using the truncated Dyson series, with state-of-art simulation in the Schr\"odinger picture with time-independent Hamiltonians using the truncated Taylor series approach~\cite{berry2015simulating} outlined in~\cref{sec:truncated_taylor_algorithm}. Up to logarithmic factors, this comparison is valid as the truncated Taylor series algorithm cost differs from optimal algorithms~\cite{Low2016HamSim,Low2016qubitization} by only logarithmic factors. For any Hamiltonian $H=A+B$, let us assume access to the oracles
\begin{align}
\label{eq:AB_standard_form_encoding}
(\bra{0}_a\otimes \openone_s) O_A (\ket{0}_a\otimes \openone_s) = \frac{A}{\alpha_A},
\quad
(\bra{0}_a\otimes \openone_s) O_B (\ket{0}_a\otimes \openone_s) = \frac{B}{\alpha_B},
\end{align}
which have gate complexity $C_A,C_B$ respectively, and encode $A,B$ using $n_a$ additional qubits. Note that in every case where $O_A$ and $O_B$ act non-trivially on $\ket{0}_a$, the cost $C_A,C_B \ge n_a$. The gate complexity of time-independent simulation $e^{-i(A+B)t}$ by prior art is then
\begin{align}
\label{eq:sch_pic_sim_cost}
C_{\text{TTS}}&\in{\mathcal{O}}\left((C_A+C_B)(\alpha_A+\alpha_B)t\frac{\log{((\alpha_A+\alpha_B)t/\epsilon)}}{\log\log{((\alpha_A+\alpha_B)t /\epsilon)}}\right)
\end{align}
In contrast, we prove the following theorem for simulation in the interaction picture.
\begin{theorem}[Gate complexity of Hamiltonian simulation in the interaction picture]
	\label{thm:int_pic_sim}
	Let $A\in\mathbb{C}^{2^{n_s}\times 2^{n_s}}$, $B\in\mathbb{C}^{2^{n_s}\times 2^{n_s}}$ be time-independent Hamiltonians that are promised to obey $\|A\|\le \alpha$, $\|B\|\le \alpha_B$, such that
	\begin{enumerate}
		\item $B$ is encoded in an oracle $(\bra{0}_a\otimes \openone_s) O_B (\ket{0}_a\otimes \openone_s) = \frac{B}{\alpha_B}$ using $n_a$ additional qubits and $C_B\ge n_a$ gates.
		\item $e^{-iAs}$ is approximated to error $\epsilon$	
		using $C_{e^{-iAs}}[\epsilon]\in{\mathcal{O}}\left(|s| {\rm log}^{\gamma}(s/\epsilon)\right)$ gates for some $\gamma>0$ and any $|s|\ge0$.
	\end{enumerate}
	For all $t>0$, the time-evolution operator $e^{-i(A+B)t}$ may be approximated to error $\epsilon$ with  gate complexity
\begin{align}
\label{eq:int_pic_sim_cost}
C_{\text{TDS}}&\in{\mathcal{O}}\left(
\alpha_B t
\left(C_B+C_{e^{-iA/\alpha_B}}\left[\frac{\epsilon}{\alpha_B t\log{(\alpha_B)}}\right]\log{\left(\frac{t(\alpha_A+\alpha_B)}{\epsilon}\right)}
\right)\frac{\log{(\alpha_B t/\epsilon)}}{\log\log{(\alpha_B t 	/\epsilon)}}
\right)
\\\nonumber
&={\mathcal{O}}\left(
\alpha_B t
\left(C_B+C_{e^{-iA/\alpha_B}}\left[\epsilon\right]\right)
\operatorname{polylog}(t(\alpha_A+\alpha_B)/\epsilon)
\right).
\end{align}
\end{theorem}
\begin{proof}
Expressing~\cref{thm:int_pic_sim_query} solely in terms of gate complexity requires an expression $C_{\operatorname{HAM-T}}[\delta]$ for the cost of approximating the oracle $\operatorname{HAM-T}$ to error $\delta$. One possible construction is
\begin{align}
\operatorname{HAM-T}&= \left(\sum^{M-1}_{m=0}\ket{m}\bra{m}_d \otimes \openone_a\otimes  e^{iA\tau m/M}\right)\cdot(\openone_d \otimes O_B)\cdot\left(\sum^{M-1}_{m=0}\ket{m}\bra{m}_d \otimes \openone_a\otimes e^{-iA\tau m/M}\right), 
\quad \tau\in\mathcal{O}(\alpha_B^{-1}).
\end{align}
Note that the controlled-Hamiltonian evolution operator $\sum^{M-1}_{m=0}\ket{m}\bra{m}_d \otimes e^{iA\tau m/M}$ may be implemented by $\lceil\log_2{(M)}\rceil$ controlled-$e^{iA\tau/M},e^{iA2\tau/M}, e^{iA4\tau/M},\cdots,e^{iA 2^{\lceil\log_2{(M)}\rceil}\tau/M}$. By approximating each controlled-$e^{iAj\tau/M}$ with error $\mathcal{O}(\delta/\log{(M)})$, the overall error will be bounded by $\mathcal{O}(\delta)$. As we assume that $C_{e^{-iAt}}[\delta/\log{(M)}]\in{\mathcal{O}}\left(|t| \log^{\gamma}{(t\log{(M)}/\delta)}\right)$, each controlled-$e^{iAj\tau/M}$ costs at most ${\mathcal{O}}\left(\tau \log^{\gamma}{(\tau\log{(M)}/\delta)}\right)$ gates. Thus the cost of all controlled-$e^{iAj\tau/M}$ sums to $
{\mathcal{O}}\left(\tau \log{(M)}\log^{\gamma}{(\tau\log{(M)}/\delta)}\right)
={\mathcal{O}}\left(\tau \log{(M)}\log^{\gamma}{(\tau/\delta)}\right)$. By adding the cost of $O_B$, 
\begin{align}
C_{\operatorname{HAM-T}}[\delta]
\in{\mathcal{O}}\left(C_B+ \tau\log{(M)}\log^{\gamma}{(\tau/\delta)}\right)
={\mathcal{O}}\left(C_B+ \frac{1}{\alpha_B}\log{(M)}\log^{\gamma}{\left(\frac{1}{\alpha_B\delta}\right)}\right)
.
\end{align}

From~\cref{thm:int_pic_sim_query}, approximating $e^{-iHt}$ to error $\mathcal{O}(\epsilon)$ requires $\mathcal{O}\left(\alpha_B t\frac{\log{(\alpha_B t/\epsilon)}}{\log\log{(\alpha_B t /\epsilon)}}\right)$ queries to $\operatorname{HAM-T}$. We obtain this overall error by approximating each $\operatorname{HAM-T}$ query with error $\delta\in\mathcal{O}\left(\epsilon\frac{\log\log{(\alpha_B t /\epsilon)}}{\alpha_B t\log{(\alpha_B t/\epsilon)}}\right)$. Thus
\begin{align}
C_{\operatorname{HAM-T}}[\delta]\in{\mathcal{O}}\left(C_B+\frac{1}{\alpha_B} \log{(M)}\log^{\gamma}{\left(\frac{1}{\epsilon}\frac{ t\log{(\alpha_B t/\epsilon)}}{\log\log{(\alpha_B t /\epsilon)}}\right)}\right)
={\mathcal{O}}\left(C_B+\frac{1}{\alpha_B} \log{(M)}\log^{\gamma}{\left(\frac{t \log{(\alpha_B)}}{\epsilon}\right)}\right).
\end{align}

From~\cref{thm:int_pic_sim_query}, we also require $\mathcal{O}(\alpha_Bt)$ queries to $e^{-iA\tau}$, and $\mathcal{O}\left((n_a + \log{(M)})\alpha_B t\frac{\log{(\alpha_B t/\epsilon)}}{\log\log{(\alpha_B t/\epsilon)}}\right)$ primitive gates. It suffices to approximate each $e^{-iA\tau}$ to error $\mathcal{O}(\epsilon/(\alpha_Bt))$. By adding all these contributions, the total gate complexity of simulation in the interaction picture is
\begin{align}
&\mathcal{O}\left(\alpha_B t\left(\left(C_{\operatorname{HAM-T}}[\delta]+(n_a + \log{(M)})\right)\frac{\log{(\alpha_B t/\epsilon)}}{\log\log{(\alpha_B t /\epsilon)}}+C_{e^{-iA\tau}}[\epsilon/(\alpha_Bt)]\right)\right)
\\\nonumber
&=\mathcal{O}
\left(
	\alpha_B t
	\left(
		\left(C_B+\frac{1}{\alpha_B} \log{(M)}\log^{\gamma}{\left(\frac{t \log{(\alpha_B)}}{\epsilon}\right)}+(n_a + \log{(M)})
		\right)\frac{\log{(\alpha_B t/\epsilon)}}{\log\log{(\alpha_B t 	/\epsilon)}} + \frac{1}{\alpha_B} \log^{\gamma}{\left(\frac{t}{\epsilon}\right)}
	\right)
\right)
\\\nonumber
&=\mathcal{O}
\left(
\alpha_B t
\left(C_B+n_a+\frac{1}{\alpha_B} \log^{\gamma}{\left(\frac{t\log{(\alpha_B)}}{\epsilon}\right)}\log{\left(M\right)}
\right)\frac{\log{(\alpha_B t/\epsilon)}}{\log\log{(\alpha_B t 	/\epsilon)}}
\right)
\\\nonumber
&=\mathcal{O}
\left(
\alpha_B t
\left(C_B+C_{e^{-iA/\alpha_B}}\left[\frac{\epsilon}{\alpha_B t\log{(\alpha_B)}}\right]\log{\left(\frac{t(\alpha_A+\alpha_B)}{\epsilon}\right)}
\right)\frac{\log{(\alpha_B t/\epsilon)}}{\log\log{(\alpha_B t 	/\epsilon)}}.
\right).
\end{align}
\end{proof}

From comparing~\cref{eq:sch_pic_sim_cost,eq:int_pic_sim_cost}, we may immediately state sufficient criteria for when simulation in the interaction picture is advantageous over simulation in the Schr\"odinger picture.
\begin{enumerate}
	\item The upper bound on the spectral norms $\alpha_A\ge \|A\|$, $\alpha_B\ge \|B\|$ of the encoding in~\cref{eq:AB_standard_form_encoding} satisfy $\alpha_A\gg \alpha_B$. Generally speaking, this is correlated with term $A$ representing fast dynamics $\|A\|\gg \|B\|$.
	\item The gate complexity of time-evolution by $A$ alone for time $\tau\in{\mathcal{O}}{(\alpha_B^{-1})}$ is comparable to that of synthesizing the oracle $O_B$, that is $C_{e^{iA\alpha^{-1}_B}} \in{\mathcal{O}}{(C_B)}$. 
\end{enumerate}
Note that satisfying condition (2) depends strongly on the structure of $A,B$. For instance, a simulation of $A$ for time $\tau \in{\mathcal{O}}{(\alpha^{-1}_B)}$ using generic time-independent techniques has gate complexity $\tilde{\mathcal{O}}\left(C_A \alpha_A / \alpha_B\right)$. As we are interested in the case $\|A\|/ \|B\| \gg 1$, this quantity could be large and scale poorly with the problem size. One very strong sufficient assumption is that $e^{-iAt}$ is cheap and can be fast-forwarded, such that the gate complexity $C_{e^{-iAt}}[\epsilon]\in{\mathcal{O}}{\left(\operatorname{polylog}(t/\epsilon)\right)}$ is constant up to logarithmic factor. This turns out to be a reasonable assumption in the application we next consider.

\section{Application to the Hubbard model with long-ranged interactions}
\label{sec:Hubbard}
We now apply the technology developed in~\cref{sec:truncated_dyson_series} and~\cref{sec:interaction_picture} for Hamiltonian simulation in the interaction picture to physical problems of practical interest. We focus on the periodic Hubbard model in $d$-dimensions with $N$ lattice sites subject to local disorder and translational-invariant two-body couplings that may be long-ranged in general. We perform a gate complexity comparison with simulation by time-independent techniques, and later in~\cref{sec:Chemistry}, we specialize this model to that of quantum chemistry simulations in the plane-wave and dual basis~\cite{babbush2017low}.

The periodic Hubbard Hamiltonian we consider has the form $H=T+U+V$, where $T$ is the kinetic energy hopping operator, $U$ is the local single-site potential, and $V$ is a symmetric translationally-invariant two-body density coupling term between opposite spins. In the dual basis, $H$ is expressed in terms of single-site Fermionic creation and annihilation operators $\{a_{\vec{x}\sigma},a_{\vec{y}\sigma'}\}=\{a^\dag_{\vec{x}\sigma},a^\dag_{\vec{y}\sigma'}\}=0$, $\{a_{\vec{x}\sigma},a^\dag_{\vec{y}\sigma'}\}=\delta_{\vec{x}\vec{y}}\delta_{\sigma\sigma'}$, and the number operator $n_{\vec{x}\sigma}=a^\dag_{\vec{x}\sigma}a_{\vec{x}\sigma}$. The subscript $\vec{x}\in [-N^{1/d},N^{1/d}]^d$ indexes one of $N$ lattice sites in $d$ dimensions, and $\sigma \in\{-1,1\}$ is a spin-$\frac{1}{2}$ index. Explicitly,
\begin{align}
\label{eq:Hubbdard_Ham}
H=
\sum_{\vec{x},\vec{y},\sigma}T(\vec{x}-\vec{y})a^\dag_{\vec{x}\sigma}a_{\vec{y}\sigma}
+
\sum_{\vec{x},\sigma}U(\vec{x},\sigma)n_{\vec{x}\sigma}
+
\sum_{(\vec{x},\sigma)\neq(\vec{y},\sigma')}V(\vec{x}-\vec{y})n_{\vec{x}\sigma}n_{\vec{y}\sigma'},
\end{align}
where the coefficients $T(\vec{s}), U(\vec{s},\sigma), V(\vec{s})$ are real functions of the $\vec{s}\in [-N^{1/d},N^{1/d}]^d$. 

Further simplification of~\cref{eq:Hubbdard_Ham} is possible as the kinetic energy operator is diagonal in the plane-wave basis. This basis related to the dual basis by a unitary rotation $\operatorname{FFFT}$, an acronym for `Fast-Fermionic-Fourier-Transform'~\cite{babbush2017low} that implements a Fourier transform over the lattice site indices, resulting in Fermionic creation and annihilation operators $c^\dag_{\vec{p}\sigma},c_{\vec{p}\sigma}$. 
\begin{align}
\label{eq:fermionic_fourier_basis}
c_{\vec{p}\sigma}=\frac{1}{\sqrt{N}}\sum_{\vec{x}}a_{\vec{x}\sigma}\;e^{i 2\pi\vec{p}\cdot\vec{x}/N^{1/d}}
= \operatorname{FFFT}^\dag a_{\vec{p}\sigma}\operatorname{FFFT},
\quad
c^\dag_{\vec{p}\sigma}
=\frac{1}{\sqrt{N}}\sum_{\vec{x}}a^\dag_{\vec{x}\sigma}\;e^{-i 2\pi\vec{p}\cdot\vec{x}/N^{1/d}}
= \operatorname{FFFT}^\dag a^\dag_{\vec{p}\sigma}\operatorname{FFFT},
\end{align}
By substituting the Fourier transform of the kinetic term $T(\vec{s}) = \frac{1}{N}\sum_{\vec{p}}\tilde{T}(\vec{p})\;e^{-i 2\pi\vec{p}\cdot\vec{s}/N^{1/d}}$, an equivalent expression for the Hubbard Hamiltonian is
\begin{align}
\label{eq:Hubbard_Fourier}
H=
\operatorname{FFFT}^\dag\cdot \left(\sum_{\vec{x},\sigma}\tilde{T}(\vec{x}){n}_{\vec{x}\sigma}\right)\cdot \operatorname{FFFT}
+
\sum_{\vec{x},\sigma}U(\vec{x},\sigma)n_{\vec{x}\sigma}
+
\sum_{(\vec{x},\sigma)\neq(\vec{y},\sigma')}V(\vec{x}-\vec{y})n_{\vec{x}\sigma}n_{\vec{y}\sigma'},
\end{align}
where each term is now diagonal in their respective bases.

A simulation of this Hamiltonian on a qubit quantum computer requires a map from its Fermionic operators to spin operators. One possibility is the Jordan-Wigner transformation, which requires some map from Fermionic indices to spin indices, such as $f(\vec{x},\sigma)= N\frac{1-\sigma}{2}+\left(\sum^{d-1}_{j=0}\vec{x}_j N^{j/d}\right)$. Subsequently, we replace
\begin{align}
a^\dag_{\vec{x}\sigma}
\rightarrow \frac{1}{2}(X_{f(\vec{x},\sigma)}-i Y_{f(\vec{x},\sigma)})\bigotimes^{f(\vec{x},\sigma)-1}_{j=0}Z_{j},
\quad
a_{\vec{x}\sigma}
\rightarrow \frac{1}{2}(X_{f(\vec{x},\sigma)}+i Y_{f(\vec{x},\sigma)})\bigotimes^{f(\vec{x},\sigma)-1}_{j=0}Z_{j}.
\end{align}
Note the very useful property where number operators map to single-site spin operators $n_{\vec{x}\sigma} \rightarrow \frac{1}{2}(I-Z_{f(\vec{x},\sigma)})$ under this encoding. Moreover, $\operatorname{FFFT}$ can be implemented using $\mathcal{O}(N \log{(N)})$ primitive quantum gates~\cite{Ferris2014} in the Jordan-Wigner representation.

Let us evaluate the worst-case gate-complexity for time-evolution by $H$ in  the Schr\"odinger picture. As an example of state-of-art using the truncated Taylor series approach in~\cref{sec:truncated_taylor_algorithm}, $e^{-i(T+U+V)t}$ may be simulated using $\mathcal{O}((\alpha_T+\alpha_U+\alpha_V)t \log{(t(\alpha_T+\alpha_U+\alpha_V)/\epsilon)})$ queries to oracles that encode $T$, $U$, and $V$ as follows.
\begin{align}
(\bra{0}_{\text{a}}\otimes \openone_{\text{s}}) O_T (\ket{0}_{\text{a}}\otimes \openone_{\text{s}}) = \frac{T}{\alpha_T},
\quad
(\bra{0}_{\text{a}}\otimes \openone_{\text{s}}) O_{U} (\ket{0}_{\text{a}}\otimes \openone_{\text{s}}) = \frac{U}{\alpha_U},
\quad
(\bra{0}_{\text{a}}\otimes \openone_{\text{s}}) O_{V} (\ket{0}_{\text{a}}\otimes \openone_{\text{s}}) = \frac{V}{\alpha_V}.
\label{eq:standard-form-Hubbard}
\end{align}
 The cost of simulation depends strongly on the coefficients $\tilde{T}(\vec{p}), U(\vec{x}), V(\vec{x})$. The most straightforward approach synthesizes these oracles using the linear-combination of unitaries outline in~\cref{eq:LCU}.
 For instance, $\operatorname{O}_T=(\operatorname{PREP_T}^\dag \otimes \operatorname{FFFT}^\dag)\cdot\operatorname{SEL}_T\cdot(\operatorname{PREP}_T\otimes \operatorname{FFFT})$, where
 \begin{align}
\operatorname{PREP}_T\ket{0}_a=\sum_{\vec{p},\sigma}\sqrt{\frac{\tilde{T}(\vec{p})}{\alpha_T}}\ket{\vec{p},\sigma}_a,
\quad
\operatorname{SEL}_T=\sum_{\vec{p},\sigma}\ket{\vec{p},\sigma}\bra{\vec{p},\sigma}_a\otimes n_{\vec{p}\sigma},
\quad
\alpha_T=\sum_{\vec{p},\sigma}|\tilde{T}(\vec{p})|.
\label{eq:LCU_Hubbard}
\end{align}
 and similarly for $U$ and $V$. As there are $\mathcal{O}(N)$ distinct coefficients in the worst-case, each of $\operatorname{PREP}_{T,U,V}$ costs $\mathcal{O}(N)$. As $V$ has $\mathcal{O}(N^2)$ terms, $\operatorname{SEL}_V$ has the largest cost of $\mathcal{O}(N^2)$. Thus overall gate complexity is $\mathcal{O}(N^2(\alpha_T+\alpha_U+\alpha_V) t \log{(t(\alpha_T+\alpha_U+\alpha_V)/\epsilon)})$. As there are $\mathcal{O}(N^2)$ coefficients, $\max\{\alpha_T,\alpha_U,\alpha_V\}\in{\mathcal{O}}{(\alpha_T+N^2)}$, and so the cost of simulation is
 \begin{align}
 \label{eq:hubbard_sch_pic_sim_cost}
\mathcal{O}\left(N^2(\alpha_T+N^2) t \log{\left(\frac{t(\alpha_T+N^2)}{\epsilon}\right)}\right).
 \end{align}

The worst-case gate-complexity may be substantially improved by instead simulating $H$ in the interaction picture using the truncated Dyson series algorithm in~\cref{sec:truncated_dyson_series}. The key idea is to simulate in the rotating frame of the interactions $e^{-i(U+V)t}$, where the Hamiltonian becomes time-dependent like $H_I(t)=e^{i(U+V)t} T e^{-i(U+V)t}$. Using the same oracle $O_T$ in~\cref{eq:standard-form-Hubbard} for the kinetic term, the cost of time-evolution $e^{-i(T+U+V)t}$ by this technique is given by~\cref{eq:int_pic_sim_cost}:
\begin{align}
\label{eq:hubbard_int_pic_sim_cost}
C_{\text{TDS}}&\in{\mathcal{O}}\left(
\left(N+C_{e^{-i(U+V)/\alpha_T}}\left[\frac{\epsilon}{\alpha_T t\log{(\alpha_T)}}\right]\log{\left(\frac{t(\|U+V\|+\alpha_T)}{\epsilon}\right)}
\right)\alpha_T t\frac{\log{(\alpha_T t/\epsilon)}}{\log\log{(\alpha_T t 	/\epsilon)}}\right)
\\\nonumber
&\in{\mathcal{O}}\left(\left(N+C_{e^{i(U+V)/\alpha_T}}[\epsilon]\right)\alpha_T t \operatorname{polylog}((\|U+V\|+\alpha_T)t/\epsilon)\right).
\end{align}
All that remains is to bound the cost of time-evolution by the term $C_{e^{i(U+V)/\alpha_T}}[\epsilon]$. Using the fact that this is diagonal in the Pauli $Z$ basis, the Hamiltonian may be fast-forwarded and so has cost that is independent of the evolution time and error. Thus the most straightforward approach decomposes 
\begin{align}
e^{i(U+V)t}=\left(
\prod_{\vec{x},\sigma}e^{-i U(\vec{x},\sigma)n_{\vec{x}\sigma}t}
\right)
\left(
\prod_{(\vec{x},\sigma)\neq(\vec{y},\sigma')}e^{-i V(\vec{x}-\vec{y})n_{\vec{x}\sigma}n_{\vec{y}\sigma'}t}
\right).
\end{align}
There are $\mathcal{O}{(N^2)}$ exponentials, and so $C_{e^{i(U+V)t}}[\epsilon]\in{\mathcal{O}}(N^2)$, which is independent of $\epsilon$ in the primitive gate set of arbitrary two-qubit rotations, hence 
$
C_{\text{TDS}}\in{\mathcal{O}}\left(
N^2\alpha_Tt\log{\left(\frac{t(\|U+V\|+\alpha_T)}{\epsilon}\right)}
\frac{\log{(\alpha_T t/\epsilon)}}{\log\log{(\alpha_T t 	/\epsilon)}}\right).
$
Compared to~\cref{eq:hubbard_sch_pic_sim_cost}, we already a factor $\mathcal{O}{(N)}$ improvement in cases where the kinetic energy is extensive, meaning that $\alpha_T\in{\mathcal{O}}{(N)}$. 

A further improvement to 
\begin{align}
\label{eq:Int_pic_sim_Hubbard_best_cost}
C_{\text{TDS}}\in{\mathcal{O}}(N\alpha_T t\operatorname{polylog}(N(\|U+V\|+\alpha_T)t/\epsilon))
\end{align}
is possible through a more creative evaluation to reduce the gate complexity of $e^{i(U+V)t}$ from $\mathcal{O}{(N^2)}$ to $\mathcal{O}{(N\log{(N)})}$. Clearly, $C_{e^{iUt}}\in{\mathcal{O}}{(N)}$ with $N$ commuting terms poses no problem. The difficulty lies in constructing time-evolution by the two-body term $e^{iVt}$ such that $C_{e^{iVt}}\in{\mathcal{O}}{(N\log{N})}$. As $V$ is a sum of $\mathcal{O}(N^2)$ commuting terms, a gate cost $\mathcal{O}(N^2)$ appear unavoidable. However, this may be reduced by exploiting the translation symmetry of its coefficients with a discrete Fourier transform. Assuming $V(\vec{x})=V(-\vec{x})$ is real and symmetric, its discrete Fourier transform  $\tilde{V}{(\vec k)}=\sum_{\vec{x}}V(\vec{x})e^{i2\pi \vec{x}\cdot\vec{k}/N^{1/d}}$ only has real coefficients. Let we re-write $V$ from~\cref{eq:Hubbdard_Ham} as
\begin{align}
V &
= \sum_{(\vec{x},\sigma)\neq(\vec{y},\sigma')} V(\vec{x}-\vec{y}) n_{\vec x \sigma} n_{\vec y \sigma'}
=
\sum_{(\vec{x},\sigma)\neq(\vec{y},\sigma')}
\frac{1}{N}\sum_{\vec{k}}\tilde{V}{(\vec k)}e^{-i2\pi (\vec{x}-\vec y)\cdot\vec{k}/N}
n_{\vec x \sigma} n_{\vec y \sigma'}
\\\nonumber
&=
\sum_{\vec{k}}\frac{\tilde{V}{(\vec k)}}{N}
\left(
\sum_{(\vec x, \sigma) , (\vec y, \sigma')}
e^{-i2\pi (\vec{x}-\vec y)\cdot\vec{k}/N}
n_{\vec x, \sigma} n_{\vec y, \sigma'}
-
\sum_{\vec x, \sigma}
n_{\vec x, \sigma}
\right)
\\\nonumber
&=
\sum_{\vec{k}}
\tilde{V}{(\vec k)}
\underbrace{\left(\frac{1}{\sqrt{N}}
	\sum_{\vec x}
	e^{-i2\pi \vec{x}\cdot\vec{k}/N}
	\sum_{\sigma}n_{\vec x, \sigma} \right)
}_{\tilde{\chi}_k}
\underbrace{
	\left(\frac{1}{\sqrt{N}}
	\sum_{\vec y}
	e^{i2\pi \vec{y}\cdot\vec{k}/N}
	\sum_{\sigma'}n_{\vec y, \sigma'}
	\right)
}_{\tilde{\chi}_k^\dag}
-
\sum_{\vec p, \sigma}
\left(\sum_{\vec{k}}
\tilde{V}{(\vec k)}\right)
n_{\vec p, \sigma}
.
\end{align}
Our strategy for implementing $e^{-iV t}$ is based on the following observation: Suppose we had a unitary oracle $O_{\vec{A}}\ket{j}\ket{0}_o\ket{0}_{\text{garb}}=\ket{j}\ket {A_j}_o\ket{g(j)}_{\text{garbage}}$ that on input $\ket{j}\in \mathbb{C}^{\operatorname{dim}[\vec{A}]}$, outputs on the $l$-qubit $o$ register, the value of the $j^\text{th}$ element of some complex vector $\vec{A}$, together with some garbage state $\ket{g(j)}_{\text{garb}}$ of lesser interest required to make the operation reversible. One may then perform a phase rotation that depends on $A_j$ as follows:
\begin{align}
\label{eq:phase_gadget}
\ket{j}\ket{0}_o\ket{0}_{\text{garb}}\ket{0}
\underset{O_{\vec{A}}}{\rightarrow} 
\ket{j}\ket{A_j}_o\ket{g(j)}_{\text{garb}}\ket{0}
\underset{\operatorname{PHASE}}{\rightarrow} 
e^{-i A_j t}\ket{j}\ket{A_j}_o\ket{g(j)}_{\text{garb}}\ket{0}
\underset{O^\dag_{\vec{A}}}{\rightarrow} 
e^{-i A_j t}\ket{j}\ket{0}_o\ket{0}_{\text{garb}}\ket{0}.
\end{align}
If $A_j$ were represented in binary, say, $A_j=\sum^{l-1}_{k=0}q_k 2^{-k}$, $\operatorname{PHASE}$ could be implemented using $\mathcal{O}(l)$ controlled-phase $\ket{0}\bra{0}\otimes I + \ket{1}\bra{1}\otimes e^{-it2^{-k}Z}$ rotations. 

Thus, we construct a unitary $O_{V,\text{binary}}$ with the property that
\begin{align}
O_{V,\text{binary}}\left(\bigotimes_{\vec{x},\sigma}\ket{n_{\vec{x},\sigma}}\right)\ket{0}\ket{0}_{\text{garb}}
&=
\left(\bigotimes_{\vec{x},\sigma}\ket{n_{\vec{x},\sigma}}\right)\left|f(\vec{n})\right\rangle\ket{g(\vec{n})}_{\text{garb}},
\\\nonumber
f(\vec{n})
&=
\sum_{(\vec x, \sigma) \neq (\vec y, \sigma')} V(\vec{x}-\vec{y}) n_{\vec p, \sigma} n_{\vec q, \sigma'},
\end{align} 
where the value $f(\vec{n})$ is encoded in $l\in{\mathcal{O}}(\log(1/\epsilon))$ bits. This is implemented by the following sequence, where we have omitted the garbage register for clarity.
\begin{align}
\left(\bigotimes_{\vec{x},\sigma}\ket{n_{\vec{x}\sigma}}\right)\ket{0}
\underset{\operatorname{ADD}}{\rightarrow}
\bigotimes_{\vec{x}}\ket{\sum_{\sigma}n_{\vec{x}\sigma}}
\underset{\operatorname{FFT}}{\rightarrow}
\bigotimes_{\vec{k}}\ket{\tilde{\chi}_{\vec{k}}}
\underset{\operatorname{|\cdot|^2}}{\rightarrow}
\bigotimes_{\vec{k}}\ket{|\tilde{\chi}_{\vec{k}}|^2}
\underset{\operatorname{\times V_k}}{\rightarrow}
\bigotimes_{\vec{k}}\ket{{V}({\vec{k}})|\tilde{\chi}_{\vec{k}}|^2},
\end{align}
The cost of $O_{V,\text{binary}}$ may be expressed in term of the four standard reversible arithmetical operations, addition, subtraction, division, and multiplication, which each cost $\mathcal{O}(\operatorname{poly}(l))$ primitive gates. 
The first steps $\operatorname{ADD}$ adds $\mathcal{O}(N)$ pairs of two bits $n_{\vec x, \sigma = 1} + n_{\vec x, \sigma = -1}$ and costs $\mathcal{O}(N)$ arithmetic operations. The second step $\operatorname{FFT}$ is a $d$-dimensional Fast-Fourier-Transform on $\mathcal{O}(N)$ binary numbers and requires $\mathcal{O}(N\log{(N)})$ arithmetic operations. The third step computes the absolute-value-squared of $\mathcal{O}(N)$ binary numbers, and uses $\mathcal{O}(N)$ arithmetic operations. The last step multiplies each $|\tilde{\chi}_k|^2$ with the corresponding $V_k$, and costs $\mathcal{O}(N)$ arithmetic operations. This last step may actually be avoided by rescaling the time parameter $e^{-iA_j t}\rightarrow e^{-iA_j V_k t}$ in~\cref{eq:phase_gadget}. Thus the total cost of $O_{V,\text{binary}}$ is $\mathcal{O}(N\log{(N)} \operatorname{poly}(l)) = \mathcal{O}(N\log{(N)}\operatorname{polylog}(1/\epsilon))$. Using one query to $O_{V,\text{binary}}$, $O^\dag_{V,\text{binary}}$, and $\mathcal{O}(N\log{(N)}\operatorname{polylog}(1/\epsilon))$ primitive quantum gates, we may thus apply $e^{-iV t}$ with a phase error $\mathcal{O}(\epsilon)$ for a fixed value of $t$. 
 
\subsection{Application to quantum chemistry in the plane-wave basis}
\label{sec:Chemistry}
The Hamiltonian that generates time-evolution for a state $\ket{\psi(t)}$ of interacting electrons in $d=3$ dimension consists of three operators: the electron kinetic energy $T$, the electron-nuclei potential energy $U$, and the electron-electron potential energy $V$. It was demonstrated by~\cite{babbush2017low} that this electronic structure Hamiltonian is a special case of the general Hubbard Hamiltonian of~\cref{eq:Hubbard_Fourier}. In the plane-wave basis,
\begin{align}
i\partial _t \ket{\psi(t)} &= H \ket{\psi(t)}
\\\nonumber
H_{P} &= \frac{1}{2} \sum_{\vec{p}, \sigma} |\vec{k}_{\vec p}|^2 \, c_{\vec{p},\sigma}^\dagger c_{\vec{p},\sigma} + \frac{4 \pi}{\Omega} \sum_{\substack{\vec{p} \neq \vec{q} \\ j,\sigma}} \left(-\zeta_j \frac{e^{i \, \vec{k}_{\vec q-\vec p} \cdot \vec R_j}}{|\vec k_{\vec p-\vec q}|^2}\right) c^\dagger_{\vec p, \sigma} c_{\vec q, \sigma} + \frac{2 \pi}{\Omega} \sum_{\substack{(\vec p, \sigma) \neq (\vec q, \sigma') \\ \vec\nu \neq 0}} \frac{c^\dagger_{\vec p,\sigma} c_{\vec q,\sigma'}^\dagger c_{\vec q + \vec \nu,\sigma'} c_{\vec p - \vec \nu,\sigma}}{|\vec k_{\vec \nu}|^2},
\end{align}
where $k_{\vec{p}}=2\pi \vec{p}/\Omega^{1/3}$, $\vec{p}\in [-N^{1/3},N^{1/3}]^3$,
$r_{\vec{p}}=\vec{p} (\Omega / N)^{1/3}$, $\Omega$ represents the volume of the simulation, and $\eta_j$ is the nuclear charge of the $j^{\text{th}}$ nucleus. Whereas $T$ is diagonal here, one may find an alternate basis where $U$ and $V$ are diagonal. This is the dual basis, defined through the unitary transform $\operatorname{FFFT}$ of~\cref{eq:fermionic_fourier_basis}. In this basis, let us define the state $\ket{\psi_D(t)}=\operatorname{FFFT}^\dag\ket{\psi(t)}$, which evolves under the Hamiltonian $H$, which is of exactly that of~\cref{eq:Hubbard_Fourier}, with coefficients
\begin{align}
\tilde{T}(\vec{p})=\frac{1}{2}|\vec{k}_{\vec p}|^2
\quad
U(\vec{p}) = -\frac{4 \pi}{\Omega }\sum_{\substack{\vec{\nu}\neq 0,\\j}}\frac{\zeta_j \, \cos\left[ \vec k_{\vec \nu} \cdot \left(\vec R_{\vec j} - \vec r_{\vec p}\right)\right]}{|\vec k_{\vec \nu}|^2},
\quad
V(\vec{s}) = \frac{2 \pi}{\Omega }\sum_{\vec{\nu}\neq 0}\frac{\cos \left[\vec k_{\vec\nu} \cdot \vec r_{\vec{s}}\right]}{|\vec k_{\vec\nu}|^2}.
\end{align}

Thus the cost of time-evolution by the electronic structure Hamiltonian $e^{-iHt}$ using the interaction picture is given by~\cref{eq:Int_pic_sim_Hubbard_best_cost}. The only dominant parameter that depends on the problem is the normalization factor 
\begin{align}
\alpha_T = \sum_{\vec{p},\sigma}\frac{|\vec{k}_{\vec p}|^2}{2}
\in{\mathcal{O}}\left(\int^{N^{1/3}}_{0}\frac{p^2}{\Omega^{2/3}}(4\pi p^2) dp\right)
=
\mathcal{O}\left(\frac{N^{5/3}}{\Omega^{2/3}}\right).
\end{align}
The spectral norms of the potential terms are bounded by $\|U+V\|\in\mathcal{O}\left(\frac{\|\zeta\|_1N^{5/3}}{\Omega^{1/3}}+\frac{N^{7/3}}{\Omega^{1/3}}\right)$, where $\|\zeta\|_1=\sum_{j}|\zeta_j|$~\cite{babbush2017low} -- the exact scaling is unimportant, so long as it is polynomial. Thus the total gate complexity of time-evolution under the assumption of constant density (i.e. $N/\Omega \in O(1)$) is
\begin{align}
C_{\text{TDS}}\in{\mathcal{O}}\left(\frac{N^{8/3}}{\Omega^{2/3}} t\operatorname{polylog}\left(\frac{(\|\zeta\|_1+N)Nt}{\epsilon}\right)\right)=\mathcal{O}\left(N^2 t\operatorname{polylog}\left(\|\zeta\|_1Nt\epsilon\right)\right).
\end{align}
In contrast, the cost of simulation in the plane-wave dual basis~\cite{babbush2017low} applies the `Qubitization' technique~\cite{Low2016qubitization} and has gate complexity that scales like $\tilde{\mathcal{O}}((\|\zeta\|_1+N)N^{8/3}t)$, and also has a polynomial dependence on  the nuclear charges $\zeta_j$. Our method outperforms this, and notably depends only poly-logarithmically on the sum of the nuclear charges.  This means that the cost of the  simulation method is largely insensitive to the nuclear charges present in the simulation, unlike Trotter-Suzuki simulation methods which sensitively depend on the nuclear charges~\cite{babbush2015chemical}.

\section{Application to sparse Hamiltonian simulation}
\label{sec:sparse_Hamiltonian_simulation}
In this section, we present a complexity-theoretic perspective of the improvements that are enabled by simulation with the truncated Dyson series, and simulation in the interaction picture. We do so by evaluating the query complexity for the simulation of sparse Hamiltonians $H$. Such Hamiltonians of dimension $N$ are called $d$-sparse if there are at most $d\in{\mathcal{O}}(\operatorname{polylog}(N))$ non-zero entries in every row, and the position and values of these entries may be efficiently output, in say a binary representation, by some classical circuit of size $\mathcal{O}(\operatorname{polylog}(N))$. This abstract model is useful in quantum complexity theory as a natural generalization these classical circuits leads to unitary quantum oracles that can be queried to access the same information, but now in superposition. With this model, we achieve in~\cref{sec:spase_ham_sim_time_dependent} time-dependent simulation with a square-root improvement with respect to the sparsity parameter, and gate complexity scaling with the average instead of worst-case rate-of-change $\|\dot H\|$. By moving to the interaction picture in~\cref{sec:sparse_diagonally_dominant}, we find more efficient time-independent simulation algorithms for diagonally-dominant Hamiltonians.

This model assumes that the Hamiltonian is input to the simulation routine through two oracles: $O_f$ and $O_H$.  $O_H$ is straight forward; it provides the values of the matrix elements of the Hamiltonian given a time index $\ket{t}$ and indices $\ket{x},\ket{y}$ to for the row and column of $H$ as follows
\begin{equation}
\label{eq:sparse_oracle_H}
O_H\ket{t,x,y,0} = \ket{t,x,y,H_{xy}(t)}.
\end{equation}
$O_f$ provides the locations of the non-zero matrix elements in any given row or column of $H$.  Specifically, let $f(x,j)$ give the column index of the $j^{\rm th}$ non-zero matrix element in row $x$ if it exists and an appropriately chosen zero element if it does not.  In particular, let $r_{t,j}$ be the list of column indices of these non-zero matrix elements in row $j$.  We then define, with a time-index $t$,
\begin{equation}
\label{eq:sparse_oracle_f}
O_f \ket{t,x,j} =\ket{t,x,f_t(x,j)}.
\end{equation}

\subsection{Simulation of sparse time-dependent Hamiltonians}
\label{sec:spase_ham_sim_time_dependent}
Applying the truncated Dyson series simulation algorithm to sparse Hamiltonians requires us to synthesize $\operatorname{HAM-T}$ in~\cref{def:HAM-T} from these oracles $O_H, O_F$. This is possible by a straightforward construction. 
\begin{lemma}[Synthesis of $\operatorname{HAM-T}$ from sparse Hamiltonian oracles]
	\label{thm:sparse_Ham-t}
	Let a time-dependent $d$-sparse Hamiltonian $H(s) : [0,t] \rightarrow \mathbb{C}^{N\times N}$ be be encoded in the oracles $O_H$ and $O_f$ from~\cref{eq:sparse_oracle_H,eq:sparse_oracle_f} to $n_p$ bits of precision. Then the unitary $\operatorname{HAM-T}$ such that for $\|H\|_{\max} := \max_s \|H(s)\|_{\max}$
	\begin{align}
	\label{eq:sparse_Ham-T}
	(\bra{0}_a\otimes \openone_s) \operatorname{HAM-T} (\ket{0}_a\otimes \openone_s) = \sum_{t}\ket{t}\bra{t}_d \otimes \frac{H(t)}{d\|H\|_{\max}},
	\end{align}
	can be implemented  with $O(1)$ queries to $O_f$ and $O_H$, and $\mathcal{O}\left(\operatorname{poly}(n_p)+\log{(N)}\right)$ primitive gates.
\end{lemma}
\begin{proof}
	The proof closely mimics the discussion by~\cite{Low2017USA}, but we formally restate the result here to make it manifestly applicable to the time-dependent case.  Let $U_{\rm col},U_{\rm row}$ be the following unitary transformations
	\begin{align}
	\label{eq:sparse_states}
	U_{\rm col} \ket{t}_d\ket{k}_{s}\ket{0}_{a} &:= \ket{t}_d\ket{\chi_{k}(t)}= \frac{1}{\sqrt{d}} \sum_{p\in r_k} \ket{t}_d\ket{k}_s\ket{p}_{a_1} \left(\sqrt{\frac{H_{p,k}(t)}{\|H\|_{\max}}}\ket{0}_{a_2} + \sqrt{1-\frac{|H_{k,p}(t)|}{\|H(t)\|_{\max}}}\ket{1}_{a_2} \right),
	\\\nonumber
	\bra{0}_{a}\bra{j}_s\bra{t}_d U_{\rm row}^\dagger  &:= \bra{\bar{\chi}_{j}(t)}\bra{t}_d= \frac{1}{\sqrt{d}} \sum_{q\in r_j}  \left(\sqrt{\frac{H_{j,q}(t)}{\|H\|_{\max}}}\bra{0}_{a_2} + \sqrt{1-\frac{|H_{q,j}(t)|}{\|H\|_{\max}}}\bra{2}_{a_2} \right)\bra{j}_{a_1}\bra{q}_{s}\bra{t}_d,
	\\\nonumber
	\braket{\bar{\chi}_{j}(t)}{\chi_{k}(t)}&=\frac{H_{j,k}}{d\|H\|_{\rm max}}.
	\end{align}
	Let $\ket{\psi} = \sum_{t,k} a_{t,k} \ket{t}_d\ket{k}_s$.  We then have that
	\begin{align}
	[\ketbra{0}{0}_{a}\otimes \openone] U_{\rm row}^\dagger \cdot U_{\rm col} \ket{t}_d\ket{\psi}\ket{0}_{a} 
	&= \ket{0}\bra{0}_{a}\otimes\sum_{t',j}(\ket{t'}\bra{t'}_d\otimes \ket{j}\bra{j}_{s} )\sum_{t,k} a_{t,k}U^{\dagger}_{\rm row} \cdot U_{\rm col}\ket{t}_d \ket{k}_s\ket{0}_{a}.
	\\\nonumber
	&=\ket{0}_{a}\sum_{t',j}\ket{t'}_d\ket{j}_s \sum_{t,k} a_{t,k}(\bra{\bar{\chi}_{j}(t')}\bra{t'}_d)(\ket{t}_d  \ket{\chi_{k}(t)})
	\\\nonumber
	&= \sum_{j,k} a_{t,k} \ket{0}_a\ket{t}_d\ket{j}_s \frac{H_{j,k}(t)}{d\|H\|_{\max}}= \frac{\ket{0}_sH(t)\ket{\psi}}{d\|H\|_{\max}} .
	\end{align}
	As this result holds for any input state $\ket{\psi}$, the choice $\operatorname{HAM-T}=U_{\rm row}^\dagger\cdot U_{\rm col}$ satisfies~\cref{eq:sparse_Ham-T}.

	The query cost then follows from the fact that $U_{\rm col}$ and $U_{\rm row}^\dagger$ can be implemented using $O(1)$ calls to $O_f$.  In particular, $U_{\rm col}$ can be prepared in the following steps:
	\begin{align}
	\ket{t}\ket{k}\ket{0} &\mapsto \ket{t}\ket{k} \frac{1}{\sqrt{d}}\sum_{\ell=1}^d\ket{\ell} \ket{0}\nonumber\\
	&\stackrel{\mapsto}{O_f}\ket{t}\ket{k} \frac{1}{\sqrt{d}}\sum_{p\in r_k}\ket{p} \ket{0}\nonumber\\
	&\stackrel{\mapsto}{O_H}\ket{t}\ket{k} \frac{1}{\sqrt{d}}\sum_{p\in r_k}\ket{p}\ket{H_{k,p}(t)} \ket{0}\nonumber\\
	&~\stackrel{\mapsto}{}\ket{t}\ket{k} \frac{1}{\sqrt{d}}\sum_{p\in r_k}\ket{p}\ket{H_{k,p}(t)}\left(\sqrt{\frac{H^{*}_{k,p}(t)}{\|H\|_{\max}}}\ket{0} + \sqrt{1-\frac{|H_{k,p}(t)|}{\|H\|_{\max}}}\ket{1} \right) \nonumber\\
	&\stackrel{\mapsto}{O_H^{-1}} \ket{t}\ket{k} \frac{1}{\sqrt{d}}\sum_{p\in r_k}\ket{p}\left(\sqrt{\frac{H_{p,k}(t)}{\|H\|_{\max}}}\ket{0} + \sqrt{1-\frac{|H_{p,k}(t)|}{\|H\|_{\max}}}\ket{1} \right)\ket{0}= U_{\rm col}\ket{t}\ket{k} \ket{0}.
	\end{align}
	Therefore accessing $U_{\rm col}$ unitaries requires $O(1)$ queries to the fundamental oracles as claimed, along with a arithmetic circuit, of size polynomial in the number of bits used to represent $\ket{H_{k,p}(t)}$, for computing trigonometric functions of the magnitudes of the complex-valued matrix elements as well as their arguments.  The argument that $U_{\rm row}^\dagger$ requires $O(1)$ queries follows in exactly the same manner, but with an additional final step that swaps the $s$ and $a_1$ registers.
	\end{proof}

Once  $\operatorname{HAM-T}$, is obtained, the complexity of simulation follows directly from previous results.
\begin{theorem}[Simulation of sparse time-dependent Hamiltonians]
	\label{Cor:sparse_time_dependent_simulation}
	Let a time-dependent $d$-sparse Hamiltonian $H(s) : [0,t] \rightarrow \mathbb{C}^{N\times N}$  have average rate-of-change $\langle\|\dot H\|\rangle=\frac{1}{t}\int^t_{0} \left\|\frac{\mathrm{d} H(s)}{ \mathrm{d} s}\right\| \mathrm{d}s$, and be encoded in the oracles $O_H$ and $O_f$ from~\cref{eq:sparse_oracle_H,eq:sparse_oracle_f} to $n_p$ bits of precision. Then the time-ordered evolution operator  $\mathcal{T}\left[ e^{-i\int_0^t H(s) \mathrm{d} s}\right]$ may be approximated with error $\epsilon$ and probability of failure $\mathcal{O}(\epsilon)$ using 
	\begin{enumerate}
		\item $\text{Queries}$ to $O_H$ and $O_f$:$\;\mathcal{O}\left(d\|H\|_{\max}t\frac{\log{(d\|H\|_{\max}t/\epsilon)}}{\log\log{(d\|H\|_{\max}t/\epsilon)}}\right)$.
		\item $\text{Qubits}:\;\mathcal{O}\left(n_p+\log{\left(\frac{Nt}{d\|H\|_{\max}\epsilon}\left({\langle\|\dot{H}\|\rangle}+{[d\|H\|_{\max}]^2}\right)\right)}\right)$.
		\item $\text{Primitive gates}:\;\mathcal{O}\left(d\|H\|_{\max}t\left(\operatorname{poly}(n_p)+ \log{\left(\frac{Nt}{d\|H\|_{\max}\epsilon}\left({\langle\|\dot{H}\|\rangle}+{[d\|H\|_{\max}]^2}\right)\right)} \frac{\log{(d\|H\|_{\max}t/\epsilon)}}{\log\log{(d\|H\|_{\max}t/\epsilon)}}\right)\right)$.
	\end{enumerate}
\end{theorem} 
\begin{proof}
From~\cref{Thm:Compressed_TDS}, the number of qubits required is $\mathcal{O}(n_s+n_a+n_d+\log\log{(1/\epsilon)})$. Values for these parameters are obtained from the construction of $\operatorname{HAM-T}$ in~\cref{thm:sparse_Ham-t}, which also requires an additional $n_p$ qubits for the bits of precision to which $H$ is encoded. In this construction, $n_a\in{\mathcal{O}}\left(n_s\right)\in{\mathcal{O}}\left(\log{(N)}\right)$. $n_d$ is obtained from the number of time-discretization points required by~\cref{Thm:Compressed_TDS}. Simulation for time $t$ is implemented by simulating segments of duration $\Theta{(\alpha^{-1})}$. As there are $\Theta(t\alpha)$ segments, we rescale $\epsilon\rightarrow \Theta(\epsilon/(t\alpha))$.
\end{proof}
Compared to prior art~\cite{Berry2014Exponential}, this is a quadratic improvement in sparsity $d$. Furthermore, instead of scaling with the worst-case $\mathcal{O}\left(\log{(\max_s\|\dot{H}(s)\|)}\right)$, we obtain scaling with average rate-of-change $\langle\|\dot H\|\rangle$. If we further assume that the computed matrix elements $H_{j,k}$ are not exact, the number of bit of precision scales like $n_p\in{\mathcal{O}}\left(\log{(\|H\|t/\epsilon)}\right)$~\cite{Berry2015Hamiltonian}. Note several generic improvements to~\cref{Cor:sparse_time_dependent_simulation} are possible, but will not be pursued further as they are straightforward. For instance, if $\|H(t)\|_{\rm max}$ as a function of time is known, we may use step sizes of varying size by encoding each segment $t\in[t_{j},t_{j+1}]$ with the largest $\max_{t\in[t_{j},t_{j+1}]}\|H(t)\|_{\rm max}$, rather than the worst-case $\max_{t}\|H(t)\|_{\rm max}$. 

\subsection{Simulation of sparse time-independent Hamiltonians in the interaction picture}
\label{sec:sparse_diagonally_dominant}
We now turn our attention to time-independent $d$-sparse Hamiltonians $H=A+B$ where $A$ is diagonal and $B_{k,k}=0$ for all $k$. In particular, we consider the case of diagonally dominant Hamiltonians, where $\|A\| \ge d \|B\|_{\rm max}$. Given norms for each of these terms $\|A\|$ and $\|B\|_{\rm max}$, it is straightforward to simulate time-evolution $e^{-iHt}$ in the Schr\"odinger picture. For instance, using the truncated Taylor series approach in~\cref{eq:sch_pic_sim_cost}, one obtains a query complexity of $\mathcal{O}\left(t(d\|B\|_{\rm max}+\|A\|)\operatorname{polylog}(t, d,\|A\|,\|B\|_{\rm max},\epsilon)\right)$. By instead simulating $H_I(t)=e^{iAt}Be^{-iAt}$ in the interaction picture, the dependence on $\|A\|$ can be removed, which is particularly useful in cases of strong diagonal dominance $\|A\| \ge d \|B\|_{\rm max}$, of which the Hubbard model with long-ranged interactions in~\cref{sec:Hubbard} is an example. Similar to our results for time-dependent sparse Hamiltonian simulation in~\cref{sec:spase_ham_sim_time_dependent}, this is easily proven by mapping the input oracles $O_H$ and $O_f$ for matrix values and positions to the oracles of~\cref{thm:int_pic_sim} for the more general result.
\begin{theorem}[Simulation of sparse diagonally dominant Hamiltonians]
	\label{thm:sparse_diagonal_dominant}
	Let a time-independent $d$-sparse Hamiltonian $H=A+B\in\mathbb{C}^{N\times N}$ be encoded in the oracles $O_H$ and $O_f$ from~\cref{eq:sparse_oracle_H,eq:sparse_oracle_f} to $n_p$ bits of precision, and be characterized by spectral norm spectral norm $\|A\|$ for the diagonal component and max-norm $\|B\|_{\rm max}$ for the off-diagonal component. Let $\alpha_B = d\|B\|_{\rm max}$. Then the time-evolution operator $e^{-iHt}$ may be approximated with error $\epsilon$ using
\begin{enumerate}
	\item Queries to $O_H$ and $O_f$: $\mathcal{O}\left(\alpha_B t\frac{\log{(\alpha_B t/\epsilon)}}{\log\log{(\alpha_B t /\epsilon)}}\right)$.
	\item Qubits: $\mathcal{O}\left(n_p+\log{(N)}+\log{\left(\frac{ t}{\epsilon}\left({\|A\|}+\alpha_B\right)\right)}\right)$.
	\item Primitive gates: $\mathcal{O}\left(\alpha_B t\left(\log{(N)} + \operatorname{poly}(n_p)\log{\left(\frac{ t}{\epsilon}\left({\|A\|}+\alpha_B\right)\right)}\right)\frac{\log{(\alpha_B t/\epsilon)}}{\log\log{(\alpha_B t/\epsilon)}}\right)$.
\end{enumerate}

\end{theorem}
\begin{proof}
	This follows immediately by combining the query complexity of~\cref{thm:int_pic_sim_query} to $e^{-iA t}$ and $\operatorname{HAM-T}$ that encodes the Hamiltonian $H_I(t)=e^{iAt}Be^{-iAt}$ in the rotating frame, with the query complexity of the approach in~\cref{thm:int_pic_sim} for synthesizing these oracles using the input oracles $O_H$ and $O_f$. One possible decomposition of $\operatorname{HAM-T}$ is	
	\begin{align}
	\operatorname{HAM-T}&= \left(\sum^{M-1}_{m=0}\ket{m}\bra{m}_d \otimes \openone_a\otimes  e^{iA\tau m/M}\right)(\openone_d \otimes O_B)\left(\sum^{M-1}_{m=0}\ket{m}\bra{m}_d \otimes \openone_a\otimes e^{-iA\tau m/M}\right),
	\\\nonumber
	(\bra{0}_a\otimes \openone_s) O_B (\ket{0}_a\otimes \openone_s) & = \frac{B}{\alpha_B},
	\end{align}
	where $\alpha_B = d\|B\|_{\rm max}$, $\tau \in{\mathcal{O}}(\alpha_B^{-1})$ and $M\in{\mathcal{O}}\left( \frac{ t}{\epsilon}\left({\|A\|}+\alpha_B\right)\right)$. Note that with this construction, $(\bra{0}_a\otimes \openone_s) \operatorname{HAM-T} (\ket{0}_a\otimes \openone_s) = \sum_{m=0}^{M-1}\ket{m}\bra{m}_d \otimes \frac{H_I(\tau m/M)}{d\|B\|_{\max}}$. 

First, let us synthesize $\left(\sum^{M-1}_{m=0}\ket{m}\bra{m}_d \otimes \openone_a\otimes  e^{iA\tau m/M}\right)$ using $\mathcal{O}(1)$ queries to the input oracles $O_H$, and $\mathcal{O}\left(\log{(N)}+n_p \log{(M)}\right)$ primitive gates. Since $A$ is diagonal, $e^{-i A t}$ can be simulated for any $t>0$ using only two queries. This is implemented by the following steps:
 \begin{align}
 \ket{k}\ket{0}\ket{0}\ket{0} &
 \mapsto  \ket{k}\ket{k}\ket{0}\ket{0}
 \\\nonumber
 &\underset{O_H}{\mapsto} \ket{k}\ket{k}\ket{H_{k,k}}\ket{0}
 \\\nonumber
 &\mapsto  \ket{k}\ket{k}\ket{H_{k,k}}e^{-i H_{k,k} Z t}\ket{0}
 = e^{-i H_{k,k}t}\ket{k}\ket{k}\ket{H_{k,k}}\ket{0}
 \\\nonumber
 &\underset{O^\dag_H}{\mapsto} e^{-i H_{k,k}t}\ket{k}\ket{k}\ket{0}\ket{0}
 \\\nonumber
 &\rightarrow e^{-i H_{k,k}t}\ket{k}\ket{0}\ket{0}\ket{0}.
 \end{align}
 Step one uses $n_s\in{\mathcal{O}}(\log{(N)})$ $\operatorname{CNOT}$ gates to copy the computational basis state $\ket{k}$. Step three applies $\mathcal{O}(n_p)$ phase rotation with angle controlled by the bits of $\ket{H_{k,k}}$, and the value of $t$, which is given beforehand. Subsequently, $\left(\sum^{M-1}_{m=0}\ket{m}\bra{m}_d \otimes \openone_a\otimes  e^{iA\tau m/M}\right)$ may be implemented by a sequence of rotations with angles increasing in a geometric series, and each controlled by a different qubit in the $d$ register, e.g. controlled-$e^{iA\tau 2^{-1}/M},e^{iA\tau 2^{-2}/M}, e^{iA\tau 2^{-n+p}/M}$. Naively, this requires $\mathcal{O}(\log{(M)})$ queries. However, it is only necessary to compute $\ket{H_{k,k}}$ once as the entire sequence of controlled-phases may be applied after step three. Similarly, it is only necessary to copy the computational basis state $\mathcal{O}(1)$ times.

Second, let us synthesize $O_B$ using $\mathcal{O}(1)$ queries to the input oracles $O_H$ and $O_f$. How this is done should be clear from~\cref{thm:sparse_Ham-t}, by omitting the time-index, and preparing the state in~\cref{eq:sparse_states} when the input indices $k=p$. This has gate complexity $\mathcal{O}(\operatorname{poly}(n_p)+\log{(N)})$.
Thus $e^{-iA \tau}$  and $\operatorname{HAM-T}$ combined have query complexity $\mathcal{O}(1)$ to $O_H$ and $O_f$, and gate complexity $\mathcal{O}(\operatorname{poly}(n_p)\log{(M)}+\log{(N)})$. By substituting into~\cref{thm:int_pic_sim_query}, we obtain the stated results.
\end{proof}

This provides a formal proof that the query complexity of simulating a Hamiltonian, within the interaction picture, is independent of the magnitude of the diagonal elements of the Hamiltonian, up to logarithmic factors.  

\section{Conclusion}
We demonstrate in this work that simulating quantum dynamics within the interaction picture rather than the Schr\"{o}dinger picture can be advantageous. This requires a good time-dependent simulation algorithm -- our formulation of the truncated Dyson series simulation algorithm is rigorous and achieves space savings over the original proposal~\cite{berry2015simulating}. When applied to simulating Hubbard models with long-ranged interactions as well as quantum chemistry within a plane-wave basis, we find that the gate complexity scales with time $t$ as $\widetilde{\mathcal{O}}({N^2t})$ for systems with $N$ sites, assuming that kinetic energy is an extensive property. In the black-box model of time-dependent sparse Hamiltonians, we find that simulation in the truncated Dyson series framework generally reduces the query complexity from $\widetilde{\mathcal{O}}(d^2)$ to $\widetilde{\mathcal{O}}(d)$. Combined with the interaction picture, this also reduces the scaling with respect to the magnitude of diagonal components from linear to logarithmic, which is particularly relevant to the common case of diagonally dominant Hamiltonians.

Some straightforward extensions of our work on time-dependent simulation are possible. For instance, the step-size of our algorithm may be adaptively chosen to scale with the worst-case norm of the Hamiltonian in each segment rather than with worst-case across all segments. Furthermore, the query complexity $\tilde{\mathcal{O}}(d\|H\|_{\rm max})$ of our time-dependent sparse Hamiltonian simulation algorithm may be easily improved in combination with~\cite{Low2017USA} to scale like $\tilde{\mathcal{O}}\left(\sqrt{d \|H\|_\text{max}\|H\|_1}\right)$ if the induced one-norm $\|H\|_1$ of the Hamiltonian is also known beforehand. Our technique of interaction picture simulation is also applicable to many other quantum systems, particularly quantum field theories. Identifying other such physical systems would be of great interest.

More generally, the complexity of time-independent quantum simulation for generic Hamiltonians, given minimal information, appears to be nearly resolved~\cite{Low2016HamSim}. Thus future advancements, such as this work, will likely focus on exploiting the detailed structure of Hamiltonians of interest. The promise of results in similar directions is exemplified by recent work that exploit the geometric locality of interactions~\cite{Haah2018quantum}, or the sizes of different terms in a Hamiltonian~\cite{Hadfield2018divideandconquer}. The challenge will be finding characterizations of Hamiltonians that are sufficiently specific so as to enable a speedup, yet sufficiently general so as to include problems of practical and scientific value.

\section{Acknowledgments}
We thank Isaac Chuang, Jeongwan Haah, Robin Kothari, and Matthias Troyer for insightful discussions.  We further thank Dominic Berry for pointing out important typographic errors in a previous version of the manuscript.

\bibliography{biblio}

\appendix

\section{Error Estimates for Truncated Dyson series}
\label{sec:proof_truncation_discretization_dyson}
In this section, we complete the proof of~\cref{Thm:Truncated_Dyson_Algorithm} for the error from truncating the Dyson series at order $K$, and the error from approximating its terms, which are time-ordered integrals, with Riemann sums. These results provide a rigorous upper bound on the error of time-dependent Hamiltonian simulation. Let $D_k$ be the $k^{\text{th}}$ term in the Dyson expansion, and let $B_k$ be the Riemann sum of $D_k$ with each dimension discretized into $M=t/\Delta$ segments.
\begin{align}
\mathcal{T}\left[e^{-i\int_0^t H(s) \mathrm{d}s}\right]&= \sum^\infty_{k=0} (-i)^k D_k =\lim_{M\rightarrow\infty}  \sum^\infty_{k=0} \frac{(-it)^k }{M^k}B_k,
\\\nonumber
D_k&:=\frac{1}{k!}\int_0^t\cdots\int_0^t \mathcal{T}\left[H(t_1)\cdots H(t_k)\right]\mathrm{d}^k t,
\\\nonumber
B_k&:=\sum_{0\le m_k <\cdots < m_1 < M}H\left({m_k \Delta}\right)\cdots H\left({m_1 \Delta}\right).
\end{align}

We now prove bounds on the error $\epsilon_1$ due to truncating the Dyson series at order $K$.
\begin{lemma}
	\label{Lem:Truncated_Dyson_Error}
	Let $H(s):\mathbb{R} \mapsto \mathbb{C}^{N\times N}$ be differentiable on the domain $[0,t]$.  For any $\epsilon_1\in[0,2^{-e}]$, an approximation to the time ordered operator exponential of $-iH(s)$ can be constructed such that
	$$
	\left\|\mathcal{T}\left[e^{-i\int_0^t H(s) \mathrm{d}s}\right]- \sum^K_{k=0} (-i)^k D_k\right\|\le \epsilon_1,
	$$
	if we take
	\begin{enumerate}
		\item $K\ge \left\lceil-1+\frac{2\ln(1/\epsilon_1)}{\ln\ln(1/\epsilon_1) +1}\right\rceil$
		\item $\max_s\|H(s)\|t\le {\ln 2}$
	\end{enumerate}
\end{lemma}
\begin{proof}
	We start by bounding $\|D_k\|$.
	\begin{align}
	\|D_k\| &= \frac{1}{k!}\left\|\int_0^t\cdots\int_0^t \mathcal{T}\left[H(t_1)\cdots H(t_k)\right]\mathrm{d}^k t\right\|
	\le \frac{1}{k!}\left\|\int_0^t\cdots\int_0^t \prod^k_{j=1}\|H(t_j)\|\mathrm{d}^k t\right\|
	\le \frac{(t\max_s\|H(s)\|)^k}{k!}.
	\end{align}
	At this point, the proof is identical to the time-independent case as $\max_s\|H(s)\|$ is independent of time. Thus using Stirling's approximation and assuming $K\ge 2\max_s\|H(s)\||t|$,
	\begin{align}\nonumber
	\epsilon_1&=\left\|\mathcal{T}\left[e^{-i\int_0^t H(s) \mathrm{d}s}\right]- \sum^K_{k=0} (-i)^k D_k\right\|\le \sum^\infty_{k=K+1} \|D_k\|
	\le \sum^\infty_{k=K+1}\frac{(t\max_s\|H(s)\|)^k}{k!}
	\\\nonumber&
	\le \frac{(t\max_s\|H(s)\|)^{K+1}}{(K+1)!}\sum^{\infty}_{k=K+2}\left(1/2\right)^{k-K-1}
	=\frac{(t\max_s\|H(s)\|)^{K+1}}{(K+1)!}
	\\
	&\le\left(\frac{te\max_s\|H(s)\|}{K+1}\right)^{K+1}
	\end{align}
	Now we find that this in turn is less than $\epsilon_1$ if $ \max_s\|H(s)\|t e < \min\{\ln(1/\epsilon_1),e\ln2\}\le 1$ given that $\epsilon_1\le 2^{-e}$ and
	\begin{align}
	K \ge \max\left\{-1+\frac{\ln(1/\epsilon_1)}{W\left(\frac{\ln(1/\epsilon_1)}{\max_s \|H(s)\| t e} \right)},2\max_s\|H(s)\||t|\right\},
	\end{align}
	where $W$ is the Lambert-W function.  Using the fact that for $x\ge 1$, $W(x) \ge (\ln(x)+1)/2$ and $\ln(e\ln 2) <1$ we obtain the simpler bound
	\begin{equation}
	K = \left\lceil-1 + \frac{2\ln(1/\epsilon_1)}{\ln\ln(1/\epsilon_1) +1}\right\rceil\in{\mathcal{O}}\left(\frac{\ln(1/\epsilon_1)}{\ln\ln(1/\epsilon_1) }\right).\label{eq:qprimebd}
	\end{equation}
	
\end{proof}

We now prove bounds on the error $\epsilon_2$ from approximating the Dyson series with its Riemann sum.
\begin{lemma}
	\label{Lem:Riemann_of_Dyson}
	Let $H(s):\mathbb{R} \mapsto \mathbb{C}^{N\times N}$ be differentiable on the domain $[0,t]$.  Let us also define the quantities 
	$\langle\|\dot{H}\|\rangle:=\frac{1}{t}\int^t_{0} \left\|\frac{\mathrm{d} H(s)}{ \mathrm{d} s}\right\| \mathrm{d}s$. For integer $K\ge 0$  and $\epsilon_2>0$, 
	$$
	\left\|\sum^K_{k=0} (-i)^k D_k- \sum^K_{k=0} {\left(-i\frac{t}{M}\right)^k }B_k\right\|\le \epsilon_2,
	$$
	by choosing any $M$ such that
	\begin{enumerate}
		\item $M \ge  \frac{t^2}{\epsilon_2} 4e^{\max_{s}\|H(s)\| t} \left(\langle \|\dot{H}\| \rangle  +\max_s \|H(s)\|^2\right)$,
		\item $M \ge K^2$.
	\end{enumerate}
\end{lemma}
\begin{proof}
	We first expand the time-ordered evolution operator using the Dyson series.  We then examine the error incurred in evaluating a given order of the Dyson series for a small hypercubic region of side-length $\Delta=t/M$.  We then upper bound the maximum number of such hypercubes within the allowed volume and use the triangle inequality to argue that the error is the product of the number of such hypercubes and the maximum error per hypercube.
	
	Since $H(s)$ is a differentiable function it holds from Taylor's theorem that for any $\delta\ll 1$ and  computational basis states $\ket{x},\ket{y}$,
	\begin{equation}
	\bra{x} H(s+\delta) \ket{y} = \bra{x} H(s) \ket{y} +\delta \bra{x} \dot{H}(s) \ket{y} +o(\max_s \|\dot{H}(s)\|_{\max}\delta). 
	\end{equation}
	Since computational basis states form a complete orthonormal basis it follows through norm inequalities that
	\begin{equation}
	H(s+\delta) = H(s) +\delta \dot{H}(s) + o(\|\dot{H}(s)\|_{\max}N^2 \delta).
	\end{equation}
	We then have from Taylor's theorem and the triangle inequality that
	\begin{align}
	\|H(s+\delta) - H(s)\| &= \left\| \sum_{j=1}^r [H(s+j\delta/r)-H(s+[j-1]\delta/r)]\right\|\nonumber\\
	&\le \left\| \sum_{j=1}^r \dot{H}(s+[j-1]\delta/r) \delta/r \right\|+ r\left[o(\max_s \|\dot{H}(s)\|_{\max} N^2\delta/r)\right] \nonumber\\
	&\le \int_0^\delta \|\dot{H}(s)\| \mathrm{d}s + r\left[o(\max_s \|\dot{H}(s)\|_{\max} N^2\delta/r)\right].
	\end{align}
	Since this equation holds for all $r$, it also holds in the limit as $r$ approaches infinity.  Therefore
	\begin{equation}
	\label{eq:Error_of_difference_Matrix}
	\| H(s+\delta) - H(s)\| \le \int_0^\delta \|\dot{H}(s)\| \mathrm{d}s.
	\end{equation}
	
	Next, let us consider the error in approximating the integral over a hypercube to lowest order and let us define the hypercube to be $C$ with $x_1,\ldots,x_q$ being the corner of the hypercube with smallest norm.
	First note that in general if $A_j$ and $B_j$ are a sequence of bounded operators and $\|\cdot \|$ is a sub-multiplicative norm then it is straight forward to show using an inductive argument that for all positive integer $q$.
	\begin{equation}
	\left\|\prod_{j=1}^q A_j - \prod_{j=1}^q B_j\right\|\le \sum_{k=1}^q \left( \prod_{j=1}^{k-1} \left\|A_j \right\|\right)\left\| A_k-B_k\right\|\left( \prod_{j=k+1}^q \left\|B_j\right\|\right).
	\end{equation}
	By applying this in combination with~\cref{eq:Error_of_difference_Matrix} to region $C$, the error induced is
	\begin{align}
	&\left\|\int_C H(x_1+y_1)\cdots H(x_q+y_q)\mathrm{d}^q y-{ \Delta^q\prod_{j=1}^q H(x_j)}\right\|\le  \int_C \left\|\prod_{j=1}^q H(x_j +y_j) - \prod_{j=1}^q H(x_j) \right\| \mathrm{d}y^q,\nonumber\\
	&\le \sum_{k=1}^q \left( \prod_{j=1}^{k-1}\int_0^\Delta \left\|H(x_j+s) \right\|\mathrm{d}s\right)\left(\int_0^\Delta\int_0^s \|\dot{H}(x_k+y)\| \mathrm{d}y\mathrm{d}s\right)\left( \prod_{j=k+1}^q \int_0^\Delta\left\|H(x_j)\right\|\mathrm{d}s\right)\nonumber\\
	&\le \sum_{k=1}^q \left( \prod_{j=1}^{k-1}\int_0^\Delta \left\|H(x_j+s) \right\|\mathrm{d}s\right)\left(\int_0^\Delta\int_0^s \|\dot{H}(x_k+y)\| \mathrm{d}y\mathrm{d}s\right)\left( \prod_{j=k+1}^{q}\int_0^\Delta \alpha \;\mathrm{d}s\right) \nonumber\\
	&\le (\alpha)^{q-1}\Delta\sum_{k=1}^q \left( \prod_{j\ne k}^{q}\int_0^\Delta \mathrm{d}s\right)\left(\int_0^\Delta \|\dot{H}(x_k+s)\| \mathrm{d}s\right)\label{eq:hypercubeerr}
	\end{align}
	where $\alpha:=\max_s\left\|H(s) \right\|$.

	There are two regions in the problem.  The first region, which we call the bulk, is the region that satisfies all the constraints of the problem namely ${\rm bulk}:=\{(t_1,\ldots,t_q):\lfloor t_1/\Delta \rfloor > \cdots > \lfloor t_q/\Delta \rfloor\}$. Thus for any index $x_1,\cdots,x_q$ to a hypercube in the bulk, the ordering of terms $H(x_1+t_1)\cdots H(x_q+t_q)$ in the integrand of~\cref{eq:hypercubeerr} is fixed. The second region is called the boundary which is the region in which the hypercubes used in the Riemann sum would stretch outside the allowed region for the integral. 
	Since we approximate the integral to be zero on all hypercubes that intersect the boundary, the maximum error in the approximation is the maximum error that the discrete approximation to the integrand can take within the region scaled to the volume of the corresponding region.

  Finally we have from~\cref{eq:hypercubeerr} that the contribution to the error from integration over the bulk of the simplex is
	\begin{align}
	&\sum_{\substack{\vec{x}\in \{0,\Delta,\cdots,(M-1)\Delta\}^{q}\\x_1<x_2<\cdots< x_q}}
	\left\|\int_C H(x_1+t_1)\cdots H(x_q+t_q)\mathrm{d}^q t-{ \Delta^q\prod_{j=1}^q H(x_j)}\right\|
	\\\nonumber
	&\le \sum_{\substack{\vec{x}\in \{0,\Delta,\cdots,(M-1)\Delta\}^{q}\\x_1<x_2<\cdots< x_q}}(\alpha)^{q-1}\Delta\sum_{k=1}^q \left( \prod_{j\ne k}^{q}\int_0^\Delta \mathrm{d}s\right)\left(\int_0^\Delta \|\dot{H}(x_k+s)\| \mathrm{d}s\right)
	\end{align}
	
In order to understand how the error scales let us examine the partial sum over $x_1$ for fixed $k>1$ is

\begin{align}
&\sum_{x_2,\ldots,x_q}\sum_{x_1=0}^{x_2-1} \int_0^{\Delta} \|\dot{H}(x_1+s)\|\mathrm{d}s \left[(\alpha)^{q-1}\Delta\left( \prod_{j\ne k}^{q}\int_0^\Delta \mathrm{d}s\right)\left(\int_0^\Delta \|\dot{H}(x_k+s)\| \mathrm{d}s\right)\right]\nonumber\\
&\le \sum_{x_2,\ldots,x_q}\int_{0}^{x_2\Delta} \|\dot{H}(s)\|\mathrm{ d}s \left[(\alpha)^{q-1}\Delta\left( \prod_{j\ne k}^{q}\int_0^\Delta \left\|H(x_j+s) \right\|\mathrm{d}s\right)\left(\int_0^\Delta \|\dot{H}(x_k+s)\| \mathrm{d}s\right)\right].
\end{align}
The integral then takes exactly the same form as the original integral and so by repeating the argument $q-1$ times it is easy to see, even in the case where $k=1$, that
\begin{align}
&\sum_{\substack{\vec{x}\in \{0,\Delta,\cdots,(M-1)\Delta\}^{q}\\x_1<x_2<\cdots< x_q}}(\alpha)^{q-1}\Delta\sum_{k=1}^q \left( \prod_{j\ne k}^{q}\int_0^\Delta \mathrm{d}s\right)\left(\int_0^\Delta \|\dot{H}(x_k+s)\| \mathrm{d}s\right)\nonumber\\
& \le (\alpha)^{q-1}\Delta\sum_{k=1}^q \left( \prod_{j\ne k}^{q}\int_0^{x_{j+1}\Delta} \mathrm{d}s\right)\left(\int_0^{x_{k+1}\Delta} \|\dot{H}(s)\| \mathrm{d}s\right)\le\frac{(\alpha t)^{q-1}}{[q-1]!}\Delta\int_0^t \|\dot{H}(s)\|\mathrm{d}s,
\end{align}
where we have used the definition that $x_{q+1}:=M=t/\Delta$ and the fact that $\int_{0}^x \|\dot{H}(s)\|\mathrm{d}s$ is a monotonically increasing function of $x$.  Thus the contribution to the error from the boundary is at most

\begin{equation}
\sum_{q=1}^K \frac{(\alpha t)^{q-1}}{[q-1]!}\Delta \int_0^t \|\dot{H}\|(s)\mathrm{d}s\le \sum_{q=1}^\infty \frac{(\alpha )^{q-1}}{[q-1]!}\Delta \int_0^t \|\dot{H}(s)\|\mathrm{d}s = e^{\alpha t}\Delta \int_0^t \|\dot{H}(s)\|\mathrm{d}s.\label{eq:bulkbd}
\end{equation}

	Next we need to estimate the volume of the boundary.  If a point is on the boundary then by definition there exists at least one $t_j$ such that $\lfloor t_j/\Delta\rfloor = \lfloor t_{j+1} /\Delta \rfloor$ taking $t_0 =t$.  All other values are consistent with points within the bulk.  It is then straight forward to see that (after relabeling the indexes for the summation) that the volume can be expressed as the sum over all possible choices of such sums with at least one matched index.  If we further assume that $q^2\Delta \max_s \|H(s)\|/[\alpha t] \le \ln(2)$ then we have the following upper bound on boundary contribution to the error in the Dyson series:
	\begin{align}
	\int \prod_{j=1}^q \|H(t_q)\| \delta_{t\in {\rm bdy}} \mathrm{d}t^q &\le  \sum_{p=1}^q\Delta^q \max_s\|H(s)\|^p\alpha^{q-p} \binom{q}{p} \sum_{j_1=1}^{t/\Delta}\sum_{j_2=1}^{j_1 -1} \cdots \sum_{j_{q-p}=1}^{j_{q-p-1} -1} 1\nonumber\\
	&= \Delta^q\alpha^q\sum_{p=1}^q \max_s\|H(s)\|^p\alpha^{-p} \binom{q}{p} \binom{t/\Delta}{q-p}.\nonumber\\
	&\le \Delta^q\alpha^q\sum_{p=1}^q \max_s\|H(s)\|^p\alpha^{-p} \frac{(t/\Delta)^q}{q!p!}\left(\frac{q^2\Delta}{t}\right)^p\nonumber\\
	&=\frac{t^q \alpha^q}{q!}\sum_{p=1}^q \frac{1}{p!}\left(\frac{q^2 \Delta \max_s \|H(s)\|}{\alpha t}\right)^p
	\\\nonumber
	&\le\frac{t^q \alpha^q}{q!}\left(\frac{q^2 \Delta \max_s \|H(s)\|}{\alpha t}\right)\sum_{p=0}^\infty \frac{1}{p!}\left(\frac{q^2 \Delta \max_s \|H(s)\|}{\alpha t}\right)^p
	\\\nonumber
	&\le\frac{2q (t\alpha)^{q-1}}{(q-1)!}\left(\Delta \max_s \|H(s)\|\right)
	\label{eq:bdy1term}
	\end{align}
	Note that when $q=1$, there is no boundary contribution. Using the fact that $q/(q-1)\le 2,\; \forall q \ge 2$, this upper bound 
	\begin{align}
	\int \prod_{j=1}^q \|H(t_q)\| \delta_{t\in {\rm bdy}} \mathrm{d}t^q
	\le\frac{4 \Delta \max_s{\|H(s)\|}(\alpha t)^{q-1} }{(q-2)!} 
	\end{align}
	It then follows from summing~\cref{eq:bdy1term} that the error is
	\begin{align}
	\delta_{\text{bdy}}\le&\sum_{q=1}^K \int \prod_{j=1}^q \|H(t_q)\| \delta_{t\in {\rm bdy}} \mathrm{d}t^q
	\le
	\sum_{q=2}^\infty \frac{4 \Delta \max_s{\|H(s)\|}(\alpha t)^{q-1} }{ (q-2)!} 
	\le
	4 \Delta \alpha t\max_s{\|H(s)\|}  e^{\alpha t}.
	\end{align}
	
	By adding this result to that of~\cref{eq:bulkbd}
	\begin{align}
	\sum_{q=1}^K \|D_q - \Delta^qB_q\| 
	\le \delta_{\text{bulk}}+\delta_{\text{bdy}}
	\le 4\Delta t e^{\max_{s}\|H(s)\| t}(\langle\|\dot H\|\rangle + \max_{s}\|H(s)\|^2)
	\end{align}
	It follows from elementary algebra that this total error is at most $\epsilon_2$ if
	\begin{equation}
	\Delta \le \frac{\epsilon_2}{4te^{\max_{s}\|H(s)\| t}[\langle \|\dot{H}\| \rangle  +\max_s \|H(s)\|^2]}.\label{eq:delbd}
	\end{equation}
	Expressed in terms of the number of points $M=\frac{t}{\Delta}$, the total error is at most $\epsilon_2$ if choose any $M$ such that
	\begin{align}
	M &\ge\frac{t^2}{\epsilon_2} 4e^{\max_{s}\|H(s)\| t} [\langle \|\dot{H}\| \rangle  +\max_s \|H(s)\|^2]
	\end{align}
	The final bound on $M$ quoted immediately follows from $\frac{q^2\Delta}{\alpha t} \max_s \|H(s)\| \le \ln(2)
	\Rightarrow
	K \le \sqrt{\ln(2) M}\le \sqrt{M}$.

\end{proof}

Now that we have proved the necessary results regarding the error in the truncated Dyson series simulation, we are now ready to prove~\cref{Thm:Truncated_Dyson_Algorithm}, which we restate for convenience.
\truncatedDysonError*
\begin{proof}
	This is proven by combining two intermediate results using a triangle inequality. The approximation error is upper-bounded by
	\begin{align}\nonumber
	&\left\|\mathcal{T}\left[e^{-i\int_0^t H(s) \mathrm{d}s}\right]- \sum^K_{k=0} \left(-i \frac{t}{M}\right)^k B_k\right\|
	=
	\left\|\mathcal{T}\left[e^{-i\int_0^t H(s) \mathrm{d}s}\right]- \sum^K_{k=0} (-i)^k D_k+\sum^K_{k=0} (-i)^k D_k-\sum^K_{k=0} \left(-i \frac{t}{M}\right)^k B_k\right\|
	\\
	&\le
	\underbrace{\left\|\mathcal{T}\left[e^{-i\int_0^t H(s) \mathrm{d}s}\right]- \sum^K_{k=0} (-i)^k D_k\right\|}_{\epsilon_1}+
	\underbrace{\left\|\sum^K_{k=0} (-i)^k D_k-\sum^K_{k=0} \left(-i \frac{t}{M}\right)^k B_k\right\|}_{\epsilon_2}\le\epsilon.
	\end{align}
	We choose the errors in both cases to obey $\epsilon_1=\epsilon/2$ and $\epsilon_2 = \epsilon/2$.  The result then follows by taking the 
	most restrictive of the assumptions of~\cref{Lem:Truncated_Dyson_Error} and~\cref{Lem:Riemann_of_Dyson}.
\end{proof}

\section{Truncated Dyson series algorithm with low space overhead}
\label{sec:compresson_gadget}
The quantum algorithm described by~\cref{Thm:Compressed_TDS} applies an $\epsilon$-approximation $\tilde{U}=\sum^K_{k=0} \frac{(-it)^k }{M^k}B_k$ of the time-ordered evolution operator $\mathcal{T}\left[e^{-i\int_0^t H(s) \mathrm{d}s}\right]$, where the truncation order $K$ and the number of discretization points $M$ are given by~\cref{Thm:Truncated_Dyson_Algorithm}. 

The simulation algorithm then proceeds in three steps. First, we construct a sequence of $K$ unitary operators $U_1,U_2,\cdots,U_K$, that encode some matrix $H_k$ such that $H_k\cdots H_2H_1\propto B_k$ implements the $k^{\text{th}}$ term of the Dyson series. We call $\operatorname{\operatorname{DYS}_{K}}$ the unitary that applies this sequence of unitaries $U_1\cdots U_k$ controlled on an index state $\ket{k}_b$. A naive implementation of this idea, as worked out in~\cref{sec:dyson_series_alg_duplicated_registers}, requires the $K$-fold duplication of registers $a$ and $d$ of the oracle $\operatorname{HAM-T}$. We present a compression gadget described by~\cref{Thm:compression_gadget} avoids this overhead. Second, we take a linear combination of $B_k$ to apply $\tilde{U}$ with some success probability on any input state to the $s$ register. As $B_K$ contains products of $K$ Hamiltonians, this step requires $K$ queries to $\operatorname{HAM-T}$. Third, since $\tilde{U}$ is $\epsilon$-close to unitary, we apply oblivious amplitude amplification~\cite{berry2015simulating} to boost this probability to $1-\mathcal{O}(\epsilon)$. 

We now prove the compression gadget, which is followed by a proof~\cref{Thm:Compressed_TDS}.
\begin{lemma}[Compression gadget]
	\label{Thm:compression_gadget}
	Let $\{U_k \;:k\in[K]\}$ be a set of $K$ unitaries that encode matrices $H_k\in\mathbb{C}^{2^{n_s}\times2^{n_s}}$ such that
	\begin{align}
	\label{eq:compression_U_definition}
	(\bra{0}_{a}\otimes \openone_s)  U_k (\ket{0}_{a}\otimes \openone_s) = H_k, \quad\|H_k\|\le 1,\quad \ket{0}_a \in \mathbb{C}^{2^{n_a}}.
	\end{align}
	Then there exists a quantum circuit $V$ such that on input states spanned by $\{\ket{k}_b: k\in\{0,\cdots,K\}\}$,
	\begin{align}
	\label{eq:compression_V_definition}
	(\bra{0}_{ac}\otimes \openone_s)  V (\ket{0}_{ac}\otimes \openone_s) = \ket{0}\bra{0}_b \otimes \openone_s + \sum_{k=1}^K\ket{k}\bra{k}_b\otimes\left(\prod^{k}_{j=1}H_j\right), 
	\quad \ket{0}_a \in \mathbb{C}^{2^{n_a}},
	\quad \ket{k}_b \in \mathbb{C}^{2^{n_b}},
	\quad \ket{0}_c \in \mathbb{C}^{2^{n_c}},
	\end{align}
	where the number of qubits $n_b\in\mathcal{O}(n_c)={\mathcal{O}}(\log{(K)})$.  The cost of $V$ is one query each to controlled-controlled-$U_k$, and $\mathcal{O}(K(n_a+\log{(K)}))$ additional primitive quantum gates. 
\end{lemma}
\begin{proof}
	Though binary control logic for this sequence is trivial when $H_k$ is unitary, the complication here is that $H_k$ is in general non-unitary and so the probability of successfully measuring $\ket{0}_{a}$ is less than one. Any other measurement outcome corresponds failure as it applies on register $a$ an operator that is not $H_k$. This complication is overcome by introducing two more registers $b,c$ of size $\mathcal{O}(\log(K))$ qubits that coherently count the number of successful measurements, and then applying $U_k$ conditional on there being no failures.
	
	Let the counter register $b$ represent an $n_{b}$-bit integer $l_{b}=\sum^{n_{b}-1}_{r=0}2^{r}q_r$ in the number state $\ket{l_{b}}_{b}:=\ket{q_0q_1\cdots q_{n_{b,c}-1}}_{b}$, where $q_r\in\{0,1\}$, and similarly for the counter register $c$.  The size of these integers are determined by $n_b=n_c+1=\lceil\log_2{(K+1)}\rceil+1$. The unitaries $U_j$ will be applied conditional on both the leading bits $q_{n_c-1}=0$ and $q_{n_b-1}=0$, that is 
	\begin{align}
	\label{eq:cc-U_k}
	\operatorname{CC-}U_k:=I^{\otimes n_b+n_c-2}\otimes\left(\ket{0}\bra{0}_{b_{n_b-1}}\otimes\ket{0}\bra{0}_{c_{n_c-1}}\otimes U_k + \cdots\right).
	\end{align}
	
	Consider the circuit in~\cref{fig:Compression-Gadget}. There, we apply $\operatorname{CC-}U_k$, then increment $k$ by one, decrement $l_c$ by one conditional on the $a$ register not being in the $\ket{0}_a$ state, and decrement $l_b$ by one conditional on the $a$ register being in the $\ket{0}_a$ state. This is accomplished by multiply-controlled modular addition
	\begin{align}
	\operatorname{ADD}_{ca}&= \operatorname{ADD}^\dag_b\otimes \openone_c\otimes \ket{0}\bra{0}_a+\openone_b\otimes \operatorname{ADD}^\dag_c\otimes\sum^{2^{n_a}-1}_{l=1}{\ket{l}\bra{l}_a},
	\\\nonumber
	\operatorname{ADD}_{b}&=\sum^{2^{n_{b}}-1}_{l=0}\ket{l+1\mod{2^{n_{b}}}}\bra{l}_{b},
	\quad
	\operatorname{ADD}_{c}=\sum^{2^{n_{c}}-1}_{l=0}\ket{l+1\mod{2^{n_{c}}}}\bra{l}_{c}.
	\end{align}
	As we add integers of size $\mathcal{O}(K)$, each application of modular addition costs $\mathcal{O}(\log (K))$ primitive gates and requires $\mathcal{O}(\log (K))$ qubits~\cite{Cuccaro2004Adder}. Implementing the multiple controls costs $\mathcal{O}(n_a)$ primitive gates and up to $n_a$ extra qubits. 
	\begin{figure*}[t]
		\begin{tikzpicture}[thick]

\matrix[row sep=0.0cm, column sep=0.1cm] (circuit) {
	\node (c) {\ketb{0}{c}}; & \node {\textbackslash}; & \node[joint] (cV) {};
	&\node () {}; &\node (cDoubleStart) {};& \node (cdots) {};&\node {};
	& \node[doubleLineVdots] (c0dots) {$\vdots$}; 
	& \node[ctrlDoubleLine] (c0) {};& \node[operatorDoubleLine] () {$\operatorname{ADD}_c^\dag$};   & \node[operatorDoubleLine] (addc0) {$\operatorname{ADD}_c$}; 
	& \node[ctrlDoubleLine] (c1) {};& \node[operatorDoubleLine] () {$\operatorname{ADD}_c^\dag$};   & \node[operatorDoubleLine] (addc1) {$\operatorname{ADD}_c$}; 
	&\node {};& \node (cdots) {};&\node {};
	& \node[ctrlDoubleLine] (c2) {};& \node[operatorDoubleLine] () {$\operatorname{ADD}_c^\dag$};   & \node[operatorDoubleLine] (addc2) {$\operatorname{ADD}_c$}; 
	& 
	&
	&
	\coordinate (endc); \\
	\node (b) {\ketb{k}{b}}; & \node {\textbackslash}; & \node[joint] () {}; 
	&&\node (bDoubleStart) {};&&
	& \node[doubleLineVdots] (b0dots) {$\vdots$}; 
	& \node[ctrlDoubleLine] (b0) {}; && \node[operatorDoubleLine] (addb0) {$\operatorname{ADD}_b^\dag$};  
	& \node[ctrlDoubleLine] (b1) {}; && \node[operatorDoubleLine] (addb1) {$\operatorname{ADD}_b^\dag$};  
	&&&
	& \node[ctrlDoubleLine] (b2) {}; && \node[operatorDoubleLine] (addb2) {$\operatorname{ADD}_b^\dag$};
	& \node[operatorDoubleLine] () {$\operatorname{ADD}_b^K$}; 	  
	&	\coordinate (endb);\\
	\node (a) {\ketb{0}{a}}; & \node {\textbackslash}; & \node[joint] () {}; 
	&&&&
	&\node {\textbackslash}; 
	& \node[joint] {}; && \node[ctrl] (aAdd) {};	
	& \node[joint] {}; && \node[ctrl] (adda1) {};	
	&&&
	& \node[joint] {}; && \node[ctrl] (adda2) {};	 
	&&	\coordinate (enda);\\
	\node (s) {\ketb{\psi}{s}}; & \node {\textbackslash}; & \node[operator] (sV) {$V$}; 
	&&\node (sDoubleStart) {};&&
	&\node {\textbackslash}; 
	& \node[operator] (U1) {$U_1$}; &	& 
	& \node[operator] (U2) {$U_2$}; &	& 
	&& \node (sdots) {};&
	& \node[operator] (UK) {$U_K$}; &	& 
	&&	\coordinate (ends);\\
};
\begin{pgfonlayer}{background}
\draw[doubleLine, thin] (cDoubleStart) -- (endc) (bDoubleStart) -- (endb);
\draw[thick] (c) -- (cDoubleStart) (b) -- (bDoubleStart); 
\draw[thick] (a) -- (enda) (s) -- (ends); 
\draw[thick] (a) -- (enda) (s) -- (ends); 
\draw[thick] (c0) -- (U1) (addc0) -- (aAdd);
\draw[thick] (c1) -- (U2) (addc1) -- (adda1);
\draw[thick] (c2) -- (UK) (addc2) -- (adda2);
\draw[thick] (cV) -- (sV);
\node[fit = (cdots) (sdots), fill=white,inner sep=8pt]  (dots) {$\hspace{-1pt}\cdots$};
\node[fit = (cDoubleStart) (sDoubleStart), fill=white,inner sep=5pt]  (dots) {$=$};
\end{pgfonlayer}
\end{tikzpicture}
		\caption{	\label{fig:Compression-Gadget} Quantum circuit representations of the gadget $V$ for probabilistically applying a sequence of operators $H_k \cdots H_2H_1$, encoded in $(\bra{0}_{a}\otimes \openone_s)  U_k (\ket{0}_{a}\otimes \openone_s) = H_k$, controlled on number state $\ket{k}_b,\; k\in\{0,1,\cdots,K\}$. Horizontal lines without a backslash depict single-qubit registers. Filled circles depict a unitary controlled by the $\ket{0}\cdots\ket{0}$ state.}	
	\end{figure*}
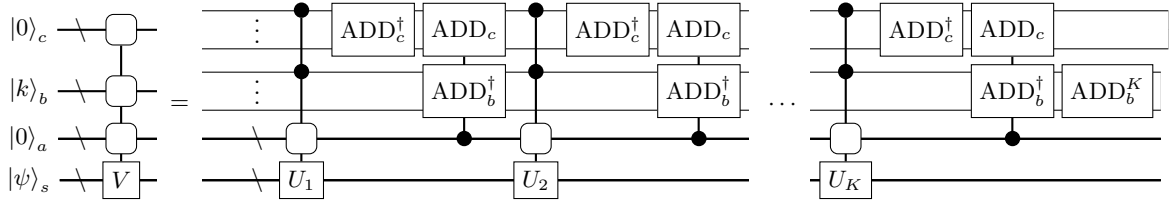
	
	Restricted to input states $\ket{0}_a\ket{l_b}_b\ket{0}_c$, where $l_b\in\{2^{n_b}-1,0,1,2,3,\cdots,K-1\}$, this implements $V$. For example, consider the evolution of an input state $\ket{0}_a\ket{1}_b\ket{0}_c\ket{\psi}_s$ for $K=3$.
	\begin{align}
	\label{eq:compression_sequence}
	\ket{0}_a\ket{1}_b\ket{0}_c\ket{\psi}_s
	\underset{\operatorname{CC-}U_1}{\rightarrow}
	&\ket{0}_a\ket{1}_b\ket{0}_c H_1\ket{\psi}_s+\ket{0^{\perp,1}}_a\ket{2}_b\ket{0}_c \cdots
	\\\nonumber
	\underset{\operatorname{ADD}_{ca}}{\rightarrow}
	&\ket{0}_a\ket{0}_b\ket{0}_c H_1\ket{\psi}_s+\ket{0^{\perp,1}}_a\ket{2}_b\ket{2^{n_c}-1}_c \cdots
	\\\nonumber
	\underset{\operatorname{CC-}U_2}{\rightarrow}
	&\ket{0}_a\ket{0}_b\ket{0}_c H_2H_1\ket{\psi}_s+\ket{0^{\perp,2}}_a\ket{1}_b\ket{0}_c \cdots+\ket{0^{\perp,1}}_a\ket{2}_b\ket{2^{n_c}-1}_c \cdots
	\\\nonumber
	\underset{\operatorname{ADD}_{ca}}{\rightarrow}
	&\ket{0}_a\ket{2^{n_b}-1}_b\ket{0}_c H_2H_1\ket{\psi}_s+\ket{0^{\perp,2}}_a\ket{1}_b\ket{2^{n_c}-1}_c \cdots+\ket{0^{\perp,1}}_a\ket{2}_b\ket{2^{n_c}-2}_c \cdots
	\\\nonumber
	\underset{\operatorname{CC-}U_3}{\rightarrow}
	&\ket{0}_a\ket{2^{n_b}-1}_b\ket{0}_c H_2H_1\ket{\psi}_s+\ket{0^{\perp,2}}_a\ket{1}_b\ket{2^{n_c}-1}_c \cdots+\ket{0^{\perp,1}}_a\ket{2}_b\ket{2^{n_c}-2}_c \cdots
	\\\nonumber
	\underset{\operatorname{ADD}_{ca}}{\rightarrow}
	&\ket{0}_a\ket{2^{n_b}-2}_b\ket{0}_c H_2H_1\ket{\psi}_s+\ket{0^{\perp,2}}_a\ket{1}_b\ket{2^{n_c}-2}_c \cdots+\ket{0^{\perp,1}}_a\ket{2}_b\ket{2^{n_c}-3}_c \cdots
	\\\nonumber
	\underset{\operatorname{ADD}^K_{b}}{\rightarrow}
	&\ket{0}_a\ket{1}_b\ket{0}_c H_2H_1\ket{\psi}_s+\cdots.
	\end{align}
	In the above, subtracting from $0$ results in the largest possible integer, hence the leading bit becomes $q_{n_b-1}=1$, ans similarly for $q_{n_c-1}=1$. As a result, the controls in~\cref{eq:cc-U_k} do not apply $U_k$. By choosing the largest integer representable by the $b$ register to be at least two times of $K$, we also ensure that this leading bit is set to $1$, it will remain in the same state after the at most $K$ subtractions. Note that~\cref{eq:compression_V_definition} applies $H_k\cdots H_1$ controlled on $\ket{k}_b$, whereas~\cref{eq:compression_sequence} applies $H_{l_b}\cdots H_1$ controlled on $\ket{l_b-1}_b$ -- we simply relabel $k=l_b-1\mod{2^{n_b}}$. 
\end{proof}

Before proceeding to the proof of~\cref{Thm:Compressed_TDS} we need to use a well-known result called `robust oblivious amplitude amplification', restated below for convenience.
\begin{lemma}[Robust oblivious amplitude amplification~\cite{berry2015simulating}]
	\label{lem:Robust_OAA}
	Let $V,U$ be unitary and let $\tilde{U}$ be an arbitrary matrix such that $\|U-\tilde{U}\|\in\mathcal{O}(\epsilon)$, and $(\bra{0}_a\otimes \openone_s)V(\ket{0}_a\otimes \openone_s)=\frac{\tilde{U}}{2}$. Let $W=-V\cdot(\operatorname{REF}\otimes \openone_s)\cdot V^\dag\cdot(\operatorname{REF}\otimes \openone_s)\cdot V$, where $\operatorname{REF}=\openone_a-2\ket{0}\bra{0}_a$. Then
	$\|(\bra{0}_a\otimes \openone_s)W(\ket{0}_a\otimes \openone_s) - U\| \in \mathcal{O}(\epsilon)$.
\end{lemma}

The proof of~\cref{Thm:Compressed_TDS} follows.
\mainThm*
\begin{proof}[Proof of~\cref{Thm:Compressed_TDS}.]
	The unitary $\operatorname{DYS}_{K}$, as defined in~\cref{eq:DYS-K-compressed}, may be implemented through~\cref{Thm:compression_gadget} provided that we find a sequence $\{U_k\}$ such that $H_k\cdots H_2H_1\propto B_k$. -- in other words, 
	\begin{align}
	\label{eq:DYS-K-compressed}
	&(\bra{0}_{ac,\text{others}}\otimes \openone_{bs}) \operatorname{DYS}_{K} (\ket{0}_{ac,\text{others}}\otimes \openone_{bs}) 
	= \sum^K_{k=0}\ket{k}\bra{k}_{b}\otimes \gamma_k B_k,
	\end{align}
	where `$\text{others}$' represent registers with size independent of $K$, and $\gamma_k$ is a scaling factor depends on the choice of $U_k$. This sequence is obtained by combining three matrices. First, a unitary matrix $U$ that prepares a uniform superposition $U\ket{0}_d=\sum^{M-1}_{m=0}\frac{1}{\sqrt{M}}\ket{m}_d$. Second, the block-diagonal matrix 
	\begin{align}
	D=\sum^{M-1}_{m=0}\ket{m}\bra{m}_d\otimes H(\Delta m), \quad\Delta=t/M,
	\end{align}
	implemented by $\operatorname{HAM-T}$. Third, a strictly upper-triangular matrix $G\in \mathbb R^{M\times M}$ with elements
	\begin{align}
	G_{ij}=
	\begin{cases}
	\frac{1}{M}, & i < j, \\
	0, & \text{otherwise},
	\end{cases}
	\quad
	G=\frac{1}{M}\sum^{M-1}_{i=0}\sum^{M-1}_{j=i+1}\ket{i}\bra{j}_d.
	\end{align}
	The non-unitary triangular operator $G$ is implemented by using an integer comparator $\operatorname{COMP}$ acting on registers $d,e,f$ consisting of $n_d=n_e\in\mathcal{O}(\log(M))$ and $n_f=1$ qubits, and thus costs $\mathcal{O}(\log{(M)})$ primitive gates.
	For any input number state index $\ket{j}_d$, let us compare $j$ with a uniform superposition state $\sum^{M-1}_{i=0}\frac{1}{\sqrt{M}}\ket{i}_e$. Conditional on $i\ge j$, the comparator perform a $\operatorname{NOT}$ gate on register $f$. We then swap registers $d,e$, and unprepare the uniform superposition. On input $\ket{j}_d\ket{0}_e\ket{0}_f$, this implements the sequence
	\begin{align}
	\ket{j}_d\ket{0}_e\ket{0}_f
	\underset{U\;\text{on}\;e}{\rightarrow} 
	&\frac{\ket{j}_d}{\sqrt{M}}\sum^{M-1}_{i=0}\ket{i}_e\ket{0}_f
	\underset{\operatorname{COMP}}{\rightarrow}  
	\frac{\ket{j}_d}{\sqrt{M}}\sum^{M-1}_{i=0}\ket{i}_e\ket{i\ge j}_f
	\underset{\operatorname{SWAP}_{de}}{\rightarrow}  
	\frac{\ket{j}_e}{\sqrt{M}}\sum^{M-1}_{i=0}\ket{i}_d\ket{i\ge j}_f
	\\\nonumber
	\underset{U^\dag\;\text{on}\;e}{\rightarrow}  
	&\frac{1}{M}\sum^{M-1}_{i=0}\ket{i}_d\ket{0}_e\ket{i\ge j}_f+\cdots,
	\end{align}
	where $\ket{i\ge j}_f = \ket{1}_f$ if $i\ge j$ and is $\ket{0}_f$ if $i<j$. This defines the following circuit $\operatorname{LT}$ that encodes $G$ using $\mathcal{O}(\log{(M)})$ primitive gates.
	\begin{align}
	\operatorname{LT}&=(\openone_f\otimes U^\dag\otimes \openone_{d})\cdot(\openone_f\otimes\operatorname{SWAP}_{de})\cdot\operatorname{COMP}\cdot(\openone_f\otimes U\otimes \openone_{d}),
	\\ \nonumber
	&\Rightarrow(\bra{0}_{ef}\otimes \openone_d)\operatorname{LT}(\ket{0}_{ef}\otimes \openone_d)=\frac{1}{M}\sum^{M-1}_{i=0}\sum^{M-1}_{j=i+1}\ket{i}\bra{j}_d=G.
	\end{align}
	\begin{figure*}[t]
		\tikzstyle{DYSKspace} = [text width = 25pt]
\tikzstyle{TDSspace} = [text width = 25pt]
\tikzstyle{REFspace} = [text width = 20pt]
{\centering
	\begin{tikzpicture}
	
	\matrix[row sep=0.0cm, column sep=0.1cm] (circuit) {
		\node (f) {\ketb{0}{f}}; && \node[DYSKspace] (fV) {}; & \node[minimum size = 10pt] (fStart) {};&\node[minimum size = 7pt] () {};&&&&
		&\node[ctrl] (addf0) {};& \node (leftbracketf) {};
		&& \node[oplus] (COMPf) {$\pmb\oplus$}; &&&& \node[ctrl] (addCTRLf-2) {};
		&	\coordinate (endfbracket); &&& \coordinate (endf);\\
		
		\node (e) {\ketb{0}{e}}; & \node {\textbackslash};&&&&\node {\textbackslash};&&&
		&\node[ctrl] () {};&
		& \node[operator] () {$U$}; &\node[operatorDoubleRegisterWidth] (COMPe) {};&\node[operatorDoubleRegisterWidth] (SWAPe) {};&& \node[operator] () {$U^\dag$};& \node[ctrl] () {};&&
		&	\coordinate (ende);\\
		
		\node (d) {\ketb{0}{d}}; & \node {\textbackslash};&&&&\node {\textbackslash};&
		\node[operator] {$U$}; &\node[joint] (HAMT-d) {T};&
		\node[operator] {$U^\dag$};&\node[ctrl] () {}; &
		&\node[operator] {$U$};&\node (COMPd) {};&\node (SWAPd) {};&\node[joint] (HAMTd-2) {T}; & \node[operator] {$U^\dag$};& \node[ctrl] () {};
		&&&	\coordinate (endd);\\
		
		\node (c) {\ketb{0}{c}}; & \node {\textbackslash}; &&\node (cDoubleStart) {};& 
		& \node[doubleLineVdots] (c0dots) {$\vdots$}; 
		&& \node[ctrlDoubleLine] (c0) {}; & \node[operatorDoubleLine] () {$\operatorname{ADD}_c^\dag$};& \node[operatorDoubleLine] (addc0) {$\operatorname{ADD}_c$};   
		&&&&
		& \node[ctrlDoubleLine] (HAMTc-2) {}; & \node[operatorDoubleLine] () {$\operatorname{ADD}_c^\dag$};  & \node[operatorDoubleLine] (addc1) {$\operatorname{ADD}_c$}; 
		&&&	\coordinate (endc); \\
		
		\node (b) {\ketb{k}{b}}; & \node {\textbackslash}; &&\node (bDoubleStart) {};&
		& \node[doubleLineVdots] (b0dots) {$\vdots$}; 
		&& \node[ctrlDoubleLine] (b0) {}; && \node[operatorDoubleLine] (addb0) {$\operatorname{ADD}_b^\dag$};  & 
		&&& & \node[ctrlDoubleLine] (HAMTb-2) {};
		&&  \node[operatorDoubleLine] (addb1) {$\operatorname{ADD}_b^\dag$};  
		&\node[minimum size = 13pt] () {};& \node[operatorDoubleLine] () {$\operatorname{ADD}_b^K$}; &	\coordinate (endb);\\
		
		\node (a) {\ketb{0}{a}}; & \node {\textbackslash}; &&&
		&\node {\textbackslash}; 
		&& \node[joint] {}; && \node[ctrl] (aAdd) {};	& 
		  &	& 
		&& \node[joint] {};&& \node[ctrl] (adda1) {}; 
		&&&	\coordinate (enda);\\
		
		\node (s) {\ketb{\psi}{s}}; & \node {\textbackslash}; & \node (sV) {};
		&\node (sStart) {};&
		&\node {\textbackslash}; 
		&& \node[operator] (U1) {$\operatorname{HAM-T}$}; &	& 
		&\node (leftbrackets) {};&&&& \node[operator] (U2) {$\operatorname{HAM-T}$}; &	& 
		&	\coordinate (endsbracket);&&\coordinate (ends);\\
	};
	\begin{pgfonlayer}{background}
	\draw[doubleLine,thin] (cDoubleStart) -- (endc) (bDoubleStart) -- (endb);
	\draw[]   (f) -- (endf); 
	\draw[thick] (a) -- (enda) (c) -- (cDoubleStart) (b) -- (bDoubleStart) (s) -- (ends) (d) -- (endd) (e) -- (ende); 
	\draw[thick] (HAMT-d) -- (U1) (addf0) -- (aAdd);
	\draw[thick] (HAMTd-2) -- (U2) (addCTRLf-2) -- (adda1);
	\draw[thick] (fV) -- (sV);
	\draw[thick] (COMPf) -- (COMPe);
	\node[fit = (fStart) (sStart), fill=white,inner sep=6pt, xshift=6pt]  (dots) {$=$};
	\node[fit = (fV) (sV),operatorDoubleRegister,DYSKspace] () {$\operatorname{DYS}_K$};
	
	\node[fill=white,inner sep = 1pt] () [fit = (leftbracketf) (leftbrackets)] {};
	\draw[decorate,decoration={brace, amplitude = 5pt,mirror},thick, fill=white] ($(leftbracketf.north west)+(0.3cm,+0.2cm)$) to ($(leftbrackets.south west)+(0.3cm,-0.2cm)$);
	\draw[decorate,decoration={brace, amplitude = 5pt},thick, fill=white] ($(endfbracket.north west)+(0.0cm,+0.2cm)$) to ($(endsbracket.south west)+(0.0cm,-0.2cm)$);
	\node[fit = (SWAPe) (SWAPd),operatorDoubleRegister] () {$\operatorname{SWAP}$};
	\node[fit = (COMPe) (COMPd),operatorDoubleRegister] () {$\operatorname{COMP}$};
	\node at ($(endfbracket.north west)+(0.7cm,+0.2cm)$) {$K-1$};
	\end{pgfonlayer}
	\end{tikzpicture}
}
\newline
\begin{tabular}{cc}
	\begin{tikzpicture}
	
	\matrix[row sep=-0.1cm, column sep=0.1cm] (circuit) {
		\node (f) {\ketb{0}{f}}; & & \node[TDSspace] (fV) {}; & \node[minimum size = 10pt] (fStart) {};&\node[minimum size = 7pt] () {};&&
		&\node[TDSspace] (fTDS2) {};&
		&	\coordinate (endf);\\
		
		\node (e) {\ketb{0}{e}}; & \node {\textbackslash};&&&&\node {\textbackslash};&
		&&
		&	\coordinate (ende);\\
		
		\node (d) {\ketb{0}{d}}; & \node {\textbackslash};&&&&\node {\textbackslash};&
		&&
		&	\coordinate (endd);\\
		
		\node (c) {\ketb{0}{c}}; & \node {\textbackslash}; &&&&\node {\textbackslash};&
		&&
		&	\coordinate (endc); \\
		
		\node (b) {\ketb{0}{b}}; & \node {\textbackslash}; &&&&\node {\textbackslash};&
		\node[operator] () {$\operatorname{COEF}$}; &&\node[operator] () {$\operatorname{COEF}'^\dag$};
		&	\coordinate (endb);\\
		
		\node (a) {\ketb{0}{a}}; & \node {\textbackslash}; &&&&\node {\textbackslash};&
		&&
		&	\coordinate (enda);\\
		
		\node (s) {\ketb{\psi}{s}}; & \node {\textbackslash}; & \node (sV) {}; &\node (sStart) {};&&\node {\textbackslash};  &
		&\node (sTDS2) {};&
		&	\coordinate (ends);\\
	};
	\begin{pgfonlayer}{background}
	\draw[thin]   (f) -- (endf); 
	\draw[thick] (a) -- (enda) (c) -- (endc) (b) -- (endb) (s) -- (ends) (d) -- (endd) (e) -- (ende);
	\draw[thick] (fV) -- (sV);
	\node[fit = (fStart) (sStart), fill=white,inner sep=6pt, xshift=6pt]  (dots) {$=$};
	\node[fit = (fV) (sV),operatorDoubleRegister,TDSspace] () {$\operatorname{TDS}_\beta$};
	\node[fit = (fTDS2) (sTDS2),operatorDoubleRegister,TDSspace] () {$\operatorname{DYS}_K$};
	\end{pgfonlayer}
	\end{tikzpicture}
	&
	\begin{tikzpicture}
	
	\matrix[row sep=-0.1cm, column sep=0.1cm] (circuit) {
		\node (f) {\ketb{0}{f}}; & & \node[TDSspace] (fV) {}; & \node[minimum size = 10pt] (fStart) {};&\node[minimum size = 7pt] () {};&&
		\node[TDSspace] (fTDS1) {};&
		\node[REFspace] (fREF2) {};&
		\node[TDSspace] (fTDS3) {};&
		\node[REFspace] (fREF4) {};&
		\node[TDSspace] (fTDS5) {};&
		&	\coordinate (endf);\\
		
		\node (e) {\ketb{0}{e}}; & \node {\textbackslash};&&&&\node {\textbackslash};&
		&&&&&
		&	\coordinate (ende);\\
		
		\node (d) {\ketb{0}{d}}; & \node {\textbackslash};&&&&\node {\textbackslash};&
		&&&&&
		&	\coordinate (endd);\\
		
		\node (c) {\ketb{0}{c}}; & \node {\textbackslash}; &&&&\node {\textbackslash};&
		&&&&&
		&	\coordinate (endc); \\
		
		\node (b) {\ketb{0}{b}}; & \node {\textbackslash}; &&&&\node {\textbackslash};&
		&&&&&
		&	\coordinate (endb);\\
		
		\node (a) {\ketb{0}{a}}; & \node {\textbackslash}; &&&&\node {\textbackslash};&
		\node (aTDS1) {};&
		\node (aREF2) {};&
		\node (aTDS3) {};&
		\node (aREF4) {};&
		\node (aTDS5) {};&
		&	\coordinate (enda);\\
		
		\node (s) {\ketb{\psi}{s}}; & \node {\textbackslash}; & \node (sV) {}; &\node (sStart) {};&&\node {\textbackslash}; &
		\node (sTDS1) {};&
		\node (sREF2) {};&
		\node (sTDS3) {};&
		\node (sREF4) {};&
		\node (sTDS5) {};&
		&	\coordinate (ends);\\
	};
	\begin{pgfonlayer}{background}
	\draw[thin]   (f) -- (endf); 
\draw[thick] (a) -- (enda) (c) -- (endc) (b) -- (endb) (s) -- (ends) (d) -- (endd) (e) -- (ende);
	\draw[thick] (fV) -- (sV);
	\node[fit = (fStart) (sStart), fill=white,inner sep=6pt, xshift=6pt]  (dots) {$=$};
	\node[fit = (fV) (sV),operatorDoubleRegister,REFspace] () {$\operatorname{TDS}$};
	\node[fit = (fTDS1) (sTDS1),operatorDoubleRegister,TDSspace] () {$\operatorname{TDS}_\beta$};
	\node[fit = (fTDS3) (sTDS3),operatorDoubleRegister,TDSspace] () {$\operatorname{TDS}^\dag_\beta$};
	\node[fit = (fTDS5) (sTDS5),operatorDoubleRegister,TDSspace] () {$\operatorname{TDS}_\beta$};
	\node[fit = (fREF2) (aREF2),operatorDoubleRegister,REFspace] () {$\operatorname{REF}$};
	\node[fit = (fREF4) (aREF4),operatorDoubleRegister,REFspace] () {$\operatorname{REF}$};
	\end{pgfonlayer}
	\end{tikzpicture}
\end{tabular}
		\caption{	\label{fig:DYS_compressed} Quantum circuit representation of (top) $\operatorname{DYS}_K$ in~\cref{eq:DYS-K-compressed}, implemented using the compression gadget~\cref{Thm:compression_gadget} depicted in~\cref{fig:Compression-Gadget}; (bottom, left) a single step of time-evolution by the truncated Dyson series algorithm from~\cref{eq:LCU_TDS} before oblivious amplitude amplification; (bottom, right) a single step of time-evolution by the truncated Dyson series algorithm from~\cref{eq:TDS_compressed}. Note that when $\beta=2$, a single-round of oblivious amplitude amplification is used.}
	\end{figure*}
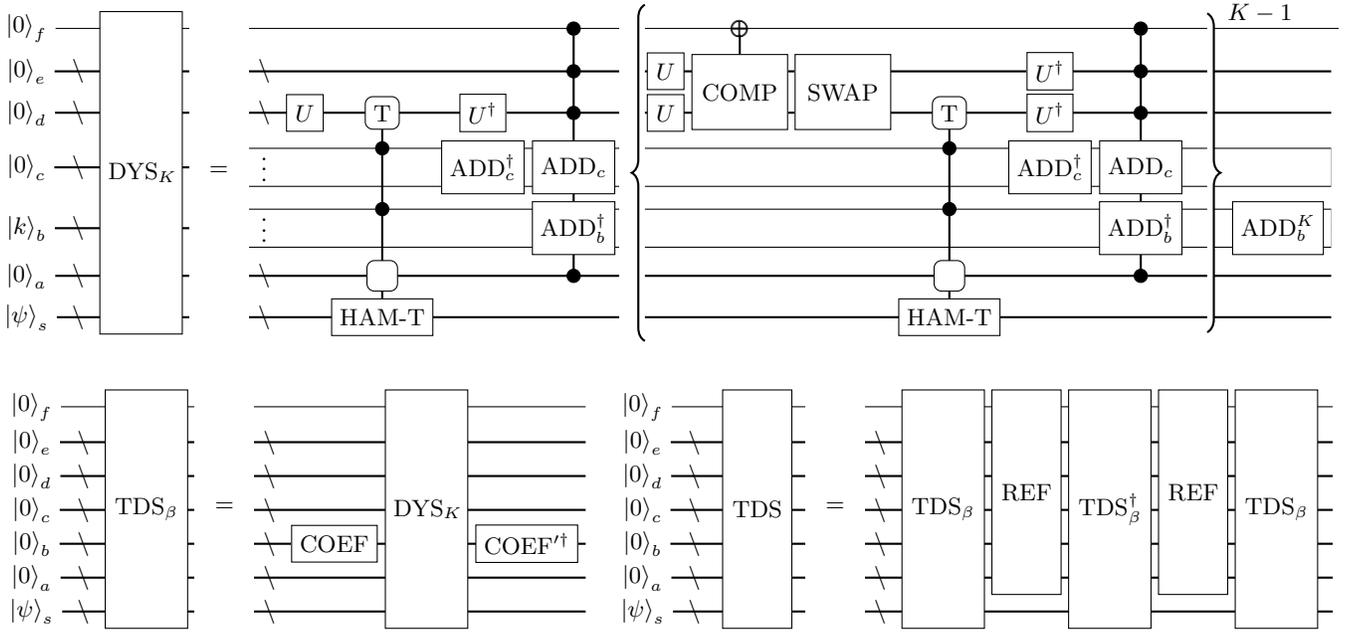
	
	One may then verify that the terms $B_k$ are generated by the following sequence
	\begin{align}
	\label{eq:TDS_intermediateStep}
	\bra{0}_d U^\dag \cdot D\cdot U\ket{0}_d &=\frac{B_1}{M}= \frac{1}{M}\sum^{M-1}_{m_1=0}H(\Delta m),
	\\ \nonumber
	\bra{0}_d U^\dag \cdot (D\cdot G)\cdot D\cdot U\ket{0}_d &=\frac{B_2}{M^2}= \frac{1}{M^2}\sum_{0\le m_1 < m_2 < M} H(\Delta m_2)H(\Delta m_1),
	\\\nonumber
	\vdots
	\\\nonumber
	\bra{0}_d U^\dag\cdot (D \cdot G)^{k-1}\cdot D \cdot U\ket{0}_d &=\frac{B_k}{M^k}= \frac{1}{M^k}\sum_{0\le m_1< m_2 < \cdots m_j < M} H(\Delta m_j)\cdots H(\Delta m_2)\cdots H(\Delta m_1).
	\end{align}
	Thus we make the choice
	\begin{align}
	U_k:=
	\begin{cases}
	(U^\dag\otimes \openone_{aefs})\cdot(\operatorname{HAM-T}\otimes \openone_{ef})\cdot(U\otimes \openone_{aefs}),& k = 1,\\
	(U^\dag\otimes \openone_{aefs})\cdot(\operatorname{HAM-T}\otimes \openone_{ef})\cdot(\operatorname{LT}\otimes \openone_{as})\cdot(U\otimes \openone_{aefs}),& k > 1.\\
	\end{cases}
	\end{align}
	Combined with~\cref{Thm:compression_gadget}, this leads to the circuit of~\cref{fig:DYS_compressed} which implements $\operatorname{DYS}_K$ in~\cref{eq:DYS-K-compressed} by identifying `$\text{others}$' with the $d$, $e$ and $f$ registers, and recognizing from~\cref{eq:TDS_intermediateStep} that the scaling factor $\gamma_k=\frac{1}{M^k}$. In other words, 
	\begin{align}
	\label{eq:DYS-K-compressed_complete}
	&(\bra{0}_{acdef}\otimes \openone_{bs}) \operatorname{DYS}_{K} (\ket{0}_{acdef}\otimes \openone_{bs}) 
	= \sum^K_{k=0}\ket{k}\bra{k}_{b}\otimes \frac{B_k}{M^k},
	\end{align}
	According to~\cref{Thm:compression_gadget}, the number of primitive gates required by $\operatorname{DYS}_{K}$, excluding that for the $U_k$, is $\mathcal{O}(K(n_a+n_d+n_e+n_f+\log{(K)}))=\mathcal{O}(K(n_a+\log{(M)}+\log{(K)}))$.
	
	We then select the desired linear combination of different orders in the Dyson series with the state preparation unitary
	\begin{align}
	\operatorname{COEF}\ket{0}_b=\frac{1}{\sqrt{\beta}}\sum^K_{k=0}\sqrt{(-it)^k}\ket{k}_b,
	\quad
	\operatorname{COEF}'\ket{0}_b=\frac{1}{\sqrt{\beta}}\sum^K_{k=0}\sqrt{t^k}\ket{k}_b,
	\quad
	\beta
	=\sum^K_{j=0}t^k
	\le \sum^\infty_{k=0}t^k=\frac{1}{1-t},
	\end{align}
	which can be implemented using $\mathcal{O}(K)$ primitive gates~\cite{shende2006synthesis}. In summing the $t^k$, we assume that $t< 1$ for convergence. The resulting unitary $\operatorname{TDS}_\beta$, is defined as follows.
	\begin{align}
	\label{eq:LCU_TDS}
	\operatorname{TDS}_\beta&:= (\operatorname{COEF}'^\dag\otimes \openone_{acdefs})\cdot \operatorname{DYS}_{K}\cdot (\operatorname{COEF}\otimes \openone_{acdefs})
	\\\nonumber
	&\Rightarrow(\bra{0}_{abcdef}\otimes \openone_s) \operatorname{TDS}_\beta (\ket{0}_{abcdef}\otimes \openone_s)
	=\frac{\sum^K_{k=0}(-it)^k B_k}{M^k\beta}\approx \frac{\mathcal{T} e^{-i\int_0^t H(s) \mathrm{d} s} }{\beta}.
	\end{align}
	Using the provided parameters for $K\in\mathcal{O}\left(\frac{\log{(1/\epsilon)}}{\log\log{(1/\epsilon)}}\right)$ and $M=\frac{t^2}{\epsilon} \left(\langle \|\dot{H}\| \rangle  +\max_s \|H(s)\|^2\right)$, the numerator, by~\cref{Thm:Truncated_Dyson_Algorithm}, approximates a unitary operation to error $\mathcal{O}(\epsilon)$.
	
	The probability of applying this operation can be boosted from  $|\frac{1+\Theta(\epsilon)}{\beta}|^2$ to $1-\mathcal{O}(\epsilon)$. If we choose $t=\Theta(1)\approx 1/2$ to be sufficiently small such that $\beta=2$, then a single round of robust oblivious amplitude amplification, outlined in~\cref{lem:Robust_OAA},  suffices. This implements a single time-step of the truncated Dyson series algorithm  $\operatorname{TDS}$ in~\cref{fig:DYS_compressed} as follows.
	\begin{align}
	\label{eq:TDS_compressed}
	\operatorname{TDS}&=-\operatorname{TDS}_2\cdot(\operatorname{REF}\otimes \openone_s)\cdot\operatorname{TDS}_2^\dag\cdot(\operatorname{REF}\otimes \openone_s)\cdot\operatorname{TDS}_2,
	\\\nonumber
	&\Rightarrow\|(\bra{0}_{abcdef}\otimes \openone_s) \operatorname{TDS} (\ket{0}_{abcdef}\otimes \openone_s)
	- \mathcal{T} [e^{-i\int_0^t H(s) \mathrm{d} s}]\|\in\mathcal{O}(\epsilon).
	\end{align}
	Note that each reflection $\operatorname{REF}=\openone_{abcdef}-2\ket{0}\bra{0}_{abcdef}$ acts on $n_a+\mathcal{O}(\log{(K)}+\log{(M)})$ qubits and therefore costs $\mathcal{O}(n_a+\log{(K)}+\log{(M)})$ gates. If we wish to simulate time-evolution for any $t\le \frac{1}{2}$, this will lead to $\beta\le 2$. In this situation, there are a variety of methods to boost $\beta$ back to $2$. Unlike oblivious amplitude amplification, this corresponds to decreasing the success probability and is easy to accomplish. For instance, introducing an additional qubit together with a $1$ single-qubit rotation may be used as described by~\cite{berry2015simulating} to artificially decrease the overlap.
	
	We now tally the query, gate, and qubit complexity. From~\cref{fig:DYS_compressed}, the number of $\operatorname{HAM-T}$ queries is $3K\in\mathcal{O}\left(\frac{\log{(1/\epsilon)}}{\log\log{(1/\epsilon)}}\right)$. The gate complexity is that of $\operatorname{REF}$, $\operatorname{COEF}$, $K$ times of $\operatorname{LT}$, and $K$ times of the multiply-controlled modular addition circuits $\operatorname{ADD}$. This is dominated by the addition circuits, with gate complexity $\mathcal{O}(K(n_a+\log{(M)}+\log{(K)}))=\mathcal{O}(n_a+\log{(M)}\frac{\log{(1/\epsilon)}}{\log\log{(1/\epsilon)}})$ as $M$ has the dominant $\epsilon$ scaling. The number of qubits in each register is $n_b=n_c\in\mathcal{O}(\log{(K)})$, $n_d=n_e\in\mathcal{O}(\log{(M)})$, and $n_f=1$. Thus $n_s+n_a+n_b+n_c+n_d+n_e+n_f=n_s+n_a+\mathcal{O}(\log{(K)}+\log{(M)})=n_s+n_a+\mathcal{O}(\log{(M)})$. However, the control logic for multiply-controlled unitaries can require up to a single duplication of the control registers. Thus the qubit complexity is $n_s+\mathcal{O}(n_a+\log{(M)})$. Note that we leave $n_s$ out of the big-$\mathcal{O}$ set as this register is never duplicated.
\end{proof}

\section{Truncated Dyson series algorithm by duplicating control registers}
\label{sec:dyson_series_alg_duplicated_registers}
In this section, we present a quantum algorithm that applies the $\epsilon$-approximation $\tilde{U}=\sum^\infty_{k=0} \frac{(-it)^k }{M^k}B_k$ to the time-ordered evolution operator $\mathcal{T}\left[e^{-i\int_0^t H(s) \mathrm{d}s}\right]$, where the truncation order $K$ and the number of discretization points $M$ are given by~\cref{Thm:Truncated_Dyson_Algorithm}. This version is based on the original proposal by~\cite{berry2015simulating}, and applies the same operator as~\cref{Thm:Compressed_TDS} with the same query and gate complexity, but has worse space complexity. Our contributions here are rigorous bounds on $K$ and $M$, and the implementation of a key step not discussed previously -- the efficient preparation of a particular quantum state that correctly selects a desired linear combination of time-ordered products of Hamiltonians. This step is non-obvious as the state has $\mathcal{O}(M!)$ different amplitudes, and in the worst-case would take $\mathcal{O}(M!)$ gates to create by arbitrary state preparation techniques. The cost of this implementation is captured by the following theorem.
\begin{theorem}[Hamiltonian simulation by a truncated Dyson series with duplicated registers]
	\label{Thm:Original_TDS}
	Let $H(s) : [0,t] \rightarrow \mathbb{C}^{2^{n_s}\times 2^{n_s}}$, let it be promised that $\max_{s}\|H(s)\|\le\alpha$ and $\langle\|\dot H\|\rangle=\frac{1}{t}\int^t_{0} \left\|\frac{\mathrm{d} H(s)}{ \mathrm{d} s}\right\| \mathrm{d}s$ and assume that the number of discretization points obeys $M\in{\mathcal{O}}\left( \frac{t^2}{\epsilon}\left({\langle \|\dot{H}\| \rangle} +{\max_s \|H(s)\|^2}\right)\right)$ in \cref{def:HAM-T}. 
	 For all $t\in[0,\frac{1}{2\alpha}]$ and $\epsilon> 0$, an operation $W$  can be implemented  such that $\left\|W-\mathcal{T}\left[ e^{-i\int_0^t H(s) \mathrm{d} s}\right]\right\| \le \epsilon$ with failure probability $\mathcal{O}(\epsilon)$ with the following costs.
	\begin{enumerate}
		\item Queries to $\operatorname{HAM-T}$: $\mathcal{O}\left(\frac{\log{(1/\epsilon)}}{\log\log{(1 /\epsilon)}}\right)$.
		\item Qubits: $n_s + \mathcal{O}\left(\left(n_a+\log{ \frac{t^2}{\epsilon}\left({\langle \|\dot{H}\| \rangle} +{\max_s \|H(s)\|^2}\right)}\right)\right)$.
		\item Primitive gates: $\mathcal{O}\left(\left(n_a + \log{ \frac{t^2}{\epsilon}\left({\langle \|\dot{H}\| \rangle} +{\max_s \|H(s)\|^2}\right)}\right)\frac{\log{(1/\epsilon)}}{\log\log{(1/\epsilon)}}\right)$.
	\end{enumerate}
\end{theorem}

\begin{proof}
	Let $\operatorname{HAM-T}_K$ be a unitary that acts jointly on registers $s,\vec{a},\vec{b},c,\vec{d}$. This unitary is defined to apply products of Hamiltonians
	\begin{multline}
	\label{eq:HAM-T_k}
	(\bra{0}_{\vec{a}}\otimes \openone_s) \operatorname{HAM-T}_{K} (\ket{0}_{\vec{a}}\otimes \openone_s) 
	\\
	:=
	\left(\sum^K_{k=0}\ket{k}\bra{k}_{\vec{b}}\otimes\left(\sum_{\vec{m}\in[M]^k}\ket{\vec{m}}\bra{\vec{m}}_{d_1\cdots d_k}\otimes \openone_{d_{k+1}\cdots d_K}\otimes \left(\prod^k_{j=1}H(m_j\Delta)\right)\right)+\cdots\right)\otimes \operatorname{SWAP}_c,
	\end{multline}
	where $\operatorname{SWAP}_c$ swaps the two qubits of register $c$, and a possible implementation is depicted in~\cref{fig:HAM-T_K}. Note that we only define the action of $\operatorname{HAM-T}_K$ for input states to register $\vec{b}$ that are spanned by basis states of the unary encoding $\ket{k}_{\vec{b}}=\ket{0}^{\otimes k}\ket{1}^{\otimes K-k}$, which determines the number of terms in the product. As seen in the figure, $\operatorname{HAM-T}_K$ makes $K$ queries to $\operatorname{HAM-T}$ and copies the $a, b$, and $d$ registers $K$ times.
	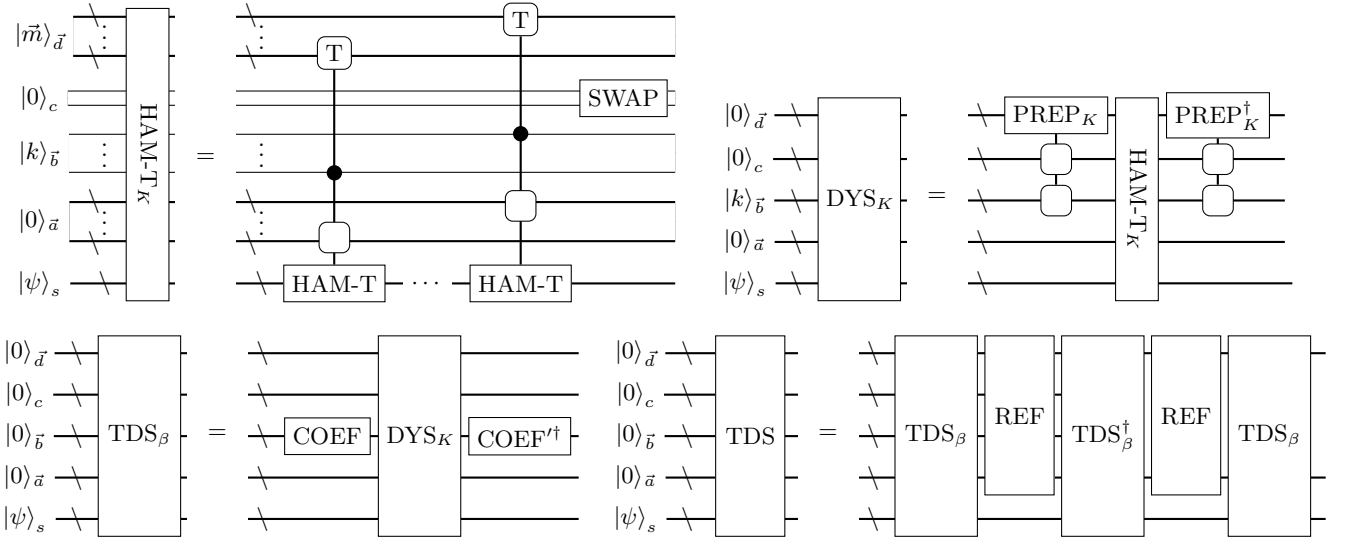
\begin{figure*}[t]
		\tikzstyle{HAMTKspace} = [text width = 10pt]
\tikzstyle{TDSspace} = [text width = 25pt]
\tikzstyle{REFspace} = [text width = 20pt]
\tikzstyle{DLup} = [above=-0pt]
\tikzstyle{DLdown} = [below=-0pt]
\tikzstyle{DLBSdup} = [DLup,xshift=-4pt]
\tikzstyle{DLBSddown} = [DLdown,xshift=-4pt]
\begin{tabular}{ll}
\begin{tikzpicture}[thick]
	\matrix[row sep=0.0cm, column sep=0.1cm] (circuit) {
	\node (d) {$|\vec{m}\rangle_{\vec{d}}$}; &\node[doubleLineVdots] (d0dots) {$\vdots$};\node[DLBSdup] {\textbackslash};\node[DLBSddown] {\textbackslash}; & \node[HAMTKspace,yshift = 4pt] (dV) {}; & \node[minimum size = 10pt] (dStart) {};&\node[minimum size = 7pt] () {};&\node[doubleLineVdots] (d0dots) {$\vdots$};\node[DLBSdup] {\textbackslash};\node[DLBSddown] {\textbackslash};&
	\node[joint, DLdown] (HAMd-1) {T};&
	\node[text width=20pt] () {};&
	\node[joint, DLup] (HAMd-2) {T};&
	&	\coordinate (endd);\\
	
	\node (c) {\ketb{0}{c}}; &  &&&&&
	&&&\node[operator] () {SWAP};
	&	\coordinate (endc); \\
	
	\node (b) {\ketb{k}{}$_{\vec{b}}$}; &\node[doubleLineVdots] (b0dots) {$\vdots$}; &&&&\node[doubleLineVdots] () {$\vdots$};&
	\node[ctrlDoubleLineDown] () {};&&
	\node[ctrlDoubleLine] () {};&
	&	\coordinate (endb);\\
	
	\node (a) {\ketb{0}{}$_{\vec{a}}$}; & \node[doubleLineVdots] (a0dots) {$\vdots$};\node[DLBSdup] {\textbackslash};\node[DLBSddown] {\textbackslash}; &&&&\node[doubleLineVdots] () {$\vdots$};\node[DLBSdup] {\textbackslash};\node[DLBSddown] {\textbackslash};&
	\node[joint, DLdown] () {};&&
	\node[joint, DLup] () {};&
	&	\coordinate (enda);\\
	
	\node (s) {\ketb{\psi}{s}}; & \node {\textbackslash}; & \node (sV) {}; &\node (sStart) {};&&\node {\textbackslash}; &
	\node[operator] (HAMs-1) {$\operatorname{HAM-T}$};&\node[fill=white] () {$\cdots$};&
	\node[operator] (HAMs-2) {$\operatorname{HAM-T}$};&
	&	\coordinate (ends);\\
};
\begin{pgfonlayer}{background}
\draw[doubleLineNarrow] (c) -- (endc);
\draw[doubleLine] (b) -- (endb) ;
\draw[doubleLine, thick] (a) -- (enda) (d) -- (endd); 
\draw[thick] (s) -- (ends); 
\draw[thick] (HAMs-1) -- (HAMd-1) (HAMs-2) -- (HAMd-2);
\node[fit = (dStart) (sStart), fill=white,inner sep=6pt, xshift=6pt]  (dots) {$=$};
\node[fit = (dV) (sV),operatorDoubleRegister,HAMTKspace] (HAMTKfit) {};
\node[rotate=-90] at (HAMTKfit) {$\operatorname{HAM-T}_K$};
\end{pgfonlayer}
\end{tikzpicture}
&
\begin{tikzpicture}[thick]
\matrix[row sep=-0.0cm, column sep=0.1cm] (circuit) {
	\node (d) {\ketb{0}{}$_{\vec{d}}$}; &\node {\textbackslash}; & \node[TDSspace] (dV) {}; & \node[minimum size = 10pt] (dStart) {};&\node[minimum size = 7pt] () {};&\node  {\textbackslash};&
	\node[operator] (PREPd-1) {$\operatorname{PREP}_K$};&
	\node[HAMTKspace] (HAMTKd) {};&
	\node[operator] (PREPd-2) {$\operatorname{PREP}^\dag_K$};&
	&	\coordinate (endd);\\
	
	\node (c) {\ketb{0}{c}}; &\node  {\textbackslash};  &&&&\node  {\textbackslash};&
	\node[joint] {};&&\node[joint] {};&
	&	\coordinate (endc); \\
	
	\node (b) {\ketb{k}{}$_{\vec{b}}$}; &\node  {\textbackslash}; &&&&\node  {\textbackslash}; &
	\node[joint] (bJoint1) {};&&	\node[joint] (bJoint2) {};&
	&	\coordinate (endb);\\
	
	\node (a) {\ketb{0}{}$_{\vec{a}}$}; &\node {\textbackslash}; &&&&\node  {\textbackslash};&
	&&&
	&	\coordinate (enda);\\
	
	\node (s) {\ketb{\psi}{s}}; & \node {\textbackslash}; & \node (sV) {}; &\node (sStart) {};&&\node {\textbackslash}; &&\node (HAMTKs) {};&&&
	&	\coordinate (ends);\\
};
\begin{pgfonlayer}{background}
\draw[thick] (c) -- (endc);
\draw[thick] (a) -- (enda) (b) -- (endb) (d) -- (endd); 
\draw[thick] (s) -- (ends); 
\draw[thick] (bJoint1) -- (PREPd-1) (bJoint2) -- (PREPd-2);
\node[fit = (dStart) (sStart), fill=white,inner sep=6pt, xshift=6pt]  (dots) {$=$};
\node[fit = (dV) (sV),operatorDoubleRegister,TDSspace] () {$\operatorname{DYS}_K$};
\node[fit = (HAMTKs) (HAMTKd),operatorDoubleRegister,HAMTKspace] (HAMTKfit) {};
\node[rotate=-90] at (HAMTKfit) {$\operatorname{HAM-T}_K$};
\end{pgfonlayer}
\end{tikzpicture}
\end{tabular}
\newline
\begin{tabular}{ll}
	\begin{tikzpicture}[thick]
	
	\matrix[row sep=0.0cm, column sep=0.1cm] (circuit) {
		\node (f) {\ketb{0}{}$_{\vec{d}}$}; &\node {\textbackslash}; & \node[TDSspace] (fV) {}; & \node[minimum size = 10pt] (fStart) {};&\node[minimum size = 7pt] () {};&\node {\textbackslash};&
		&\node[TDSspace] (fTDS2) {};&
		&	\coordinate (endf);\\
	
		\node (c) {\ketb{0}{c}}; & \node {\textbackslash}; &&&&\node {\textbackslash};&
		&&
		&	\coordinate (endc); \\
		
		\node (b) {\ketb{0}{}$_{\vec{b}}$}; & \node {\textbackslash}; &&&&\node {\textbackslash};&
		\node[operator] () {$\operatorname{COEF}$}; &&\node[operator] () {$\operatorname{COEF}'^\dag$};
		&	\coordinate (endb);\\
		
		\node (a) {\ketb{0}{}$_{\vec{a}}$}; & \node {\textbackslash}; &&&&\node {\textbackslash};&
		&&
		&	\coordinate (enda);\\
		
		\node (s) {\ketb{\psi}{s}}; & \node {\textbackslash}; & \node (sV) {}; &\node (sStart) {};&&\node {\textbackslash};  &
		&\node (sTDS2) {};&
		&	\coordinate (ends);\\
	};
	\begin{pgfonlayer}{background}
	\draw[thick] (c) -- (endc) (b) -- (endb); 
	\draw[thick] (a) -- (enda) (s) -- (ends) (f) -- (endf); 
	\draw[thick] (fV) -- (sV);
	\node[fit = (fStart) (sStart), fill=white,inner sep=6pt, xshift=6pt]  (dots) {$=$};
	\node[fit = (fV) (sV),operatorDoubleRegister,TDSspace] () {$\operatorname{TDS}_\beta$};
	\node[fit = (fTDS2) (sTDS2),operatorDoubleRegister,TDSspace] () {$\operatorname{DYS}_K$};
	\end{pgfonlayer}
	\end{tikzpicture}
	&
	\begin{tikzpicture}[thick]
	
	\matrix[row sep=0.0cm, column sep=0.1cm] (circuit) {
		\node (f) {\ketb{0}{}$_{\vec{d}}$}; &\node {\textbackslash}; & \node[TDSspace] (fV) {}; & \node[minimum size = 10pt] (fStart) {};&\node[minimum size = 7pt] () {};&\node {\textbackslash};&
		\node[TDSspace] (fTDS1) {};&
		\node[REFspace] (fREF2) {};&
		\node[TDSspace] (fTDS3) {};&
		\node[REFspace] (fREF4) {};&
		\node[TDSspace] (fTDS5) {};&
		&	\coordinate (endf);\\

		\node (c) {\ketb{0}{c}}; & \node {\textbackslash}; &&&&\node {\textbackslash};&
		&
		&	\coordinate (endc); \\
		
		\node (b) {\ketb{0}{}$_{\vec{b}}$}; & \node {\textbackslash}; &&&&\node {\textbackslash};&
		&
		&	\coordinate (endb);\\
		
		\node (a) {\ketb{0}{}$_{\vec{a}}$}; & \node {\textbackslash}; &&&&\node {\textbackslash};&
		\node (aTDS1) {};&
		\node (aREF2) {};&
		\node (aTDS3) {};&
		\node (aREF4) {};&
		\node (aTDS5) {};&
		&	\coordinate (enda);\\
		
		\node (s) {\ketb{\psi}{s}}; & \node {\textbackslash}; & \node (sV) {}; &\node (sStart) {};&&\node {\textbackslash}; &
		\node (sTDS1) {};&
		\node (sREF2) {};&
		\node (sTDS3) {};&
		\node (sREF4) {};&
		\node (sTDS5) {};&
		&	\coordinate (ends);\\
	};
	\begin{pgfonlayer}{background}
	\draw[thick] (c) -- (endc) (b) -- (endb); 
	\draw[thick] (a) -- (enda) (s) -- (ends)  (f) -- (endf); 
	\draw[thick] (fV) -- (sV);
	\node[fit = (fStart) (sStart), fill=white,inner sep=6pt, xshift=6pt]  (dots) {$=$};
	\node[fit = (fV) (sV),operatorDoubleRegister,REFspace] () {$\operatorname{TDS}$};
	\node[fit = (fTDS1) (sTDS1),operatorDoubleRegister,TDSspace] () {$\operatorname{TDS}_\beta$};
	\node[fit = (fTDS3) (sTDS3),operatorDoubleRegister,TDSspace] () {$\operatorname{TDS}^\dag_\beta$};
	\node[fit = (fTDS5) (sTDS5),operatorDoubleRegister,TDSspace] () {$\operatorname{TDS}_\beta$};
	\node[fit = (fREF2) (aREF2),operatorDoubleRegister,REFspace] () {$\operatorname{REF}$};
	\node[fit = (fREF4) (aREF4),operatorDoubleRegister,REFspace] () {$\operatorname{REF}$};
	\end{pgfonlayer}
	\end{tikzpicture}
\end{tabular}
		\caption{	\label{fig:HAM-T_K} Quantum circuit representation of (top, left) $\operatorname{HAM}_L$ in~\cref{eq:HAM-T_k}; (top, right) $\operatorname{DYS}_K$ in~\cref{eq:DYS-K}; (bottom, left) a single step of time-evolution by the truncated Dyson series algorithm from~\cref{eq:LCU_TDS_duplicated} before oblivious amplitude amplification; (bottom, right) a single step of time-evolution by the truncated Dyson series algorithm with duplicated ancilla registers . Note that when $\beta=2$, a single-round of oblivious amplitude amplification is used.}
	\end{figure*}

	\begin{align}
	\ket{s_k}_{\vec{d}}:=\sqrt{\frac{k!(M-k)!}{M!}}\left(\sum_{0\le m_1< m_2<\cdots <m_k<M}\ket{\vec{m}}_{d_1\cdots d_k}\right) \ket{0}_{d_{k+1}\cdots d_K}.
	\end{align}
	This state is easy to prepare when $k=1$ -- there, it is simply a uniform superposition over $M$ number states, and costs $\mathcal{O}(\log{M})$ gates. Otherwise, naive methods based on rejection sampling have some success probability $|\gamma_k|^2$ that decreases exponentially with large $k$. Let $\operatorname{PREP}_{K}$ be one such unitary that prepares $\ket{s_k}_{\vec{d}}$ on measurement outcome $\ket{00}_c$.
	\begin{align}
	\label{eq:DYS-state-prep}
	\operatorname{PREP}_{K}\ket{k}_{\vec{b}}\ket{0}_{c\vec{d}}&:=\ket{k}_{\vec{b}}\left(
	\gamma_k\ket{00}_c\ket{s_k}_{\vec{d}}+\sqrt{1-|\gamma_k|^2}\ket{01}_c\cdots\right).
	\end{align}
	For each order $k$, the Riemann sum $B_k$ may be implemented by $\operatorname{DYS}_{K}
	:=(\operatorname{PREP}^\dag_{K}\otimes \openone_{as})\cdot \operatorname{HAM-T}_{K}\cdot (\operatorname{PREP}_{K}\otimes \openone_{as})$, as depicted in~\cref{fig:HAM-T_K}. The unitary $\operatorname{DYS}_{K}$ encodes precisely terms $B_k$ of the Dyson series as follows
	\begin{align}
	\label{eq:DYS-K}
	&(\bra{0}_{\vec{a}c\vec{d}}\otimes \openone_{\vec{b}s}) \operatorname{DYS}_{K} (\ket{0}_{\vec{a}c\vec{d}}\otimes \openone_{\vec{b}s}) 
	= \sum^K_{k=0}\ket{k}\bra{k}_{\vec{b}}\otimes \frac{|\gamma_k|^2 k!(M-k)!}{M!}B_k.
	\end{align}
	
	Now, a linear combination of Dyson series terms is implemented by preparing a state with the appropriate amplitudes in the basis $\ket{k}_{\vec{b}}$. The required state preparation operators are
	\begin{align}
	\label{eq:TDS_components}
	\operatorname{COEF}\ket{0}_{\vec{b}}&:=\frac{1}{\sqrt{\beta}}\sum^K_{k=0}\sqrt{\frac{M!(-it)^k}{M^k|\gamma_k|^2 k! (M-k)!}}\ket{k}_{\vec{b}},
	\quad
	\beta=\sum^K_{k=0}\frac{M!t^k}{M^k|\gamma_k|^2  k!(M-k)!},
	\\\nonumber
	\operatorname{COEF}'\ket{0}_{\vec{b}}&:=\frac{1}{\sqrt{\beta}}\sum^K_{k=0}\sqrt{\frac{M!t^k}{M^k|\gamma_k|^2  k!(M-k)!}}\ket{k}_{\vec{b}},
	\end{align}
	and may be implemented using $\mathcal{O}(K)$ primitive gates. Up to a proportionality factor $\beta$, we obtain the desired linear combination for simulating time-evolution.
	\begin{align}
	\label{eq:LCU_TDS_duplicated}
	\operatorname{TDS}_\beta&:= (\operatorname{COEF}'^\dag\otimes \openone_{\vec{a}c\vec{d}s})\cdot \operatorname{DYS}_{K}\cdot (\operatorname{COEF}\otimes \openone_{\vec{a}c\vec{d}s})
	\\\nonumber
	(\bra{0}_{\vec{a}\vec{b}c\vec{d}}\otimes \openone_s) \operatorname{TDS}_\beta (\ket{0}_{\vec{a}\vec{b}c\vec{d}}\otimes \openone_s)
	&=\frac{\sum^K_{k=0}(-it)^k B_k}{M^k\beta}\approx \frac{\mathcal{T} e^{-i\int_0^t H(s) \mathrm{d} s} }{\beta}.
	\end{align}
	Using the provided parameters for $K$ and $M$, the numerator, by~\cref{Thm:Truncated_Dyson_Algorithm}, approximates a unitary operation to error $\mathcal{O}(\epsilon)$. The probability of applying this operation can be boosted from  $|\frac{1+\Theta(\epsilon)}{\beta}|^2$ to $1-\mathcal{O}(\epsilon)$ using oblivious amplitude amplification~\cite{Berry2014Exponential}. If we choose $t$ to be sufficiently small such that $\beta=2$, then a single round of oblivious amplitude amplification suffices, and we obtain a single time-step of the truncated Dyson series algorithm  $\operatorname{TDS}$ in~\cref{fig:HAM-T_K}, where each reflection $\operatorname{REF}$ acts on $K(n_a+n_d+1)+2$ qubits and therefore costs $K(n_a+n_d+1)+2$ gates. All that remains is to find an implementation of $\operatorname{PREP}_K$ that prepares $\ket{s_k}_{\vec{d}}$ with an amplitude that $|\gamma_k|$ that is sufficiently large so that $t =\Theta(1)$.

	The state $\ket{s_k}_{\vec{b}c\vec{d}}$ can be prepared in a number of ways. The most straightforward approach creates a uniform superposition of states over the dimension-$k$ hypercube using $n_d \times k$ Hadamard gates $\operatorname{HAD}$, then uses $k$ reversible adders to flag states $\ket{\vec{m}}_{d_1\cdots d_k}$ with the correct ordering. This circuit $\operatorname{PREP}_{\ket{s_k}}$ produces $\ket{s_k}_{\vec{d}}$ with amplitude $\gamma_k=\sqrt{\frac{M!}{M^k k!(M-k)!}}$.  $\operatorname{PREP}_K$ is then obtained by controlling $\operatorname{PREP}_{\ket{s_k}}$ on input state $\ket{k}_{\vec{b}}$. Thus
	\begin{align}
	\beta=\sum^K_{k=0}t^k\le \sum^\infty_{k=0}t^k = \frac{1}{1-t}.
	\end{align}
	Thus by choosing $t=\Theta(1)\approx 1/2$, we obtain the desired $\beta=2$. Notably, even though the success probability of naive state preparation $|\gamma_k|^2$ decays rapidly, this only amounts to a constant factor slowdown compared to more sophisticated techniques that effectively prepare $\ket{s_k}_{\vec{d}}$ with success probability $\approx 1$. For example, rather than rejection sampling, one may perform a reversible sort on on uniform superposition of states $\frac{1}{\sqrt{M^k}}\sum_{\vec m}\ket{\vec{m}}_{d_1\cdots d_k}\ket{0}_{\text{garbage}}\rightarrow \frac{1}{\sqrt{M^k}}\sum_{\vec{m}}\ket{\mathcal{T}[m_{d_1}\cdots m_{d_k}]}_{d_1\cdots d_k}\ket{\vec{m}}_{\text{garbage}}$, such as with the quantum bitonic sorting network~\cite{babbush2017sorting}. This effectively increases $\gamma^2_k$ by a factor of $k!$, and uses significantly more ancilla qubits, but ultimately allows us to implement time steps $t\approx \ln 2\approx 0.693$ larger by a constant factor.
	
	If we wish to simulate time-evolution for any $t\le \frac{1}{2}$, this will lead to $\beta\le 2$. In this situation, $\beta$ may be increased to $2$ using single-qubit rotations, as described by~\cite{berry2015simulating}, to artificially worsen the success probability.
\end{proof}

\section{Truncated Taylor series algorithm}
\label{sec:truncated_taylor_algorithm}
The truncated Taylor series simulation algorithm was a major advance in quantum simulation for its conceptual simplicity and computational efficiency.  The original algorithm~\cite{berry2015simulating} is motivated by truncating the Taylor expansion of the time-evolution operator at degree $K$.
\begin{align}
e^{-iHt}&=1-iHt+\frac{(-iHt)^2}{2!}+\frac{(-iHt)^3}{3!}\cdots = \underbrace{\sum^{K}_{k=0}\frac{(-iHt)^k}{k!}}_{\bar{R}_K} + \underbrace{\sum^{\infty}_{k=K+1}\frac{(-iHt)^k}{k!}}_{R_K}.
\end{align}
Assuming that $t>0$ and that the truncation order $K\ge 2\|H\|t$, the norms of $\bar{R}_K$ and the remainder term $R_K$ are bounded by 
\begin{align}
\|\bar{R}_K\|&=\|e^{-iHt}-R_K\| \le 1+\|R_K\|,
\\ \nonumber
\|R_K\| 
&\le \sum^{\infty}_{k=K+1}\frac{(\|H\|t)^k}{k!}
\le \frac{(\|H\|t)^{K+1}}{(K+1)!}\sum^{\infty}_{k=K+2}\left(1/2\right)^{k-K-1}
=\frac{2(\|H\|t)^{K+1}}{(K+1)!}.
\end{align}
Thus any unitary quantum circuit $\operatorname{TTS}$ that acts jointly on registers $a,b,s$ and applies the non-unitary operator $(\bra{00}_{ab}\otimes \openone_s) \operatorname{TTS} (\ket{00}_{ab}\otimes \openone_s)\approx \bar{R}_K$ approximates the time-evolution operator with error $\delta$ and failure probability $p$ given by
\begin{align}
\label{eq:TTS_error}
\delta
&=\left\|e^{-iHt}-\bar{R}_K\right\|
= \left\|R_K\right\|
\le \frac{2(\|H\|t)^{K+1}}{(K+1)!},
\\\nonumber
p&
\le 1-\min_{\ket{\psi}_s}\left|\frac{\bar{R}_K\ket{\psi}_s}{1+\|R_K\|}\right|^2
=1-\min_{\ket{\psi}_s}\left|\frac{(e^{-iHt}-R_K)\ket{\psi}_s}{1+\|R_K\|}\right|^2
\le 1-\left|\frac{1-\|R_K\|}{1+\|R_K\|}\right|^2
= 4\|R_K\|= 4\delta.
\end{align}
Solving~\cref{eq:TTS_error} for $\|H\|t\in{\mathcal{O}}(1)$ gives the required truncation order $K\in{\mathcal{O}}\left(\frac{\log{(1/\delta)}}{\log\log{(1/\delta)}}\right)$. 

The simulation algorithm $\operatorname{TTS}$ in~\cref{Fig:TTS_algorithm} is obtained by constructing two oracles. $\operatorname{HAM}_{K}$, which applies positive integer powers of $(-iH)^k$ up to $k=K$, and $\operatorname{COEF}$, which prepares a quantum state that selects these terms with the right coefficients. $\operatorname{HAM}_{K}$ will require additional ancilla registers, which we index with $\vec{a}$ and $\vec{b}$. Note that the gate and space complexity in the truncated Taylor series algorithm is dominated by that of $\operatorname{HAM}_K$.
\begin{align}
\label{eq:TTS_components}
(\bra{0}_{\vec{a}}\otimes \openone_{s\vec{b}}) \operatorname{HAM}_{K} (\ket{0}_{\vec{a}}\otimes \openone_{s\vec{b}}) &:= \sum^K_{k=0}\ket{k}\bra{k}_{\vec{b}}\otimes (-i H)^k,
\\\nonumber
\operatorname{COEF}\ket{0}_{\vec{b}}&:=\frac{1}{\sqrt{\beta}}\sum^K_{k=0}\sqrt{\frac{t^k}{k!}}\ket{k}_{\vec{b}},
\quad
\beta=\sum^K_{k=0}\frac{t^k}{k!}\le e^{t}.
\end{align}

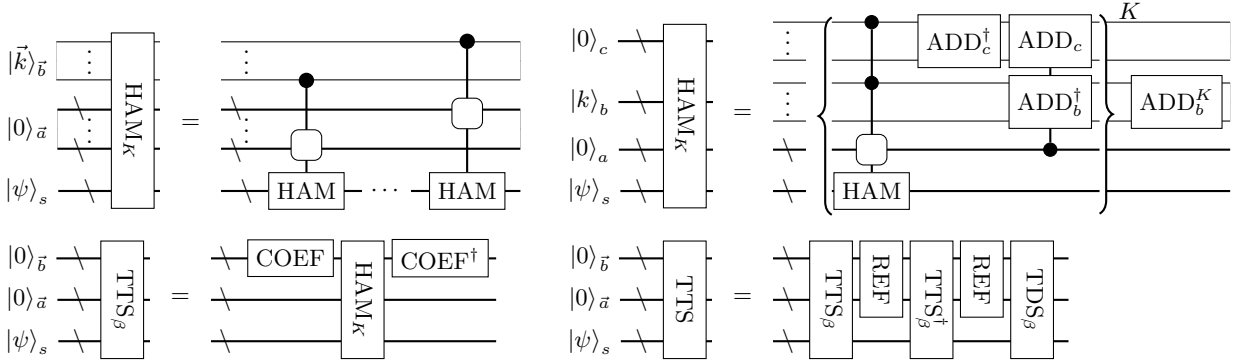
\begin{figure*}[h]
	\tikzstyle{HAMTKspace} = [text width = 10pt]
\tikzstyle{DLup} = [above=-0pt]
\tikzstyle{DLdown} = [below=-0pt]
\tikzstyle{DLBSdup} = [DLup,xshift=-4pt]
\tikzstyle{DLBSddown} = [DLdown,xshift=-4pt]

\begin{tabular}{ll}
\begin{tikzpicture}[thick]
	\matrix[row sep=0.0cm, column sep=0.1cm] (circuit) {
	\node (d) {$|\vec{k}\rangle_{\vec{b}}$}; &\node[doubleLineVdots] (d0dots) {$\vdots$}; & \node[HAMTKspace,yshift = 4pt] (dV) {}; & \node[minimum size = 10pt] (dStart) {};&\node[minimum size = 7pt] () {};&\node[doubleLineVdots] (d0dots) {$\vdots$};&
	\node[ctrlDoubleLineDown] (HAMd-1) {};&
	\node[text width=20pt] () {};&
	\node[ctrlDoubleLine] (HAMd-2) {};&
	&	\coordinate (endd);\\
	
	\node (a) {\ketb{0}{}$_{\vec{a}}$}; & \node[doubleLineVdots] (a0dots) {$\vdots$};\node[DLBSdup] {\textbackslash};\node[DLBSddown] {\textbackslash}; &&&&\node[doubleLineVdots] () {$\vdots$};\node[DLBSdup] {\textbackslash};\node[DLBSddown] {\textbackslash};&
	\node[joint, DLdown] () {};&&
	\node[joint, DLup] () {};&
	&	\coordinate (enda);\\
	
	\node (s) {\ketb{\psi}{s}}; & \node {\textbackslash}; & \node (sV) {}; &\node (sStart) {};&&\node {\textbackslash}; &
	\node[operator] (HAMs-1) {$\operatorname{HAM}$};&\node[fill=white] () {$\cdots$};&
	\node[operator] (HAMs-2) {$\operatorname{HAM}$};&
	&	\coordinate (ends);\\
};
\begin{pgfonlayer}{background}
\draw[doubleLine]  (d) -- (endd); 
\draw[doubleLine, thick] (a) -- (enda) ; 
\draw[thick] (s) -- (ends); 
\draw[thick] (HAMs-1) -- (HAMd-1) (HAMs-2) -- (HAMd-2);
\node[fit = (dStart) (sStart), fill=white,inner sep=6pt, xshift=6pt]  (dots) {$=$};
\node[fit = (dV) (sV),operatorDoubleRegister,HAMTKspace] (HAMTKfit) {};
\node[rotate=-90] at (HAMTKfit) {$\operatorname{HAM}_K$};
\end{pgfonlayer}
\end{tikzpicture}
&
\begin{tikzpicture}
\matrix[row sep=0.0cm, column sep=0.1cm] (circuit) {
	\node (c) {\ketb{0}{c}}; & \node {\textbackslash}; & \node[HAMTKspace] (fV) {}; &\node[minimum size = 10pt] (cDoubleStart) {};& \node[minimum size = 7pt] () {};
	& \node[doubleLineVdots] (c0dots) {$\vdots$}; 
	& \node (leftbracketf) {};
	& \node[ctrlDoubleLine] (HAMTc-2) {}; & \node[operatorDoubleLine] () {$\operatorname{ADD}_c^\dag$};  & \node[operatorDoubleLine] (addc1) {$\operatorname{ADD}_c$}; 
	&	\coordinate (endc); &\node[minimum size = 1pt] () {};&& \coordinate (endcZ); \\
	
	\node (b) {\ketb{k}{b}}; & \node {\textbackslash}; &&\node (bDoubleStart) {};&
	& \node[doubleLineVdots] (b0dots) {$\vdots$}; 
	&& \node[ctrlDoubleLine] (b0) {};&&  \node[operatorDoubleLine] (addb1) {$\operatorname{ADD}_b^\dag$};  &
	&&  \node[operatorDoubleLine] (addb1) {$\operatorname{ADD}_b^K$};
	&	\coordinate (endb);\\
	
	\node (a) {\ketb{0}{a}}; & \node {\textbackslash}; &&&
	&\node {\textbackslash}; 
	&& \node[joint] {};&&\node[ctrl] (addCTRLa) {} ;&
	&&&	\coordinate (enda);\\
	
	\node (s) {\ketb{\psi}{s}}; & \node {\textbackslash}; & \node (sV) {};
	&\node (sStart) {};&
	&\node {\textbackslash}; 
	&\node (leftbrackets) {};& \node[operator] (U2) {$\operatorname{HAM}$}; &	& 
	&	\coordinate (ends); &&& \coordinate (endsZ);\\
};
\begin{pgfonlayer}{background}
\draw[doubleLine,thin] (cDoubleStart) -- (endcZ) (bDoubleStart) -- (endb);
\draw[thick] (a) -- (enda) (c) -- (cDoubleStart) (b) -- (bDoubleStart) (s) -- (endsZ); 
\draw[thick] (HAMTc-2) -- (U2) (addCTRLa) -- (addc1);
\draw[thick] (fV) -- (sV);
\node[fit = (cDoubleStart) (sStart), fill=white,inner sep=6pt, xshift=6pt]  (dots) {$=$};
\node[fit = (fV) (sV),operatorDoubleRegister,HAMTKspace] (HAMKfit) {};
\node[rotate=-90] at (HAMKfit) {$\operatorname{HAM}_K$};

\node[fill=white,inner sep = 1pt] () [fit = (leftbracketf) (leftbrackets)] {};
\draw[decorate,decoration={brace, amplitude = 5pt,mirror},thick, fill=white] ($(leftbracketf.north west)+(0.3cm,+0.2cm)$) to ($(leftbrackets.south west)+(0.3cm,-0.2cm)$);
\draw[decorate,decoration={brace, amplitude = 5pt},thick, fill=white] ($(endc.north west)+(0.0cm,+0.35cm)$) to ($(ends.south west)+(0.0cm,-0.3cm)$);

\node at ($(endc.north west)+(0.4cm,+0.4cm)$) {$K$};
\end{pgfonlayer}
\end{tikzpicture}
\\
	\begin{tikzpicture}[thick]
	
	\matrix[row sep=0.0cm, column sep=0.1cm] (circuit) {
		\node (f) {\ketb{0}{}$_{\vec{b}}$}; &\node {\textbackslash}; & \node[HAMTKspace] (fV) {}; & \node[minimum size = 10pt] (fStart) {};&\node[minimum size = 7pt] () {};&\node {\textbackslash};&\node[operator] () {$\operatorname{COEF}$};
		&\node[HAMTKspace] (fTDS2) {};&\node[operator] () {$\operatorname{COEF}^\dag$};
		&	\coordinate (endf);\\
	
		\node (a) {\ketb{0}{}$_{\vec{a}}$}; & \node {\textbackslash}; &&&&\node {\textbackslash};&
		&&
		&	\coordinate (enda);\\
		
		\node (s) {\ketb{\psi}{s}}; & \node {\textbackslash}; & \node (sV) {}; &\node (sStart) {};&&\node {\textbackslash};  &
		&\node (sTDS2) {};&
		&	\coordinate (ends);\\
	};
	\begin{pgfonlayer}{background}
	\draw[thick] (a) -- (enda) (s) -- (ends) (f) -- (endf); 
	\draw[thick] (fV) -- (sV);
	\node[fit = (fStart) (sStart), fill=white,inner sep=6pt, xshift=6pt]  (dots) {$=$};
	\node[fit = (fV) (sV),operatorDoubleRegister,HAMTKspace] (TTSfit) {};
	\node[fit = (fTDS2) (sTDS2),operatorDoubleRegister,HAMTKspace] (HAMTKfit) {};
	\node[rotate=-90] at (TTSfit) {$\operatorname{TTS}_\beta$};
	\node[rotate=-90] at (HAMTKfit) {$\operatorname{HAM}_K$};
	\end{pgfonlayer}
	\end{tikzpicture}
	&
	\begin{tikzpicture}[thick]
	
	\matrix[row sep=0.0cm, column sep=0.1cm] (circuit) {
		\node (f) {\ketb{0}{}$_{\vec{b}}$}; &\node {\textbackslash}; & \node[HAMTKspace] (fV) {}; & \node[minimum size = 10pt] (fStart) {};&\node[minimum size = 7pt] () {};&\node {\textbackslash};&
		\node[HAMTKspace] (fTDS1) {};&
		\node[HAMTKspace] (fREF2) {};&
		\node[HAMTKspace] (fTDS3) {};&
		\node[HAMTKspace] (fREF4) {};&
		\node[HAMTKspace] (fTDS5) {};&
		&	\coordinate (endf);\\
		
		\node (a) {\ketb{0}{}$_{\vec{a}}$}; & \node {\textbackslash}; &&&&\node {\textbackslash};&
		\node (aTDS1) {};&
		\node (aREF2) {};&
		\node (aTDS3) {};&
		\node (aREF4) {};&
		\node (aTDS5) {};&
		&	\coordinate (enda);\\
		
		\node (s) {\ketb{\psi}{s}}; & \node {\textbackslash}; & \node (sV) {}; &\node (sStart) {};&&\node {\textbackslash}; &
		\node (sTDS1) {};&
		\node (sREF2) {};&
		\node (sTDS3) {};&
		\node (sREF4) {};&
		\node (sTDS5) {};&
		&	\coordinate (ends);\\
	};
	\begin{pgfonlayer}{background}
	\draw[thick] (a) -- (enda) (s) -- (ends)  (f) -- (endf); 
	\draw[thick] (fV) -- (sV);
	\node[fit = (fStart) (sStart), fill=white,inner sep=6pt, xshift=6pt]  (dots) {$=$};
	\node[fit = (fV) (sV),operatorDoubleRegister,HAMTKspace] (TDSfit0) {};
	\node[fit = (fTDS1) (sTDS1),operatorDoubleRegister,HAMTKspace] (TDSfit1) {};
	\node[fit = (fTDS3) (sTDS3),operatorDoubleRegister,HAMTKspace] (TDSfit3) {};
	\node[fit = (fTDS5) (sTDS5),operatorDoubleRegister,HAMTKspace] (TDSfit5) {};
	\node[fit = (fREF2) (aREF2),operatorDoubleRegister,HAMTKspace] (REF2) {};
	\node[fit = (fREF4) (aREF4),operatorDoubleRegister,HAMTKspace] (REF4) {};
	\node[rotate=-90] at (TDSfit0) {$\operatorname{TTS}$};
	\node[rotate=-90] at (TDSfit1) {$\operatorname{TTS}_\beta$};
	\node[rotate=-90] at (TDSfit3) {$\operatorname{TTS}^\dag_\beta$};
	\node[rotate=-90] at (TDSfit5) {$\operatorname{TDS}_\beta$};
	\node[rotate=-90] at (REF2) {$\operatorname{REF}$};
	\node[rotate=-90] at (REF4) {$\operatorname{REF}$};
		
	\end{pgfonlayer}
	\end{tikzpicture}
\end{tabular}
	\caption{\label{Fig:TTS_algorithm}Quantum circuit representation of (top, left) an example implementation of $\operatorname{HAM}_K$ from~\cref{eq:TTS_components} using $K$ queries to controlled-$\operatorname{HAM}$; (top, right) an example implementation of $\operatorname{HAM}_K$ with fewer ancilla qubits using the compression gadget of~\cref{Thm:compression_gadget} ;(bottom, left) a single step of the truncated Taylor series algorithm before oblivious amplitude amplification;  (bottom, right) a single step of time-evolution by the truncated Taylor series algorithm from~\cref{eq:TTS}. Note that $\beta=2$ as a single-round of oblivious amplitude amplification is used.}
\end{figure*}

The original algorithm~\cite{berry2015simulating} implements $\operatorname{HAM}_{K}$ using $K$ queries to controlled-$\operatorname{HAM}$
\begin{align}
\operatorname{C-HAM}:= \ket{1}\bra{1}_b\otimes \openone_{as} + \ket{0}\bra{0}_b\otimes (-i\operatorname{HAM})
\end{align}
with $K$ copies of registers $a$ and $b$. The state $\ket{k}_{\vec{b}}=\ket{0}^{\otimes k}\ket{1}^{\otimes K-k}$ that selects desired powers of $H$ is encoded in unary, and so $\operatorname{COEF}$ may be implemented using $\mathcal{O}(K)$ primitive gates. 	Up to a proportionality factor $\beta$, the unitaries of~\cref{eq:TTS_components} allow us to implement the desired linear combination $\bar{R}_K$ for simulating time-evolution.
\begin{align}
\label{eq:LCU_taylor}
\operatorname{TTS}_\beta&:= (\operatorname{COEF}^\dag\otimes \openone_{\vec{a}s}) \operatorname{HAM}_{K} (\operatorname{COEF}\otimes \openone_{\vec{a}s})
\\\nonumber
(\bra{0}_{\vec{a}b}\otimes \openone_s) \operatorname{TTS}_\beta (\ket{0}_{\vec{a}b}\otimes \openone_s)
&=\frac{\bar{R}_K}{\beta}\approx \frac{e^{-iHt}}{\beta} .
\end{align}
As $\bar{R}_k$ is close to unitary, the success probability $\approx 1/\beta^2$ may be boosted using oblivious amplitude amplification~\cite{Berry2014Exponential}. When $\beta=2$, a single round of oblivious amplitude amplification suffices to boost the success probability to $1-\mathcal{O}(\delta)$. Thus we chose $\ln 2\le t \in{\mathcal{O}}(1)$ such that $\beta=2$. If we desire $|t|< \ln 2$, $\beta$ may be decreased by appending a single-qubit ancilla and noting that $|\bra{0}e^{i\theta X}\ket{0}|=|\cos{\theta}|\le 1$. Thus simulation is accomplished with the circuit
\begin{align}
\label{eq:TTS}
\operatorname{TTS} &= \operatorname{TTS}_{\beta=2}\cdot(\operatorname{REF}\otimes \openone_s)\cdot\operatorname{TTS}^\dag_{\beta=2}\cdot(\operatorname{REF}\otimes \openone_s)\cdot\operatorname{TTS}_{\beta=2},
\\\nonumber
\operatorname{REF}&=\openone_{\vec{a}\vec{b}}-2\ket{0}\bra{0}_{\vec{a}\vec{b}}.
\end{align}
This approximates time-evolution by $e^{-iHt}$ with error 
$
\left\|(\bra{0}_{\vec{a}b}\otimes \openone_s)\operatorname{TTS}(\ket{0}_{\vec{a}b}\otimes \openone_s)-e^{-iHt}\right\|\in{\mathcal{O}}(\delta).
$
In order to simulate evolution $e^{-iHT}$ by longer times $T> t$, we apply $\operatorname{TTS}^{T/ t}$ -- here $t=\Theta(1)$ is chosen such that $T/t$ is an integer. The overall error
\begin{align}
\epsilon=\left\|[\operatorname{TTS}^{T/t}-\openone_{\vec{a}\vec{b}}\otimes e^{-iHT}](\ket{0}_{\vec{a}\vec{b}}\otimes \openone_s)\right\|\in{\mathcal{O}}(T\delta),
\end{align}
and success probability $1-\mathcal{O}(\epsilon)$ may thus be controlled by choosing the error of each segment to be $\delta\in{\mathcal{O}}\left(\frac{\epsilon}{T}\right)$. This requires a truncation order of $K\in{\mathcal{O}}\left(\frac{\log{(\alpha T/\epsilon)}}{\log\log{(\alpha T/\epsilon)}}\right)$. We may drop the implicit assumption that $\|H\|\le 1$, by rescaling $H\rightarrow H/\alpha$, for some normalization constant $\alpha\ge \|H\|$. Thus simulation of $e^{-iHt}$ requires $\mathcal{O}\left(\alpha T\frac{\log{(\alpha T/\epsilon)}}{\log\log{(\alpha T/\epsilon)}}\right)$ queries to $\operatorname{C-HAM}$. Note that the gate cost of all queries to $\operatorname{COEF}$ at $\mathcal{O}\left(\alpha T\frac{\log{(\alpha T/\epsilon)}}{\log\log{(\alpha T/\epsilon)}}\right)$ and that of $\operatorname{REF}$ at $\mathcal{O}\left( n_a\alpha T  \frac{\log{(\alpha T/\epsilon)}}{\log\log{(\alpha T/\epsilon)}}\right)$, is typically dominated by the gate cost of all applications of $\operatorname{C-HAM}$.

The ancilla overhead of the truncated Taylor series algorithm, at $n_s+\mathcal{O}(n_a\frac{\log{(1/\epsilon)}}{\log\log{(1/\epsilon)}})$ qubits, may be significantly improved by choosing the sequence of unitaries in the compression gadget~\cref{Thm:compression_gadget} of~\cref{sec:compresson_gadget} to be $U_j=-i\operatorname{HAM}$. This straightforwardly furnishes the following result.
\begin{corollary}[Hamiltonian simulation by a compressed truncated Taylor series]
	\label{Thm:Compressed_TTS}
	Let a time-independent Hamiltonian $H$ be encoded in standard-form with normalization $\alpha$ and $n_s + n_a$ qubits, as per~\cref{eq:standard-form-TI}. Then the truncated Taylor series algorithm approximates the time-evolution operator $e^{-iHt}$ for any $|\alpha t|\le \ln{2}$ to error $\epsilon$ using
	\begin{enumerate}
		\item Queries to $\operatorname{HAM}$: $\mathcal{O}\left(\frac{\log{(1/\epsilon)}}{\log\log{(1/\epsilon)}}\right)$.
		\item Qubits: $n_s+\mathcal{O}(n_a+\log\log{(1/\epsilon)})$.
		\item Primitive gates: $\mathcal{O}\left((n_a + \log\log{(1/\epsilon)})\frac{\log{(1/\epsilon)}}{\log\log{(1/\epsilon)}}\right)$.
	\end{enumerate}
\end{corollary}
For longer-time simulations $e^{-iHT}$ of duration $T>t$,~\cref{Thm:Compressed_TTS} is applied $\alpha T/\ln{(2)}$ times, each with error $\mathcal{O}(\frac{\epsilon}{\alpha T})$. This leads a query complexity $\mathcal{O}(\alpha T \frac{\log{(\alpha T/\epsilon)}}{\log\log{(\alpha T/\epsilon)}})$. Though the compressed algorithm is still worse than the quantum signal processing approach, which uses $n_s+\mathcal{O}(n_a)$ qubits, the technique is applicable to simulating time-dependent Hamiltonians, as demonstrated in~\cref{sec:truncated_dyson_series}.

\end{document}